%% file: main.tex
\documentclass[runningheads]{llncs}
\usepackage[T1]{fontenc}
\usepackage{newtxtext}
\usepackage[varvw]{newtxmath}
\input{packages}
\begin{document}
\title{QScale: Probabilistic Chained Consensus for Moderate-Scale Systems}

\author{
Hasan Heydari\inst{1} \and
Alysson Bessani\inst{1} \and
Kartik Nayak\inst{2}
}

\authorrunning{H. Heydari et al.}

\institute{
LASIGE, Faculdade de Ciências, Universidade de Lisboa, Portugal \\
\email{\{hheydari,anbessani\}@ciencias.ulisboa.pt}
\and
Duke University, USA \\
\email{kartik@cs.duke.edu}
}

\maketitle




\input{abstract}
\input{intro}

\section{Preliminaries}\label{sec:preliminaries}

\subsection{System Model}
We consider a distributed system composed of a fixed set $\Pi =\{p_1,\dots,p_n\}$ of $n$ processes (also called servers or replicas) in a partially synchronous communication model~\cite {dwork_1988}.
We also consider a non-empty set of external clients, disjoint from $\Pi$.
In this model, the network and processes may operate asynchronously until some unknown global stabilization time (GST). 
After GST, the system becomes synchronous, with a known time bound~$\Delta$ for communication.

Our protocol description considers that processes have access to loosely synchronized clocks with bounded drift and skew to coordinate epochs.
This assumption is well-aligned with several production blockchain protocols, such as 
Ethereum 2.0~\cite{eth_clock}, 
Algorand~\cite{algorand},
Solana~\cite{solana},
Cardano~\cite{cardano}, and
Tezos~\cite{tezos}, to cite some.
This assumption allows us to focus on the protocol's core ideas while abstracting away epoch synchronization issues~\cite{makingByzConsLive}.

Processes are subject to Byzantine failures~\cite{lamport_1982}.
We assume that at most $f=\epsilon \, n$ processes within~$\Pi$ are faulty, where $\epsilon \in [0,1/3)$.
A process that is not faulty is said to be \textit{correct}.
We assume a \emph{static corruption adversary} chooses the set of faulty processes at the beginning of execution, and such a set does not change throughout the execution.
Byzantine processes may collude and coordinate their actions.
Since external clients participate only in applying the finalization rule, we consider only correct clients, as there is little point in considering finalization rules from malicious clients.


Processes communicate via message passing over reliable point-to-point channels and, when needed, can sign messages with digital signatures. 
We assume that key distribution occurs before the system starts. 
At run-time, the private key of a correct process never leaves the process and, therefore, remains unknown to faulty processes. 
We assume the digital signature scheme supports multi-signatures~\cite{boneh2003aggregate}; for instance, BLS~\cite{boneh2001short} or ECDSA~\cite{johnson2001elliptic}.
Formally, we assume the following set of functions to be available to each process~$i\in \Pi$:
\begin{itemize}[leftmargin=1.1em,topsep=1pt, itemsep=1pt]
    \item $\texttt{sign}_i(m) \mapsto \mathit{sig}$ -- sign a message $m$ with process's own secret key;
    \item $\texttt{aggregate}(\mathit{sig}_1, \ldots, \mathit{sig}_k) \mapsto \mathit{sig}$ -- aggregate several signatures into one multi-signature;
    \item $\texttt{validate}(m, \mathit{sig}, [id_1, \ldots, id_k]) \mapsto \textit{bool}$ -- check that $\mathit{sig}$ is an aggregation of signatures of message $m$ by validators with ids $id_1, \ldots, id_k$.
\end{itemize}

We assume that processes have access to a globally known \emph{verifiable random function} (VRF)~\cite{algorand,goldberg22verifiable}, which enables the generation of verifiable random values using two operations:
\begin{itemize}[leftmargin=1.1em,topsep=1pt, itemsep=1pt]
    \item $\texttt{vrf\_prove}_i(s) \mapsto S, P_i$ -- given a seed~$s$, process~$i$ computes a pseudo-random string~$S \in \{0,1\}^\lambda$ of fixed length~$\lambda$.
    Along with $S$, it returns a proof~$P_i$ that certifies that $S$ was generated correctly by process~$i$ using its VRF key.
    \item $\texttt{vrf\_verify}(s, S, P_i) \mapsto \textit{bool}$ -- given a seed~$s$, a string~$S$, and a proof~$P_i$, this function checks whether $S$ is a valid output of $\texttt{vrf\_prove}_i$ on input~$s$.
\end{itemize}



We assume a \emph{computationally bounded adversary}, i.e., Byzantine processes have a polynomial advantage in computational power over the correct processes.
Accordingly, a Byzantine process cannot forge signatures of correct processes except with negligible probability.
We also assume each process has access to a local, unbiased, independent source of randomness.
We require a collision-resistant cryptographic hash function $\fn{H}$, which maps an arbitrary-length input to a fixed-length output~\cite{rogaway_2004}.


\subsection{Distributed Ledger}
\label{sec:definition}

\myparagraph{Blocks, certificates, and chains.}
A \emph{block} consists of four main elements:
a set of transactions $\txs$,
the process that creates/proposes the block, 
the epoch $\epoch$ in which the block was created, and
the hash $\parentHash$ of the parent block that this block extends.
There is a distinguished genesis block $\genesis$ known by all processes.
A \emph{quorum certificate} (QC) for a block~$B$ is a set of valid votes produced for $B$ by a (probabilistic) quorum of processes, and a block is said to be \emph{certified} if there exists a valid QC for that block.
A \emph{chain} is a sequence of blocks $\chain = B_0,B_1,\dots,B_\ell$ such that $B_0 = \genesis$ and, for every $i \in \{1,\dots,\ell\}$, $B_i.\parentHash = \hash{B_{i-1}}$ and $B_{i-1}.\epoch < B_i.\epoch$.
A chain is \emph{certified} if all its blocks are certified.
Thus, by repeatedly following the block referenced by each parent hash, one eventually reaches the genesis block.
The height of a block $B$, denoted by $B.\height$, is its distance from $\genesis$.
Formally, $\genesis.\height = 0$, and for every non-genesis block $B \in \chain$, if $B.\parentHash = \hash{B'}$ for some block $B'\in \chain$, then $B.\height = B'.\height + 1$.
The epochs of successive blocks in a chain are not necessarily consecutive.

\myparagraph{Clients and observed chains.}
External clients interact with the system by submitting transactions and querying processes for certified chains.
Correct processes eventually create blocks containing the submitted transactions.
Upon receiving a query, a correct process returns the longest certified chain it knows.

When a client observes that its transaction $\tx$ is included in a block $B$ of a returned chain $\textsf{ch}$, the client decides whether to consider $\tx$ finalized or to wait for additional blocks extending the chain.
For making this decision, we assume the client applies a deterministic function $\fn{safety}(\textsf{ch}, B)$ that computes the probability that another chain not containing $B$ eventually becomes longer than $\textsf{ch}$ and is adopted by correct processes.
The client considers block~$B$ finalized once $\fn{safety}(\textsf{ch}, B) \le \varepsilon_\tx$, where $\varepsilon_\tx$ is the maximum rollback probability tolerated by the client for $\tx$.

\myparagraph{Probabilistic distributed ledger abstraction.}
We assume external clients use the same $\fn{safety}$ function.
For every client transaction $\tx$, let $\varepsilon_{\tx} \in (0,1)$ denote a parameter chosen by the client.
A client considers $\tx$ finalized only after observing a certified chain~$\chain$ containing a block $B$ such that (1) $\tx \in B.\txs$, and (2) the blocks extending $B$ satisfy the client-side finalization predicate parameterized by $\varepsilon_{\tx}$.
Let $\chain_{\le h}$ denote the prefix of a chain~$\chain$ up to and including height~$h$.
For any transaction $\tx$ submitted by a client, 
\begin{itemize}
    \item \textbf{Safety:}
    If a process returns a certified chain $\chain$ that contains a block~$B$ and the client considers $\tx \in B.\txs$ finalized, then, for every certified chain $\chain'$ with $|\chain'| \ge B.\height$ returned by a process, $\chain_{\le B.\height} \preceq \chain'$ with probability $1-\epsilon_\tx$.
    \item \textbf{Liveness:}
    If $\tx$ is provided as input to all correct processes, then, with probability~$1$, eventually there exists a certified chain $\chain$ returned by a correct process such that $\chain$ contains $\tx$ in a block that the client finalizes.
\end{itemize}

Note that in many blockchain protocols, safety is defined in terms of the \textit{committed} ledger output by \textit{correct} processes (e.g.,~\cite{streamlet2020}).
However, our definition is expressed from the perspective of a correct client interacting with the ledger.
This choice allows each client to determine when to finalize a transaction~$\tx$ based on the certified chain it observes and its chosen error bound~$\epsilon_\tx$.

\section{\protocolName}
\label{sec:pro:hotstuff}

This section presents \protocolName, starting with an overview of the protocol and its key mechanisms, then presenting its specification, and finally presenting its complexity and security analysis.

\begin{figure*}[!t]
    \centering
    \includegraphics[scale=1.1]{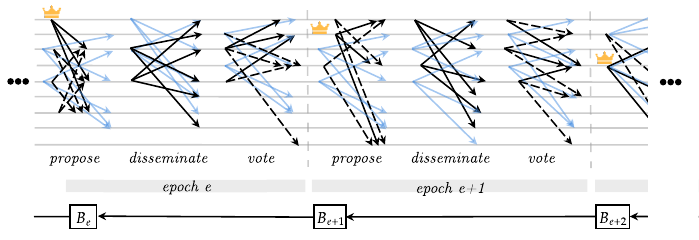}
    \caption{Execution of the \protocolName over two epochs.
    \textbf{Black} arrows illustrate the \textsc{propose}, \textsc{disseminate}, and \msgVote messages, with dashed lines representing indirect certification messages, and \textbf{\color{cyan} blue} arrows illustrate messages of the propagation sub-protocol.}
    \label{fig:exec}
\end{figure*}

\subsection{Overview and Key Mechanisms}

\protocolName aims to build a unique ledger of blocks containing client-issued transactions, as defined in Section~\ref{sec:definition}.
Thus, it defines how blocks are proposed, disseminated, voted on, and certified.
A separate procedure, the \textit{client finalization}, is executed by any external client and specifies how a client interacts with the processes, enabling the client to decide whether to consider a block containing its transaction finalized or to continue waiting for additional confirmations.

\myparagraph{The main protocol.}
\protocolName runs in a sequence of epochs, each with a designated leader known to all processes in advance.
Each epoch consists of three phases: \emph{propose}, \emph{disseminate}, and \emph{vote}.
In parallel with the main (sub-)protocol, processes execute an additional propagation sub-protocol.
Fig.~\ref{fig:exec} depicts an execution of two epochs of \protocolName, with each sub-protocol shown in a different color.

Periodically, at every $3\Delta$ time units, a new epoch starts.
At the beginning of each epoch~$e$, its designated leader sends a proposal block~$B$ that includes the current epoch number, a set of pending transactions, and the hash of its parent block, which is the certified block with the highest height it observed in previous epochs. 
To reduce the leader's communication overhead, this block is broadcast in two steps.
First, the leader sends the block to a random sample~$\sampleLOne{e}$ with the expected size of $O(\sqrt{n})$.
Next, if a process~$i \in \sampleLOne{e}$ receives block~$B$ from the leader of epoch~$e$ and has formed a quorum certificate via \textit{indirect certification} (described below) for the parent of~$B$, it starts the dissemination phase by sending $B$ along with the quorum certificate of $B$'s parent to a new random sample $\sampleLTwo{e}{i}$ also with the expected size of $O(\sqrt{n})$.
Each of these random samples is defined by independently including each process in the sample with probability $\prSample = O(1/\sqrt{n})$.

By including the parent block's hash in each block, the blocks are cryptographically linked into a chain.
As in classical BFT consensus protocols such as PBFT and HotStuff, each process maintains a local lock variable $\lockedBlock$, initially set to the genesis block $B_g$.
The variable $\lockedBlock$ stores the highest block on which the process is currently locked.
A proposal block $B$ is said to be consistent with the process's lock if either $B$ extends $\lockedBlock$, the height of $B$'s parent is greater than the height of $\lockedBlock$, or $B$'s parent has the same height as $\lockedBlock$ but was proposed in a later epoch.

During epoch~$e$, if a process receives the proposal block~$B$ together with a quorum certificate for $B$'s parent, then it becomes a \emph{candidate} to vote for $B$ only if it has not yet voted in this epoch, the quorum certificate is valid and certifies $B$'s parent, all predecessors of $B$ are certified, and $B$ is consistent with the process's lock.
If these conditions hold, the process starts the vote phase and updates $\lockedBlock$ to $B$'s parent.
In the vote phase, each candidate independently and probabilistically decides whether to send a \msgVote message to the leader of epoch~$e+1$.
As we later show, the number of correct processes that become candidates for a proposal directly impacts the probability of ensuring safety.
If the processes that vote are determined by an ordinary random function, the adversary (i.e., Byzantine processes together with adversarial scheduling) could selectively disseminate the proposal only to those processes, thereby reducing the number of correct candidates who receive it.
To mitigate this issue, correct processes use a VRF to determine whether they vote only after receiving the proposal. 
Consequently, the proposal must be disseminated to a larger set of correct processes before the adversary can identify the actual voters.

Given a random sample of size~$s$ drawn from a set of at least~$s$ processes, a \emph{probabilistic quorum} is any subset of the sample of size~$q$, where $q \le s$ (e.g., $q=2s/3$).
The leader of epoch~$e+1$ makes $B$ the parent of its proposed block if, during epoch~$e$, it receives matching \msgVote messages for~$B$ from a probabilistic quorum of voters.

\myparagraph{Indirect certification.}
As mentioned before, to create a QC for block~$B$ proposed in epoch~$e$, we initially considered a design in which the leader of epoch~$e+1$ forms the QC upon receiving \msgVote messages from a probabilistic quorum (i.e., from $q$ members of a random sample of candidates for block~$B$) and then disseminates it together with its proposal block.
However, this design increases the protocol's communication complexity to $O(n\sqrt{n})$, as each \msgVote message carries a VRF proof, and thereby the resulting QC contains $q=O(\sqrt{n})$ VRF proofs.

To preserve $O(n)$ communication complexity per epoch, we adopt an \textit{indirect certification} approach.
The initial idea is to have voters send their \msgVote messages to a random sample of processes, in addition to the leader of epoch~$e+1$.
A process in this random sample that receives \msgVote messages from a probabilistic quorum of voters sends a signature share to the next dissemination sample~$\sampleLOne{e+1}$.
If a member of $\sampleLOne{e+1}$ receives signature shares from a probabilistic quorum, it aggregates them to create a QC.
Thus, the QC contains only an aggregated signature and does not need to include the VRF proofs carried by the original \msgVote messages.
Unfortunately, we observed that the resulting protocol had a very bad probabilistic safety guarantee, as the QC is produced by a small probabilistic quorum of voters.
To improve this, we refine our indirect certification strategy by making a QC using $\alpha \ge 1$ probabilistic quorums of voters, where each probabilistic quorum is a subset of an independently selected random sample of candidates.

Using this mechanism, during the vote phase of epoch~$e$, in addition to having a random sample of candidates that send their \msgVote messages to the leader of epoch~$e+1$, there are $\alpha = O(1)$ other random samples $\mathcal{S}$ selected from the candidates.
Each candidate is selected for a sample $S \in \mathcal{S}$ with probability~$\prSample$ (a candidate can be selected for multiple samples).
For each sample $S \in \mathcal{S}$, every member of $S$ sends a \msgVote message to all members of $\sampleLOne{e+1}$. 
If a process $i \in \sampleLOne{e+1}$ receives a probabilistic quorum of \msgVote messages from every $S \in \mathcal{S}$ (i.e., the proposal was certified by $\alpha$ probabilistic quorums), it sends a signature share to all members of $\sampleLOne{e+1}$.
When a process $i \in \sampleLOne{e+1}$ receives at least $q$ signature shares sent by members of $\sampleLOne{e+1}$, it aggregates them to create a \textit{multi-quorum certificate} for block $B$ (for simplicity, we still call this a quorum certificate).
These messages are represented as dashed lines in Fig.~\ref{fig:exec}.
This procedure reduces the concrete probability of a safety violation exponentially in~$\alpha$.

\myparagraph{The propagation sub-protocol.}
Due to the probabilistic nature of the block dissemination in \protocolName, even when all processes are correct, there is a chance that some processes do not receive the block proposed in an epoch.
We address this issue by running a simple gossip protocol in parallel with the main protocol.
This protocol works as follows: periodically, every $\Delta$ time units, each process selects a random sample of processes by including each one with probability~$\prProp = O(1/n)$, and forwards its most recent certified block (i.e., the one with the highest epoch number) together with its QC to the sample.
In the best case, each proposal block and QC are propagated through three phases -- the dissemination and vote phases of an epoch, followed by the proposal phase of the next epoch -- until a new block is proposed and a new certificate is generated.

\myparagraph{Client finalization.}
Clients can submit transactions and wait for their finalization based on the observed progress on the distributed ledger.
Specifically, a client sends a transaction $\tx$ to some process(es), obtains a chain $\textsf{ch}$ from some process(es), identifies a block $B \in \textsf{ch}$ containing $\tx$, and uses a function $\fn{safety}(\textsf{ch}, B)$ to compute the probability of ensuring safety for $B$ based on what comes after $B$ on $\textsf{ch}$.
If the computed probability $p \ge \varepsilon_\tx$, the client finalizes~$\tx$.
Otherwise, if no block containing $\tx$ is observed or if $p < \varepsilon_\tx$, the client waits for additional certified blocks to be generated and repeats the procedure.

For computing the safety probability of a block $B$, we consider the worst-case scenario based on the observed chain $\textsf{ch}$.
We present a complete specification of how this probability can be calculated in Appendix~\ref{app:non-consecutive}, and present an overview of the approach here.

The analysis accounts for two main sources of risk.
First, Byzantine leaders may equivocate by proposing multiple conflicting blocks in the same epoch. 
Consequently, another hidden conflicting chain may exist containing a block at the same height as~$B$.
Second, the blocks extending~$B$ may contain gaps, meaning that successive blocks are not necessarily proposed in consecutive epochs.
In such cases, even if all processes are correct, the adversarial scheduler may manipulate message delivery to allow another conflicting chain to grow during epochs when no block from the observed chain is proposed.

For the first type of risk, we show the probability that such equivocation remains undetected decreases as more blocks are added after~$B$ in $\textsf{ch}$.
However, the second type of risk is more subtle, as it depends on the specific pattern of epochs in which blocks are proposed in $\textsf{ch}$ and those in which no block is proposed.
In particular, long sequences of blocks proposed in consecutive epochs strengthen safety guarantees, whereas larger gaps between blocks create more opportunities for hidden conflicts.
In summary, based on the observed chain, we consider all possible cases and identify the strategy that maximizes the probability of creating a hidden conflicting chain.

\input{algorithms/alg_qscale_process_side}

\subsection{Protocol Specification}

Algorithm~\ref{alg:qscale} presents the main protocol logic of \protocolName. 
The data structures and auxiliary functions employed in the protocol are specified in Appendix~\ref{sec:aux:func}. 
Each block is represented as $(e,\txs,h)_\ell$, where $e$ is the epoch in which the block was proposed, $\txs$ is the set of pending transactions included in the block, $h$ is the hash of the parent block, and $\ell$ denotes the creator process, with $(\cdot)_\ell$ indicating that the block is signed by~$\ell$.

We assume there is a deterministic function~$\fn{leader}$ that returns the leader of each epoch~$e \ge 1$.
For practical purposes, leaders are assumed to remain in their role for multiple consecutive epochs. 
The duration of a leadership period can be chosen to exceed the typical values of~$\kappa$ required to achieve a high safety probability, thereby avoiding the exponential latency that would arise under independently randomized leader selection.

Additionally, a function $\fn{valid}$ is available to verify transaction validity.
The \fn{sample} function takes as input a probability~$p$ and a seed~$s$, and returns a random sample~$S \subseteq \Pi$. 
It relies on a helper function \fn{bool} that outputs a boolean value given a probability and a seed. 
The sample~$S$ is constructed by independently including each process~$j \in \Pi$ if $\fn{bool}(p, s||j)$ returns~$\var{true}$.

\myparagraph{Local variables (line~\ref{line:vars}).}
The algorithm defines the following local variables maintained by each process:
the current epoch number (\curEpoch); a lock block (\lock); the next parent block hash (\parentHash); a set of certified blocks (\certblocks); and sets of received proposal blocks (\textit{blocks}), votes (\textit{votes}), and quorum certificates (\qcs).

\myparagraph{Epoch advancement (lines~\ref{line:new:epoch}-\ref{line:set:cur:e:r}).}
The protocol proceeds in a sequence of epochs.
Specifically, at every $3\Delta$ time units, the $\fn{new\_epoch}$ function is triggered to start an epoch. 


\myparagraph{Propose phase (lines~\ref{line:propose:start}-\ref{line:propose:end}).}
At the beginning of epoch~$e$, the leader creates a block~$B$ whose parent is the certified block of greatest height known to it (possibly the block proposed in epoch~$e-1$).
The leader then sends $B$ to the sample~$\sampleLOne{e}$ and to the leader of epoch~$e+1$.


\myparagraph{Disseminate phase (lines~\ref{line:disseminate:start}-\ref{line:disseminate:end}).}
When a process $i \in S_e$ receives a valid proposal~$B$, it stores and forwards it together with the QC of $B$'s parent to a random sample $S_{e||i}$.


\myparagraph{Vote phase (lines~\ref{line:vote:phase:start}-\ref{line:vote:diss:2:end}).}
When a process receives a valid block~$B$ and a valid quorum certificate for its parent, and the $\fn{is\_candidate}$ predicate is \emph{true}, it becomes a candidate to vote for~$B$.
Then, the process determines whether it belongs to one of the $\alpha+1$ random samples that can send \msgVote messages.
Members of the first random sample send their \msgVote messages to the leader of epoch~$e+1$, while members of the remaining samples send their votes to the processes in~$\sampleLOne{e+1}$.


\myparagraph{Processing vote message (lines~\ref{line:processing:votes:start}-\ref{line:processing:votes:end}).}
A process in $\sampleLOne{e+1}$ collects votes for the block $B$ proposed in epoch~$e$.
Once it receives $q$ valid votes from each of $\alpha$ random samples, it sends a \msgCertify message for $B$ to all members of $\sampleLOne{e+1}$.


\myparagraph{Block certification (lines~\ref{line:block:certification:start}-\ref{line:block:certification:end}).}
A member of $\sampleLOne{e+1}$ forms a quorum certificate for a block once it receives $q$ \msgCertify messages for that block from members of $\sampleLOne{e+1}$.
It aggregates the corresponding signature shares into a certificate and records the block as certified.


\myparagraph{Propagation.}
Periodically, every $\Delta$, each process calls the $\fn{propagate}$ task to propagate the blocks and quorum certificates further.
Specifically, it selects a random sample of processes by including each one independently with probability $\prProp$, and then forwards the most recent proposal block and quorum certificate it knows (i.e., the ones with the greatest heights) to that random sample.

\myparagraph{Transaction submission and finalization.}
Clients submit transactions and wait for their finalization by using the procedure $\fn{submit}(\tx,\varepsilon_\tx)$, described in Algorithm~\ref{alg:client:side:simple}.
The client submits a transaction $\tx$ and continues observing the ledger until it sees a block $B$ containing $ \ tx$ in the certified chain, and this block is safe with an error probability $\varepsilon_\tx$, which is considered acceptable for this transaction.

\input{algorithms/alg_client_consecutive_epochs}

In Algorithm~\ref{alg:client:side:simple}, we consider a simple implementation of the $\fn{safety}$ function that checks whether there are $\kappa \ge 2$ blocks, starting from $B$, with consecutive epoch numbers.
This version is useful for reasoning about the protocol's safety in simpler terms (see Section~\ref{sec:evaluation}).
In Appendix~\ref{app:non-consecutive}, we extend this function to the general case where certified blocks following $B$ are not proposed in consecutive epochs (see Algorithm~\ref{alg:client}).

\subsection{Complexity Analysis}

We now analyze the round, message, and communication complexity of \protocolName for block generation in the best-case scenario.
We assume that $\prSample = O(1/\sqrt{n})$, $\prProp = O(1/n)$, and $\alpha = O(1)$.
With these parameters, in each epoch~$e$, the leader of $e$, the members of the random sample~$\sampleLOne{e}$, and the voters send an expected $O(\sqrt{n})$ messages; besides, each remaining process sends $O(1)$ messages in expectation due to the propagation sub-protocol.
Observe that, during an epoch, the expected number of messages sent is: 
\begin{align*}
    \underbrace{n\prSample + \overbrace{(n\prSample)^2}^{\text{indirect cert.}}}_{\text{propose}}
    + \underbrace{(n\prSample)^2}_{\text{disseminate}}
    + \underbrace{n\prSample + \overbrace{\alpha(n\prSample)(n\prSample)}^{\text{indirect cert.}}}_{\text{vote}}
    + \underbrace{3 n^2\prProp}_{\text{propagation}}
    = O(n).
\end{align*}
Accordingly, the expected total message complexity per block is $O(\kappa\,n)$.
Observe that the average number of messages sent per process per epoch is $O(1)$.
Although we assume the use of multi-signatures~\cite{boneh2003aggregate}, their communication overhead can be ignored if we assume that the transaction set in each block has more than $n$ bits. 
Hence, the expected per-block communication complexity is $O(\kappa\,n)$.


The above bounds apply to finalizing a specific block after observing $\kappa$ certified blocks.
As \protocolName certifies blocks in a pipelined manner, these costs can be amortized across finalized blocks.
Consider $r \geq \kappa$ consecutive epochs, each certifying a block.
Since each epoch incurs $O(n)$ expected messages, the total expected communication over these epochs is $O(rn)$.
After the first $\kappa-1$ certified blocks, each additional certified block finalizes one earlier block, yielding $r-\kappa+1$ finalized blocks.
Thus, the amortized expected message and communication complexity per finalized block is $\frac{O(rn)}{r-\kappa+1}$, which is $O(n)$ for $r \gg \kappa$.

\subsection{Correctness}
Due to space limitations, safety and liveness proofs are deferred to Appendices~\ref{sec:app:safety} and~\ref{sec:app:liveness}.

\begin{thm*}[Safety -- Corollary~\ref{cor_psync_safety} in Appendix~\ref{sec:app:safety}]
    Suppose $q \ge (n+f)\prSample/2$, $\kappa \ge 2$, and $\prSample = O(1/\sqrt{n})$.
    Then, a safety violation for the \protocolName protocol occurs with a probability of at most
    $\exp\!\big(O(-(\kappa-1)\sqrt{n})\big)$.
\end{thm*}

\ignore{
\myparagraph{Liveness.}
We compute the probability of committing a block after $\kappa \ge 2$ epochs in the best-case scenario.
Specifically, given $\kappa+1$ consecutive correct leaders, we determine the probability that the block proposed by the first leader is committed.
To this end, we proceed as follows.
Suppose the $\ell$th leader ($1 \le \ell \le \kappa$) proposes a block $B_\ell$ that extends the previous block. 
We compute the probability that block $B_\ell$ is certified.
Note that a correct process votes for this proposal only if it has received the certified blocks $B_1,B_2,\dots,B_{\ell-1}$.
We compute the probability that a correct process receives the certified blocks $B_1,B_2,\dots,B_{\ell-1}$ to participate in voting.
Based on this, we compute the probability that the last block $B_{\kappa}$ is certified and enables correct processes to commit~$B_1$.
}

\begin{thm*}[Liveness -- Theorem \ref{thm:liveness:probability-one} in Appendix~\ref{sec:app:liveness}]
If $\tx$ is provided as input to all correct processes, then, with probability~$1$, eventually there exists a certified chain $\chain$ returned by a correct process such that $\chain$ contains $\tx$ in a block that the client finalizes.
\end{thm*}

\section{Numeric Evaluation}
\label{sec:evaluation}

\ignore{
To better illustrate the security of \protocolName, we analyze its safety violation probability and compare it with potential competitors, including Nakamoto consensus, Algorand, and non-scalable deterministic protocols with fixed committees.
Due to space limitations, the comparisons with Nakamoto consensus and non-scalable deterministic protocols are presented in Appendix~\ref{app:sec:evaluation}.}
To illustrate the security of \protocolName, we compare its safety violation probability with potential competitors, including Algorand; comparisons with Nakamoto consensus and non-scalable deterministic protocols with fixed committees are deferred to Appendix~\ref{app:sec:evaluation}.

\begin{figure*}[!t]
    \centering
    \includegraphics[scale=0.58]{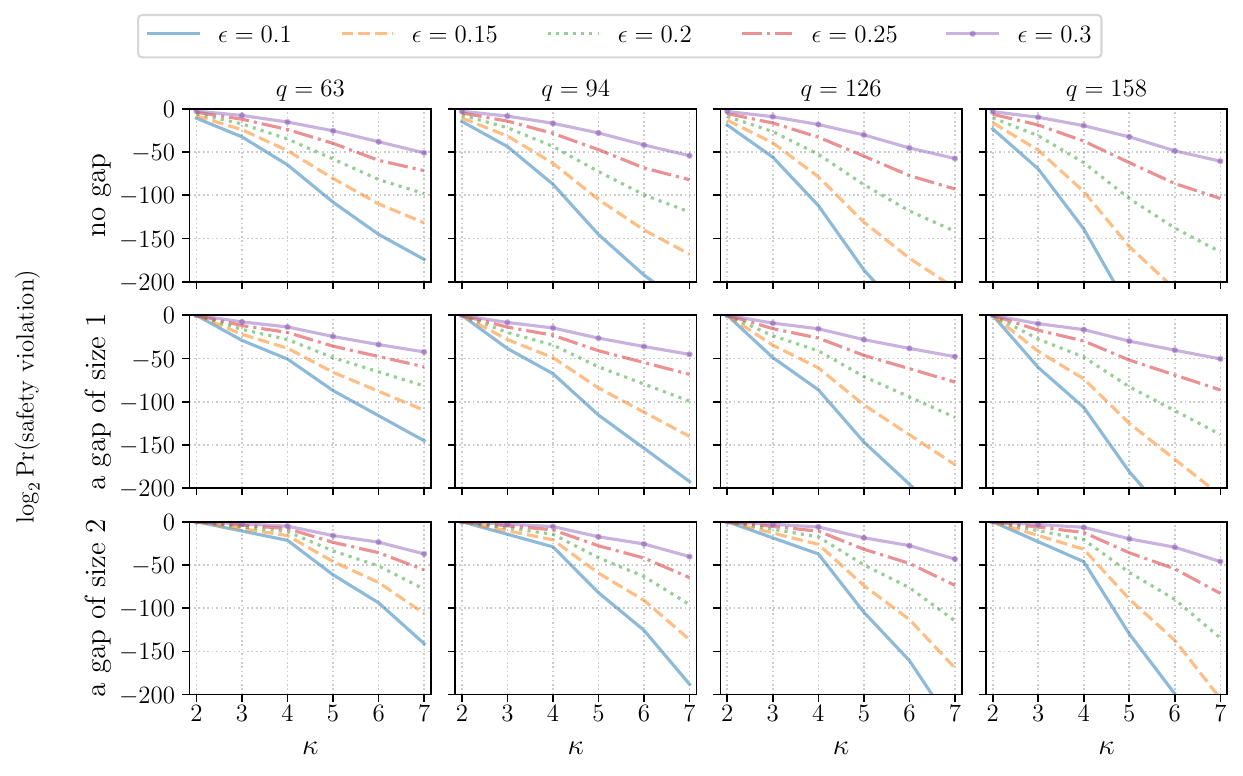}
    \caption{The probability of safety violation 
    with $n=1000$ for different $\epsilon$ values.
    The expected sample size is $1.5q$, and 
    $q \in \{2\sqrt{n},3\sqrt{n},4\sqrt{n},5\sqrt{n}\}$.}
    \label{fig:psync:1000:e}
\end{figure*}

\subsection{\protocolName Security}

We start by presenting concrete probabilities for the safety guarantee of \protocolName.
As an illustrative setup, we consider a system with $n=1000$ processes, different quorum sizes ($q \in \{2\sqrt{n}, 3\sqrt{n}, 4\sqrt{n}, 5\sqrt{n}\}$),
$\alpha = 3$, $\prSample = 1.5q/n$ (samples of size $1.5q$), and 
$\prProp = 5/n$.
Fig.~\ref{fig:psync:1000:e} shows the log-scaled probability of a safety violation, which can be interpreted as \emph{bits of security}, as a function of~$\kappa$ for five different values of the fault ratio: $f/n = \epsilon \in \{0.1, 0.15, 0.2, 0.25, 0.3\}$.
At the top of the figure, we can see curves for $\kappa$ consecutive blocks (no gaps in the observed sequence), and similar curves considering gaps of one and two epochs on the $\kappa$ certified blocks sub-chain.

There are several takeaways from this figure.
First, waiting for more confirmation blocks significantly increases the system's security.
Second, increasing the quorum and sample sizes enhances the system's security.
Third, as more Byzantine nodes are introduced into the system, the probability of a safety violation increases (i.e., we provide fewer bits of security).
This means there is a trade-off between \protocolName's security for a given threat level ($\epsilon$) and its expected finalization latency ($3\Delta\kappa$), as well as the concrete number of messages exchanged per epoch (ruled by $q$).
For example, if we consider $\epsilon = 20\%$, \protocolName provides $100$ bits of security by having $\kappa = 6$ and $q=94$, or $\kappa = 5$ and $q = 158$.

The results considering gaps (middle and lower graphs) follow the same trend, but with strictly higher safety violation probabilities across all setups.
Notice that these are worst-case probabilities, in the sense that they assume a fully adversarial scheduler.
If the network is synchronous or the scheduler is constrained~\cite{probft}, the concrete probabilities will be much better. 

\subsection{\protocolName vs Algorand}

\begin{figure}[!t]
    \centering
    \includegraphics[scale=0.55]{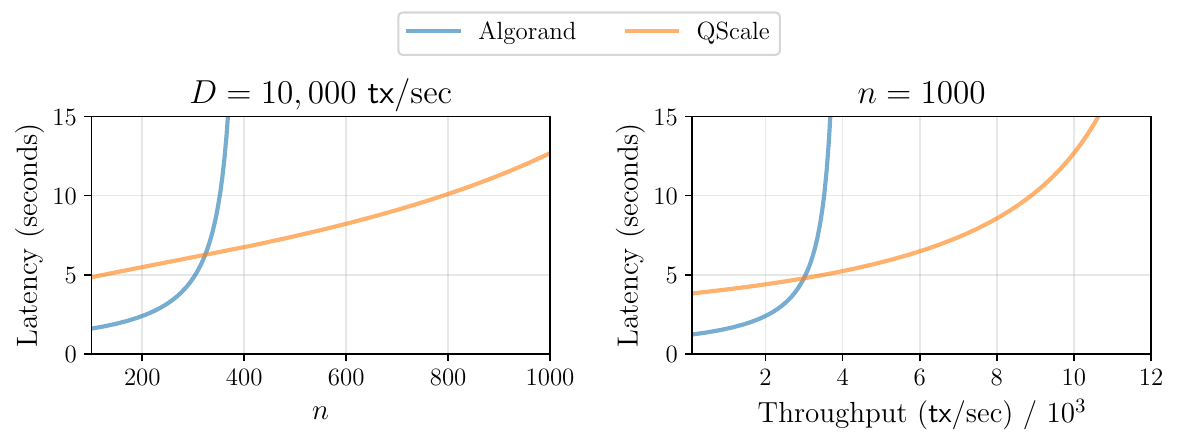}
    \caption{Latency comparison between \protocolName and Algorand in the Pipes model. The left figure varies $n$ with fixed transaction arrival $D$, while the right figure varies throughput with fixed $n=1000$.}
    \label{fig:qscale-algorand}
\end{figure}

Algorand~\cite{algorand,blum2023analyzing} is a committee-based blockchain protocol that also employs probabilistic techniques for scalability.
Therefore, it can be seen as the most direct competitor of \protocolName.
We compare \protocolName and Algorand in terms of expected latency and throughput under the best-case scenario, using the \emph{Pipes model} recently introduced by Lewis-Pye et al.~\cite{lewispye2026pipes}.
Since Algorand targets $\epsilon = 0.2$ with $\approx\!56$ bits of security, we consider a \protocolName' configuration with $q=3\sqrt{n}=94$ and $\kappa=5$.
Fig.~\ref{fig:qscale-algorand} shows the results of this analysis for a fixed workload (left) and system size (right). \kartik{specify other parameters used, e.g., latency between machines, size of lambda, M.}
The details about this analysis are presented in Appendix~\ref{app:pipes}.

The results show \protocolName is more scalable than Algorand.
This happens because Algorand employs a committee-to-all dissemination strategy, which requires a proposer to disseminate a full block to all processes; as a result, latency increases rapidly as either the number of processes or the workload (transaction arrival rate) grows.
In contrast, \protocolName disseminates blocks only to samples of expected size $O(\sqrt n)$. 
Consequently, for high-throughput deployments with hundreds to a few thousand processes, \protocolName maintains bounded latency over a wider range of transaction rates and is expected to achieve a substantially higher maximum throughput than Algorand.

\section{Conclusion}\label{sec:conclusion}

We present \protocolName for moderate-scale systems, a leader-based protocol that operates under partial synchrony, and achieves expected $O(\kappa\,\sqrt{n})$ per-process communication, $O(\kappa\,n)$ total communication, and $O(\kappa)$ best-case latency, while preserving safety and liveness with high probability.
We proved its complexity bounds and computed the probabilities of ensuring safety and liveness.
\protocolName bridges the gap between protocols optimized for small‐scale deployments -- which suffer from high communication complexity when scaled to hundreds or thousands of processes -- and committee-based protocols designed for large-scale systems; when deployed in moderate-scale settings, these committee-based protocols require committees nearly as large as the entire system to maintain safety and liveness with high probability, resulting in poor performance.

\bibliographystyle{splncs04}
\bibliography{ref}

\newpage\clearpage
\appendix
\markboth{}{}

\section{Data Structures and Auxiliary Functions of \protocolName}\label{sec:aux:func}
Algorithm~\ref{alg:aux} specifies the data structures and auxiliary functions of \protocolName.

\input{algorithms/alg_qscale_ds}

\input{evaluation}

\section{Static Committee}
\label{app:static:commitee}

In this section, we analyze an alternative design for \protocolName that uses a fixed committee to scale the system.
In further detail, consider a \emph{static} committee formed by selecting a random subset of $c$ processes from the system.
Besides, consider a variant of Algorithm~\ref{alg:qscale}, in which each leader sends its proposal to this designated committee, and assume that a block is certified if the leader receives at least $q = \lceil c\cdot o\rceil$ vote messages from the committee members, where $o \in [0, 1]$ is a fixed threshold parameter.
It is important to note that defining~$o$ is essential for ensuring liveness, as requiring the leader to collect $c$ votes allows even a single silent Byzantine committee member to block progress.

Our goal is to ensure liveness with probability at least $1-2^{-30}$ and to evaluate the corresponding probability of maintaining safety across different committee sizes.
To this end, let $\hypergeometric(N,R,s)$ denote the hypergeometric distribution obtained by sampling $s$ elements uniformly without replacement from a population of size $N$, among which $R$ elements are distinguished, and let $\hypergeometricCDF(N, R, s, k)$ denote the cumulative distribution function of the hypergeometric distribution:
\begin{align*}
    \hypergeometricCDF(N, R, s, k) = \sum_{i = \max(0, s - (N - R))}^{k} \frac{\binom{R}{i} \binom{N - R}{s - i}}{\binom{N}{s}}.    
\end{align*}
Here, $k$ is the upper bound on the number of distinguished elements in the sample.
Let $X \sim \hypergeometric(n,f,c)$ denote the number of Byzantine processes in the static committee.
If the leader is Byzantine, then Byzantine committee members may vote for two conflicting blocks, while correct committee members vote for at most one of them. 
Therefore, if the committee contains $X$ Byzantine processes, then two conflicting blocks can both collect $q$ votes only if the Byzantine votes, which may be used for both blocks, together with a split of the correct votes are sufficient to give each block at least $q$ votes. This requires $(c+X)/2 \ge q$.
Similarly, let $Y \sim \hypergeometric(n,n-f,c)$ denote the number of correct processes in the static committee.
Since Byzantine committee members may remain silent, progress is impossible if fewer than $q$ committee members are correct. 
Thus, a liveness violation occurs if $Y<q$.
We seek the optimal value of $o$ that minimizes the probability of a safety violation, subject to the constraint that liveness is ensured with probability at least~$1-2^{-30}$. 
Formally, the optimization problem is:
\begin{align*}
\text{Minimize:} \quad & \Pr((c+X)/2 \ge q) =  1-\hypergeometricCDF(n, f, c, 2q - c -1) \\
\text{Subject to:} \quad & \Pr(Y < q) = \hypergeometricCDF(n, n - f, c, q - 1) < 2^{-30} \\
\text{Where:} \quad & q=\lceil c\cdot o\rceil, o \in [0, 1].
\end{align*}

Table~\ref{tab:safety:violation:static:committee} shows the probability of safety violation for $n=500$ and $n=1000$.
As the size of committees increases, the probability of a safety violation decreases. 
What is important to note is that the committee size must be close to~$n$ to ensure a sufficiently small probability of safety violation.

\begin{table}[!t]
    \centering
    \begin{subtable}[b]{0.45\linewidth}
        \centering
        \begin{tabular}{c@{\hspace{1em}}c}
            \hline
            $c$ & $\Pr(\text{safety violation})$\\
            \hline 
            $100$ & $2^{-0.05}$ \\
            $150$ & $2^{-6.48}$ \\
            $200$ & $2^{-27.97}$ \\
            $250$ & $2^{-104.29}$ \\
            \hline
        \end{tabular}
        \caption{$n=500$ and $f=100$.}
        \label{tab:safety:violation:static:committee:500}
    \end{subtable}
    \hfill
    \begin{subtable}[b]{0.45\linewidth}
        \centering
        \begin{tabular}{c@{\hspace{1em}}c}
            \hline
            $c$ & $\Pr(\text{safety violation})$\\
            \hline
            $200$ & $2^{-9.44}$ \\
            $250$ & $2^{-25.99}$ \\
            $300$ & $2^{-9.447}$ \\
            $350$ & $2^{-49.62}$ \\
            $400$ & $2^{-86.42}$ \\
            \hline
        \end{tabular}
        \caption{$n=1000$ and $f=200$.}
        \label{tab:safety:violation:static:committee:1000}
    \end{subtable}
    \caption{Probability of safety violations when only a static committee of size $c$ is used, and liveness is ensured with probability of at least $1-2^{-30}$.}
    \label{tab:safety:violation:static:committee}
\end{table}

\input{propagation}
\input{safety}
\input{liveness}
\input{pipes}

\end{document}

%% file: packages.tex
\usepackage{graphicx}
\usepackage{float}
\usepackage{amsmath,mathtools,bm}

\usepackage{algorithmicx}
\usepackage{algorithm}
\usepackage[noend]{algcompatible}
\usepackage[inline]{enumitem}
\usepackage[dvipsnames]{xcolor}
\definecolor{myPurple}{HTML}{9013FE}
\usepackage{dashbox}
\usepackage{hyperref}
\hypersetup{colorlinks=true,urlcolor=black,linkcolor=black,citecolor=black}
\usepackage{caption}
\usepackage{subcaption}
\usepackage{multicol}
\usepackage{multirow}
\usepackage{booktabs}
\usepackage{makecell}

\usepackage{amssymb}

\newenvironment{thm}{\begin{theorem}}{\end{theorem}}
\newtheorem{thm*}{Theorem}
\newenvironment{lem}{\begin{lemma}}{\end{lemma}}
\newenvironment{cor}{\begin{corollary}}{\end{corollary}}
\newenvironment{deff}{\begin{definition}}{\end{definition}}

\newif\ifcomments
\commentstrue
\commentsfalse

\newcommand{\kartiknew}[1]{\ifcomments{\color{magenta} \footnotesize[Kartik: #1]}\fi}

\newcommand{\kartik}[1]{\ifcomments{\color{magenta} \footnotesize[Kartik: #1] }\fi}
\newcommand{\hasan}[1]{\ifcomments{\color{cyan} \footnotesize[Hasan: #1] }\fi}

\usepackage{xspace}
\newcommand{\protocolName}{QScale\xspace}
\newcommand{\parentHash}{\textit{pHash}\xspace}
\newcommand{\curEpoch}{\textit{curEpoch}\xspace}

\newcommand{\height}{\textit{height}\xspace}
\newcommand{\txs}{\textsf{txs}\xspace}
\newcommand{\tx}{\textsf{tx}\xspace}
\newcommand{\epoch}{\textit{epoch}\xspace}

\newcommand{\propagate}{\texttt{propagate}()\xspace}
\newcommand{\hash}[1]{\texttt{H}(#1)\xspace}

\newcommand{\var}[1]{\textit{#1}\xspace}
\newcommand{\fn}[1]{\texttt{#1}\xspace}
\newcommand{\EX}[1]{\mathbb{E}[#1]\xspace}

\newcommand{\prProp}{p_{\textsf{prop}}\xspace}
\newcommand{\prSample}{p_{\textsf{sample}}\xspace}

\newcommand{\leader}[1]{\texttt{leader}(#1)\xspace}
\newcommand{\lockedBlock}{B_{\mathsf{lock}}}

\newcommand{\sampleLOne}[1]{S_{#1}\xspace}
\newcommand{\sampleLTwo}[2]{S_{#1||#2}\xspace}

\newcommand{\blockFive}[4]{(#1,#2,#3)_{#4}\xspace}
\newcommand{\qcs}{\textit{qcs}\xspace}

\newcommand{\msgVote}{\textsc{vote}\xspace}
\newcommand{\msgCertify}{\textsc{certify}\xspace}

\newcommand{\genesis}{B_{g}\xspace}

\newcommand{\Bin}{\textsc{Bin}}
\newcommand{\PiC}{\Pi_C}
\newcommand{\PiF}{\Pi_F}
\newcommand{\myparagraph}[1]{\vspace{0pt}\noindent\textbf{#1}}
\newcommand{\ignore}[1]{}

\makeatletter
\renewcommand{\alglinenumber}[1]{\ifnum#1<10 0\fi#1}
\makeatother

\usepackage{wrapfig}

\newcommand{\function}{\textbf{function}\xspace}
\newcommand{\canvote}{\texttt{canVote}\xspace}
\newcommand{\certblocks}{\textit{certBlocks}\xspace}
\newcommand{\lock}{\textit{lock}\xspace}

\newcommand{\hypergeometric}{\mathrm{HG}\xspace}
\newcommand{\hypergeometricCDF}{\mathrm{HG\_CDF}\xspace}
\newcommand{\poly}{{\sf poly}}

\newcommand{\chain}{\textsf{ch}}

%% file: abstract.tex
\begin{abstract}
Existing distributed ledger protocols either incur high communication complexity and are thus suited to systems with a small number of processes (e.g., PBFT), or rely on committee-sampling-based approaches that only work for very large numbers of processes (e.g., Algorand).
Neither of these lines of work is well-suited for moderate-scale distributed ledgers ranging from a few hundred to a thousand processes, which are common in production (e.g., Redbelly, Sui, Solana).
This work aims to design an efficient distributed ledger with low per-process resource utilization for moderate-scale systems. 
To achieve this, we propose \protocolName, in which a block can be considered finalized after the certification of $\kappa-1$ consecutive blocks extending it, where obtaining such a finalization incurs, in expectation, only $O(\kappa \sqrt{n})$ per-process communication complexity, $O(n\kappa)$ total communication complexity, and a best-case latency of $O(\kappa)$ rounds, while the probability of a safety violation decreases exponentially with~$\kappa$ and liveness is guaranteed with probability $1$.
An interesting feature of \protocolName is that, even when blocks are non-consecutive, a client can decide whether to finalize the block containing its transaction based on the block sequence it observes, just as in Nakamoto consensus. 
This opens the possibility of clients waiting for different numbers of blocks to be produced, depending on the value of their transactions.

\keywords{Byzantine fault-tolerance \and Consensus \and Probabilistic protocols \and Blockchains.}
\end{abstract}

%% file: intro.tex
\section{Introduction}

Protocols for a distributed ledger allow a set of processes to agree on an unbounded sequence (i.e., a chain) of blocks, each containing multiple transactions~\cite{nakamoto_2008}. 
Security for a distributed ledger tolerating a fraction of malicious processes is expressed in terms of \emph{safety} and \emph{liveness}.
Safety requires that all honest processes agree on any blocks they output.
Liveness guarantees that if all honest processes hold some transaction as input, then that transaction will eventually be included in some block output by those processes.
Solving this problem requires solving the well-known distributed consensus problem~\cite{lamport_1982}.

Scaling distributed ledger (or consensus) protocols efficiently to a large number of processes is a fundamental problem in modern distributed computing.
Two key metrics are used to measure this scalability: latency and communication complexity. 
Latency refers to the number of rounds required to commit/finalize a transaction in the ledger of honest processes.
Total (resp.\ per-process) communication complexity refers to the number of bits sent by honest processes (resp.\ each honest process) to commit a transaction to the ledger.

\myparagraph{State-of-the-art approaches to scaling.} 
Several high-level approaches have been developed to improve these metrics.
The first one was popularized by PBFT~\cite{pbft} and adopted and improved by several other works, such as Tendermint~\cite{tendermint-zarko}, Simplex~\cite{chan2023simplex,sing_song_simple}, ICC~\cite{camenisch2022internet}, and the HotStuff family of protocols~\cite{yin19hotstuff,hotstuff2,abraham2020sync}, among others.
At a high level, in this approach, processes operate in a sequence of all-to-all (sometimes one-to-all and all-to-one) communication rounds and rely on the intersection of quorums to achieve safety and liveness.
These protocols typically incur $O(1)$ latency in the best case, $\poly(n)$ total communication complexity, and $\Omega(n)$ per-process communication complexity.
In practice, this approach has been adopted by several blockchains~\cite{tendermint-blockchain,sui-blockchain,aptos-blockchain,dfinity-blockchain}.
However, as $n$ increases, say, to hundreds or thousands of processes, the high per-process (particularly on the leader) and total communication costs adversely affect the system's latency~\cite{yin19hotstuff,neiheiser2021kauri,bessani2014state}.

The second approach, popularized by Algorand~\cite{algorand,blum2023analyzing}, aims to efficiently scale to a large number of processes.
The key idea is to elect random \emph{committees} of size $c$ such that each committee has an honest majority of processes with overwhelming probability, where $c$ is a security parameter.
In this approach, only processes in the committee send messages to all other processes. 
This yields a latency of $O(1)$ rounds and a total communication complexity of $O(n\!\cdot\!\poly(c))$.
However, since committee members communicate with all processes, it still incurs $\Omega(n)$ per-process communication.
To prevent committee members from becoming bottlenecks, the system employs a gossip layer to disseminate messages, at the expense of increased latency.
More importantly, in practice, sampling an honest committee with overwhelming probability is useful only when $n$ is large, e.g., Algorand considers $n \ge 10^{12}$, with no more than $20\%$ of processes compromised.
In such situations, the size of the committee is typically of the order of thousands of processes~\cite{algorand}.
For a small $n$, to ensure there is an honest majority in the committee with overwhelming probability, we need $c\!\approx\!n$. 
Further, at this moderate scale, the use of the gossip layer might be difficult to justify due to its added complexity and increased latency, making the committee members a bottleneck.
Thus, when $n$ is in the range of a few hundred to a few thousand processes, this approach devolves into the first approach. 

A third approach, popularized by works such as~\cite{10.1145/3465084.3467897,10.1145/1989727.1989732,king2011load,gelles2024optimal}, takes this a step further by reducing per-process communication to a sublinear number of processes by creating several sub-committees of poly-logarithmic size. 
However, this approach is suitable only for even larger $n$'s and thus has limited practical applicability.

All of the above approaches aim to achieve safety and liveness properties with overwhelming probability.
Another approach popularized by Nakamoto consensus~\cite{nakamoto_2008} relaxes this requirement.
In this line of work, processes obtain probabilistic confirmations, with the probability of a safety violation decreasing as a given block gets deeper in the chain.
While not necessary, systems based on this approach have relied on peer-to-peer communication (gossip) to disseminate transactions.
The protocol incurs high latency, since disseminating a block requires $O(\log n)$ rounds, and we need to wait for $\kappa$ blocks to commit a transaction.
However, because gossip is used, this protocol incurs a small per-process communication overhead.
In practice, this has been shown to scale to a large number of processes with low latency ($O(\kappa \log n)$ rounds and security of $2^{-\kappa}$).
Moreover, the protocol assumes that the gossip layer allows all honest processes to communicate with each other (which may not be true if an honest process is connected to only Byzantine processes).

The final approach leverages probabilistic techniques to disseminate the leader's proposal, e.g., via expander-graph overlays~\cite{yandamuri2023communication}, or to cast votes more efficiently~\cite{probft}. 
While expander-graph–based dissemination can theoretically reduce per-process communication complexity to sublinear levels, it incurs additional multi-hop latency proportional to the graph's diameter, typically $O(\log n)$.
In contrast, the vote-casting process still requires the leader to send or collect messages involving a linear number of processes, which limits scalability in practice.

\myparagraph{Scaling to moderate-sized systems.}
Our state-of-the-art analysis shows there are efficient solutions tailored to settings with a small (e.g., HotStuff) and large (e.g., Algorand) values of $n$.
However, these approaches may not be ideal for moderate-scale systems, e.g., 200-2000 nodes; this is the key problem we address in this work.
We note that several prominent blockchains are currently deployed on this scale; for example, Sui with between $120$ to $333$ validators~\cite{sui_size}, Redbelly at $300$ nodes~\cite{redbelly_size}, Stellar at $186$ nodes~\cite{stellar_size}, XRP with over $150$ validators~\cite{xrp_size}, and Solana with around $800$ validators~\cite{solana_size}.

This motivates the problem we study: implementing a \textit{probabilistic distributed ledger} for moderate-scale systems. 
Informally, a probabilistic distributed ledger guarantees the following:
(1) with a probability determined by the protocol's finalization rule, in the finalized ledger, there can be no difference between any two correct processes' ledgers (safety), and (2) any value received by a correct process will eventually appear in the finalized ledger with probability $1$ (liveness).
Asymptotically, our goal is to provide a probabilistic distributed ledger with sublinear per-process communication, sub-quadratic overall communication complexity, and low latency per finalized block.
From a practical standpoint, we aim to achieve these goals with small constants, making the approach applicable to systems ranging from hundreds to thousands of processes, under the same parameters and assumptions of blockchains currently in production.

\myparagraph{Overview of our solution.}
Our solution -- \protocolName\ -- operates under 
the partially synchronous communication model~\cite{dwork_1988}.
At a high level, \protocolName integrates ideas from chained consensus protocols~\cite{yin19hotstuff,streamlet2020} with techniques for probabilistic confirmations~\cite{probft,nakamoto_2008}.
The protocol proceeds in epochs, and no process communicates with all processes in an epoch; instead, $O(\sqrt{n})$ processes send and receive messages to and from random samples of sublinear size $s = O(\sqrt{n})$, while the remaining processes send and receive messages to and from a constant number of processes.
This leads to a linear communication complexity per epoch, ensuring scalability and enabling the protocol to operate efficiently in moderate-scale systems.
Table~\ref{tab:comparision} summarizes the asymptotic performance of our protocol and the relevant related protocols in terms of communication complexity, as well as best-case latency.

\begin{table*}[!t]
    \centering
    \begin{tabular}{l@{\hspace{0.5em}}l@{\hspace{0.5em}}l@{\hspace{0.5em}}l@{\hspace{0.5em}}l@{\hspace{0.5em}}l}
         \toprule
         \vspace{-0.2em}
         & \vtop{
         \hbox{\strut per-process}
         \hbox{\strut comm. comp.}}
         & \vtop{\hbox{\strut comm.} \hbox{\strut complexity}} 
         & latency 
         & \vtop{\hbox{\strut corruption}\hbox{\strut adversary}}
         & \vtop{\hbox{\strut comm.}\hbox{\strut model}}
         \\
        \midrule
        \vspace{-0.2em}
        \vtop{\hbox{\strut PBFT~\cite{pbft}, Simplex~\cite{chan2023simplex},}
        \hbox{\strut ICC~\cite{camenisch2022internet}, Tendermint~\cite{tendermint-zarko}}
        }
        & \raisebox{-1.3ex}{$O(n)$} 
        & \raisebox{-1.3ex}{$O(n^2)$} 
        & \raisebox{-1.3ex}{$O(1)$} & \raisebox{-1.3ex}{static}
        &\raisebox{-1.3ex}{psync.}\\

        
        \midrule
        \vspace{-0.2em}
        \vtop{\hbox{\hfill HotStuff~\cite{yin19hotstuff,hotstuff2}}}
          & $O(n)$
          & $O(n)$
          & $O(1)$ & adaptive
          & psync.
          \\

        \midrule
        \vspace{-0.2em}
        Algorand~\cite{algorand} & $O(n)$ & 
        $O(n\!\cdot\!\poly(c))$$^{\dagger}$ & $O(1)$  & adaptive & sync. \\

        \midrule
        \vspace{-0.2em}
         King and Saia~\cite{gelles2024optimal} & $\widetilde{O}(\sqrt{n})$ 
         & $\widetilde{O}(n\,\sqrt{n})$ & $O(\log{n})$ 
         & adaptive 
         & sync.\\
        \midrule

        \vspace{-0.2em}
        
        Nakamoto consensus~\cite{nakamoto_2008}
        & $\widetilde{O}(\kappa)$ 
        & $\widetilde{O}(\kappa\,n)$ & $O(\kappa \log{n})$
        & adaptive & sync.
        
        \\

        \midrule

        \vspace{-0.2em}
        
        ProBFT~\cite{probft} & $O(n)$ 
        & $O(n\,\sqrt{n})$ & $O(1)$  
        & static & psync.$^{\ddagger}$\\
         \midrule
         \textbf{This paper}
         & $O(\kappa\,\sqrt{n})$
         & $O(\kappa\,n)$ 
         & $O(\kappa)$
         & static 
         & psync.\\
         \bottomrule
    \end{tabular}
    \caption{Comparison of Byzantine consensus protocols in the \textit{common case}. 
    Here, $n$ is the number of processes, $c$ is the committee size, and $\kappa$ is the number of blocks required to commit a block. 
    $^{\dagger}$Algorand requires $c\!\sim\!n$ when $n \le 1000$ to ensure there is an honest majority in the committee with overwhelming probability.
    $^{\ddagger}$ProBFT assumes an adversarial scheduler that manipulates the delivery time of messages independently of the sender's id and whether it is faulty or not.
    \kartiknew{I am not sure of the CC values here. PBFT incurs more than quadratic during view change, and quadratic in the common case. Tendermint is also a quadratic in the common case since it does all to all, so perhaps it can be grouped with Simplex. For King and Saia, which paper is this from? for Algorand, can we put a $*$ and say $c\sim n$ is when $n$ is in blah range?} \kartiknew{also, is this comm complexity per round, or to achieve consensus? It is surprising if we are better than King/Saia asymptotically}
    \kartik{in the absence of VRF agg, is it true that we have no advantage over probft? or perhpas probft's CC also worsens in that case?}}
    \label{tab:comparision}
\end{table*}

\ignore{
The protocol proceeds in epochs where, in each epoch~$e$, a designated leader sends a proposal block~$B$ extending a parent block~$B'$ to a randomly selected sample~$S_e$ of processes.
Upon receiving a valid proposal, any process in~$S_e$ that holds a quorum certificate for~$B'$ forwards the proposal to a new random sample, thereby offloading part of the proposal dissemination to that sample.
If the proposal is valid, each process that receives it independently decides, for each process~$i$ that is either a member of~$S_{e+1}$ or the leader of epoch~$e+1$, whether to send a vote message to~$i$.
In parallel, processes propagate their most recently received proposals and quorum certificates using the propagation sub-protocol.
The leader of epoch~$e+1$ extends the block proposed in epoch~$e$ upon receiving a sufficient number of matching votes, which is sublinear in the number of processes.
Besides, if a process in~$S_{e+1}$ receives a sufficient number of matching votes, it sends a certification message to all members of~$S_{e+1}$. Upon collecting a sufficient number of such messages from members of~$S_{e+1}$, a process forms a quorum certificate for the block proposed in epoch~$e$.
Once a process observes $\kappa$ consecutive certified blocks that extend a previously committed block, it can safely commit the first of these blocks.
}

In each \protocolName epoch, a designated leader sends a proposal block that extends its longest certified chain to a random sample of processes.
If the proposal is validated, it is disseminated to other correct processes and subsequently voted upon via a probabilistic mechanism. 
If the required votes are gathered for the proposal, it becomes certified and can be the parent of a block proposed by the leader of the next epoch.

To construct such a certificate, we initially adopted the traditional approach to certification, inspired by HotStuff-like protocols, in which a random sample of $s$ voters in a given epoch sends its votes to the leader of the next epoch, which then uses them to create a quorum certificate.
However, we observed that this approach does not provide sufficiently strong safety guarantees when certificates are formed from a small number of votes.
Furthermore, it increases the protocol's communication complexity to $O(n\sqrt{n})$.
Therefore, we introduce an \emph{indirect certification} approach, in which voters send their votes to the next leader and, additionally, to a random sample of processes, using two rounds of communication to build a certificate used by the voters of the next epoch to validate the next proposed block's parent. 
Using this approach, we significantly improved the safety guarantees without increasing the protocol's communication complexity or latency.


Another distinguished feature of \protocolName is its subjective client-based finalization rule, which enables a client whose transaction appears in a block~$B$ to decide whether to consider $B$ finalized based on the certified blocks it observes after~$B$.
Since not every epoch necessarily produces a certified block, successive certified blocks in a chain may have non-consecutive epoch numbers; the observed pattern after~$B$, which shows which subsequent epochs contain certified descendants of~$B$ and which form gaps, determines the probability that~$B$ will be replaced in a fork, and the client finalizes the transaction only when this probability is sufficiently small.
Since this pattern-dependent probability is difficult to express in a compact closed form, we present our latency and communication bounds in terms of the number~$\kappa$ of certified blocks starting with~$B$ that have consecutive epoch numbers.

Notice this feature plays a role analogous to probabilistic finality in Nakamoto consensus~\cite{nakamoto_2008}: each client can decide when to consider its transaction final based on the pattern of certified blocks it observes and the safety probability it requires.
For instance, a client submitting a low-value transaction may, after observing the transaction in a certified block~$B$, finalize under a weaker pattern, such as a sequence of $\kappa$ certified blocks starting with~$B$ that has a gap between one pair of successive blocks, whereas a client submitting a high-value transaction may wait for a stronger pattern, such as a sequence of $\kappa$ certified blocks starting with~$B$ with consecutive epoch numbers.


Putting together these innovative features led us to the following result for \protocolName:

\begin{theorem}[Informal main result]
In a partially synchronous message-passing system with $n$ processes, where up to $f = \epsilon\,n$ processes may be Byzantine under a static corruption adversary with $\epsilon \in [0,\tfrac{1}{3})$, \protocolName implements a probabilistic distributed ledger by providing the following per-finalized-block guarantees:

\begin{itemize}[leftmargin=1.1em]
    \item $O(\kappa\sqrt{n})$ per-process communication complexity,
    \item $O(\sqrt n)$ amortized per-process communication complexity,
    \item $O(\kappa\,n)$ total communication complexity,
    \item $O(n)$ amortized total communication complexity.
\end{itemize}

The amortized bounds are taken over finalized blocks.
Moreover, \protocolName ensures liveness with probability $1$ and safety with probability at least $1 - e^{-O(\kappa\,\sqrt{n})}$.
\end{theorem}

\ignore{
\kartik{maybe we can kill this statement?}Regarding liveness, we show that a block can be committed with high probability after $\kappa$ consecutive epochs in the best-case scenario, where $\kappa$ consecutive leaders are honest. 
We also show that the protocol runs indefinitely; hence, this scenario eventually occurs, ensuring liveness with probability $1$.
}

\myparagraph{Paper organization.}
The remainder of the paper is organized as follows. 
Section~\ref{sec:preliminaries} introduces the preliminaries, such as the system model.
Section~\ref{sec:pro:hotstuff} describes the \protocolName protocol.
Section~\ref{sec:evaluation} provides numerical evaluations, and
Section~\ref{sec:conclusion} concludes the paper.

We also include an appendix with additional numerical simulations, as well as the full correctness proofs and the protocol analysis.

%% file: algorithms/alg_qscale_process_side.tex
\begin{algorithm}[t!]
\caption{Process-side of \protocolName \ -- process $i$.}
\label{alg:qscale}
\begin{algorithmic}[1]

\STATEx{\hspace{-1.5em}\textbf{local variables}}
\STATE{
$\curEpoch \gets 0$; 
$\lock \gets B_g$;
$\parentHash \gets \perp$; 
$\certblocks \gets \{ \genesis \}$; 
$\var{blocks}, \qcs, \var{votes} \gets \emptyset, \emptyset, \emptyset$
}\label{line:vars}

\STATEx{\hspace{-1.5em}\textbf{periodically}}
\STATE{every $3 \Delta$, trigger $\fn{new\_epoch}(\curEpoch+1)$}\label{line:new:epoch}    
\STATE{every $\Delta$, trigger $\propagate$}\label{line:set:cur:e:r}

\STATEx{\hspace{-1.5em}\textbf{upon} $\fn{new\_epoch}(e)$ \hfill ({\color{blue}``propose'' phase})}
\STATE{$\curEpoch \gets e$; \label{line:propose:start}
$\var{voted} \gets \var{false}$;
$\var{votes} \gets \emptyset$} \label{line:ph:leader:sample}
\IF{$\texttt{leader}(e) = i$}\label{line:if:leader}
    \STATE{$B \gets \fn{new\_block}(\curEpoch,\parentHash)$}\label{line:create:block}
    \STATE{$\forall j \in \fn{sample}(\prSample,e) \cup \{\leader{e+1}\}$, send $\langle \textsc{propose}, B \rangle$ to $j$}\label{line:ph:propose}
\ENDIF \label{line:propose:end}

\STATEx{\hspace{-1.5em}\textbf{upon} receiving $m = \langle \textsc{propose},B\rangle$
\hfill ({\color{blue}``disseminate'' phase})} 
\IF{$\fn{valid\_proposal}(m)$} \label{line:disseminate:start}
    \STATE{$\textit{blocks} \gets \textit{blocks} \cup \{ B \}$}\label{line:receive:propose}
    \STATE{$\var{sample} \gets \fn{sample}(\prSample,\curEpoch||i)$}\label{line:sample:diss:start}
    \STATE{let $\var{QC}$ be the quorum certificate associated with the parent of $B$}
    \STATE{$\forall j \in \var{sample}\cup \{\fn{leader}(\curEpoch+1)\}$, send $\langle \textsc{disseminate}, B, \var{QC}\rangle$ to $j$}\label{line:sample:diss:end}
\ENDIF \label{line:disseminate:end}

\STATEx{\hspace{-1.5em}\textbf{upon} receiving $\langle \textsc{disseminate},B,\var{QC}\rangle$
\hfill ({\color{blue}``vote'' phase})}
\STATE{\textbf{if} $\fn{valid\_block}(B)$ \textbf{then} $\textit{blocks} \gets \textit{blocks} \cup \{ B \}$}\label{line:check:block} \label{line:vote:phase:start}
\STATE{\textbf{if} $\fn{valid\_quorum\_cert}(\var{QC})$ \textbf{then} 
$\qcs \gets \qcs \cup \{ \var{QC} \}$; 
$\fn{try\_to\_certify}()$}\label{line:check:qc}
\IF{$\fn{is\_candidate}(B,\var{QC})$}\label{line:vote:diss:2:start}
\STATE{$\var{voted}\gets\var{true}$; $\lock \gets \text{a block } B'\in \certblocks$ s.t. $\hash{B'} = B.\parentHash$}
\FOR{each $j \in \{ 1,\dots,\alpha+1 \}$}
    \STATE{$\var{flag},S,P \gets \canvote(B,j)$; $m \gets \langle \msgVote, \curEpoch,\hash{B},j, S,P\rangle$}
    \STATE{\textbf{if} $\var{flag} = \var{true} \land j = 1$ \textbf{then} send $m$ to $\leader{\curEpoch+1}$}
    \STATE{\textbf{else if} $\var{flag} = \var{true}$ \textbf{then} $\forall k \in \fn{sample}(\prSample,\curEpoch+1),$ send $m$ to $k$}
\ENDFOR
\ENDIF\label{line:vote:diss:2:end}

\STATEx{\hspace{-1.5em}\textbf{upon} receiving $M = \{ \langle \msgVote,e,h,k,\_,P_j\rangle \mid j \in Q \}$ from $q$ processes $Q$}
\IF{$e = \curEpoch \land \bm{(} \forall m \in M, \fn{valid\_vote}(m) \bm{)}$}\label{line:processing:votes:start}
    \STATE{\textbf{if} $i = \leader{e+1}$ \textbf{then} $\parentHash \gets h$}
    \IF{$i \in \fn{sample}(\prSample,e+1)$}
        \STATE{$\var{votes} \gets \var{votes} \cup \{Q\}$}\label{line:if:msg:votes:q:statement}
        \STATE{\textbf{if} $|\var{votes}| = \alpha$ \textbf{then} $\forall j \in \fn{sample}(\prSample,e+1)$, send $\langle \textsc{certify}, (e,h)_{i} \rangle$}
    \ENDIF
\ENDIF\label{line:processing:votes:end}

\STATEx{\hspace{-1.5em}\textbf{upon} receiving $M = \{ \langle \textsc{certify}, (e,h)_{j} \rangle \mid j \in Q \}$ from $q$ processes $Q$}
\IF{$e = \curEpoch -1 \land i \in \fn{sample}(\prSample,\curEpoch)$}\label{line:block:certification:start}
    \STATE{$\var{QC} \gets \fn{aggregate}( \{ (e,h)_{j} \mid \langle \textsc{certify}, (e,h)_{j} \rangle \in M \} )$}\label{line:msg:certify}
    \STATE{$\qcs \gets \qcs \cup \{ \var{QC} \}$;
    $\fn{try\_to\_certify}()$}\label{line:qcs:update}    
\ENDIF\label{line:block:certification:end}



\newcounter{savedALGline}
\makeatletter
\setcounter{savedALGline}{\value{ALG@line}}
\makeatother

\end{algorithmic}
\end{algorithm}

%% file: algorithms/alg_client_consecutive_epochs.tex
\begin{algorithm}[t!]
\caption{Client-side of \protocolName.}
\label{alg:client:side:simple}
\begin{algorithmic}[1]

\STATEx{\hspace{-1.5em}\function $\fn{submit}(\tx,\varepsilon_\tx)$}
    \STATE{send $\tx$ to a process or a set of processes}
    \LOOP
        \STATE{$\textsf{ch} \gets$ the longest certified chain from a process or from a set of processes} 
        \STATE{let $B \in \textsf{ch}$ be the block containing $\tx$, if such a block exists}
        \IF{$B \neq \perp$}

            \STATE{\textbf{if} $\fn{safety}(\textsf{ch},B) \le \varepsilon_\tx$ \textbf{then} \textbf{return} \emph{true}}
        \ENDIF
    \ENDLOOP

\STATEx{\hspace{-1.5em}\function $\fn{safety}(\textsf{ch},B)$}
\IF{$\exists B_{1},\dots,B_{\kappa} \in \textsf{ch}$ proposed in the consecutive epochs starting in $e \ge B.epoch$}
    \STATE{\textbf{return} the probability of violating safety after $\kappa$ consecutive blocks}
\ENDIF
\STATE{\textbf{return} $1$}

\end{algorithmic}
\end{algorithm}

%% file: algorithms/alg_qscale_ds.tex
\begin{algorithm}[t!]
\caption{\protocolName (data structures and auxiliary functions) -- process~$i$.}
\label{alg:aux}
\begin{algorithmic}[1]
\makeatletter
\setcounter{ALG@line}{\value{savedALGline}}
\makeatother

\STATEx{\hspace{-1.5em}\textbf{struct} block: $\epoch$, $\txs$, $\var{hash}$}

\vspace{0.2em}
\STATEx{\hspace{-1.5em}\textbf{struct} quorum certificate: a multi-signature over $\epoch$ and $\var{hash}$}

\vspace{0.2em}
\STATEx{\hspace{-1.5em}\function $\fn{sample}(\var{prob},\var{seed})$}
    \STATE{\textbf{return} $\{j \in \Pi \mid \fn{bool}(\var{prob},\var{seed}||j)\}$}\label{line:sample}

\vspace{0.2em}
\STATEx{\hspace{-1.5em}\function $\propagate$}
    \STATE{$B \gets$ the block with the greatest height in \textit{blocks}}\label{line:vote:diss:1:start}
    \STATE{$\var{QC} \gets$ the quorum certificate with the greatest height in \qcs}
    \STATE{$\var{sample} \gets \fn{sample}(\prProp,\text{a random seed})$} 
    \STATE{$\forall j \in \textit{sample}$, send $\langle \textsc{disseminate}, B, \var{QC}\rangle$ to $j$}\label{line:vote:diss:1:end}

\vspace{0.2em}
\STATEx{\hspace{-1.5em}\function $\fn{try\_to\_certify}()$}
\FOR{each $B = \blockFive{\_}{\_}{h}{\_} \in \var{blocks}$}
    \STATE{\textbf{if} $\exists \var{QC} \in \qcs \text{ s.t. } \var{QC}.\var{hash}=h$ \textbf{then} $\certblocks \gets \certblocks \cup \{ B \}$}
\ENDFOR

\vspace{0.2em}
\STATEx{\hspace{-1.5em}\function $\fn{new\_block}(\curEpoch,\parentHash)$}
\STATE{$h \gets \parentHash$}
\STATE{\textbf{if} $h = \perp$ \textbf{then} let $h$ be the hash of the certified block with the maximum height}
\STATE{\textbf{return} $\blockFive{\curEpoch}{\fn{txs}()}{h}{i}$}

\vspace{0.2em}
\STATEx{\hspace{-1.5em}\function $\fn{valid\_block}(\blockFive{e}{\var{txs}}{h}{\ell})$}
\STATE{\textbf{return} $\texttt{valid}(\txs)
\land \texttt{leader}(e) = \ell \land
\bm{(}
\exists B' \in \var{blocks},\hash{B'} = h \bm{)} \land \exists \var{QC} \in \qcs, \var{QC}.\var{hash} = h$}

\vspace{0.2em}
\STATEx{\hspace{-1.5em}\function $\fn{valid\_proposal}(\langle \textsc{propose},B = \blockFive{e}{\_}{h}{\_}\rangle)$}
\STATE{\textbf{return} $e = \curEpoch \land 
i \in \fn{sample}(\prSample,e) \land 
\fn{valid\_block}(B)$}

\vspace{0.2em}
\STATEx{\hspace{-1.5em}\function $\fn{valid\_quorum\_cert}(\var{QC})$}
\STATE{\textbf{return} 
$\exists e,  \exists B \in \var{blocks}, \exists Q \subseteq \fn{sample}(\prSample,e), \text{ s.t. } B.\epoch=e \land |Q| = q \,\land$}
\STATEx{$\fn{validate}((e,\hash{B}),\var{QC},Q)$}

\vspace{0.2em}
\STATEx{\hspace{-1.5em}\function $\fn{lock\_consistent}(B)$}
\STATE{find block $P \in \certblocks$ s.t. $\hash{P} = B.\parentHash$}
\STATE{\textbf{return} 
$\lock = P$ $\lor$
$P.\height > \lock.\height \lor
(P.\height = \lock.\height \land P.\epoch > \lock.\epoch)$}

\vspace{0.2em}
\STATEx{\hspace{-1.5em}\function $\fn{is\_candidate}(B =\blockFive{e}{\_}{h}{\_},\var{QC})$}
\STATE{\textbf{return} $\var{voted}=\var{false} \land 
e = \curEpoch \land
\var{QC}.\var{hash} = h \land \fn{lock\_consistent}(B) \, \land $}
\STATE{
(all predecessors of $B$ are certified)}

\vspace{0.2em}
\STATEx{\hspace{-1.5em}\function $\fn{valid\_vote}(\langle \msgVote,e,h,k,S,P_j\rangle)$}
\STATE{\textbf{return} $e = \curEpoch \land 
\fn{vrf\_verify}(e||h||k,S,P_j) \land 
(\exists B \in \var{blocks}, \hash{B} = h)$}
\label{line:receive:vote}



\vspace{0.2em}
\STATEx{\hspace{-1.5em}\function $\canvote(B, j)$}
\STATE{$S,P \gets \fn{vrf\_prove}_i(\curEpoch||\hash{B}||j)$}
\STATE{$\var{flag} \gets \fn{bool}(\prSample,S)$}\label{line:can:vote}
\STATE{\textbf{return} $\var{flag}, S, P$}

\end{algorithmic}
\end{algorithm}

%% file: evaluation.tex
\section{Extended Numeric Evaluation}
\label{app:sec:evaluation}

\ignore{
To better illustrate the security of \protocolName, we analyze its safety violation probability and compare it with potential competitors, including Nakamoto consensus, Algorand, and non-scalable deterministic protocols with fixed committees.}

This appendix extends the numeric evaluation of \protocolName by presenting additional comparisons with Nakamoto consensus and non-scalable deterministic protocols with fixed committees.

\subsection{\protocolName vs Nakamoto Consensus}

\begin{figure}[t!]
    \centering
    \begin{subfigure}[t]{0.5\textwidth}
        \centering
        \includegraphics[scale=0.48]{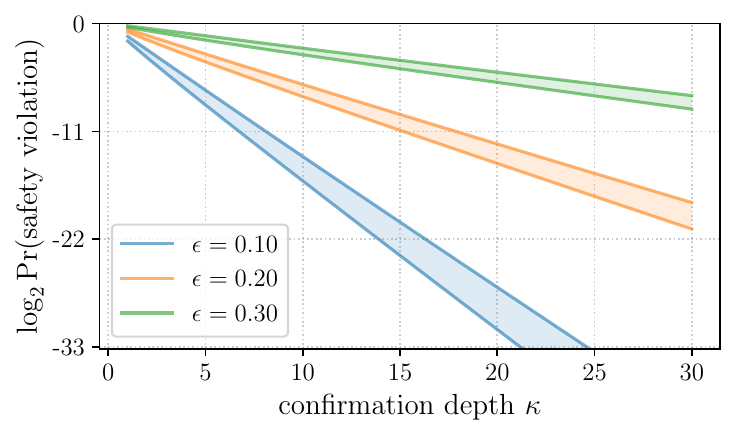}
        \caption{Safety violation in Nakamoto consensus as a function of the confirmation depth~$\kappa$.}
    \label{fig:nakamoto}
    \end{subfigure}%
    ~ 
    \begin{subfigure}[t]{0.4\textwidth}
        \centering
        \includegraphics[scale=0.48]{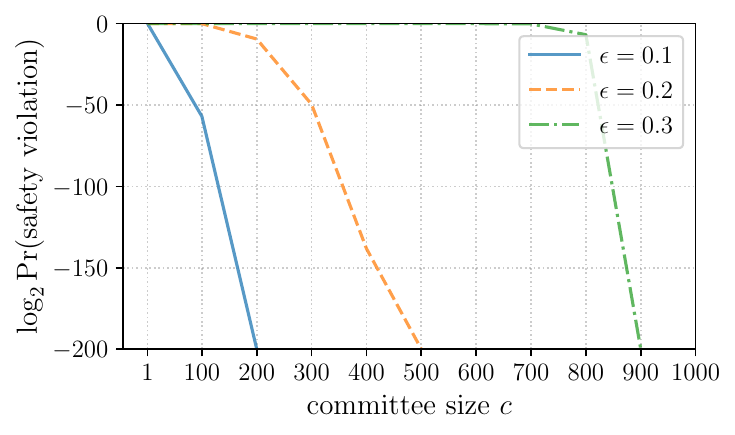}
        \caption{Safety violation in static committees of size $c$ for a system of $n = 1000$ processes.}
    \label{fig:committees}
    \end{subfigure}
    \caption{Safety violation probability of competing approaches for different failure ratios~$\epsilon = f/n$.}
\end{figure}

Fig.~\ref{fig:nakamoto} illustrates the probability of a safety violation in Nakamoto consensus~\cite{nakamoto_2008} as a function of the confirmation depth~$\kappa$ for different adversarial mining fractions ($\epsilon$), using the bounds of~\cite{guo2022bitcoin}.\footnote{The shaded region in the figure represents the gap between the lower and upper bounds from the analysis~\cite{guo2022bitcoin}; thus, the actual safety-violation probability lies within this interval.} 
Compared with \protocolName, Nakamoto consensus offers significantly worse security and higher latency, even in setups with smaller quorums and a chain gap of size~2 (e.g., the down-left graph in Fig.~\ref{fig:psync:1000:e}). 
For example, with $\kappa$ between $4$ and $6$, \protocolName already reaches negligible failure probabilities, whereas Nakamoto consensus requires tens of confirmations to achieve comparable security levels.
This leads to many hours of waiting for a confirmation, as a block is generated every few minutes in the Nakamoto consensus.
This difference stems from \protocolName's use of probabilistic quorum intersections, which provides stronger guarantees than Nakamoto consensus's probabilistic chain growth mechanism.

\subsection{\protocolName vs Static Committees}

Many proof-of-stake protocols (e.g., \cite{eth,solana}) use (non-scalable) traditional consensus protocols to run in a committee elected to process transactions for a time period.
A natural question is how the security of these static committees compares with \protocolName.

Fig.~\ref{fig:committees} shows the safety violation probability as a function of committee size for three adversarial fractions.
The figure shows that to achieve higher levels of security (e.g., $100$ bits), the committees need to have $13\%$, $36\%$, and $84\%$ to handle $10\%$, $20\%$, and $30\%$ of failures, respectively, in a system with 1000 processes.
These committee sizes, combined with the number of exchanged messages in traditional protocols and the need to contact all processes at least once per new block, make the committee-based approach less scalable than \protocolName (see Appendix~\ref{app:static:commitee} for further detail).

%% file: propagation.tex
\section{Probability and Propagation Bounds}
In this section, we first recall the Chernoff bounds used to analyze the random samples employed by \protocolName, and then study the propagation sub-protocol, which disseminates blocks and quorum certificates beyond the processes reached during the main protocol. 
The bounds and propagation results presented here are used in the following sections, mainly in the liveness analysis.

\subsection{Probability Bounds}
We use the Chernoff bounds~\cite{motwani1995randomized} for bounding the probability that the sum of independent random variables deviates significantly from its expected value.
Suppose $X_1, \dots, X_n$ are independent Bernoulli random variables, and let $X$ denote their sum.
Then, for any $\delta\in(0,1)$:
\begin{align}
&\Pr\boldsymbol{(} X \le (1-\delta)\mathbb{E}[X] \boldsymbol{)} \le \exp( -\delta^2\mathbb{E}[X]/2 ).\label{ineq:chernof2}
\end{align}
Besides, for any $\delta \ge 0$:
\begin{align}
\Pr\boldsymbol{(} X \ge (1+\delta)\mathbb{E}[X] \boldsymbol{)} \le \exp\boldsymbol{(} -\delta^2\mathbb{E}[X]/(2+\delta) \boldsymbol{)}.\label{ineq:chernof}
\end{align}

\subsection{Propagation Sub-Protocol}
\label{sec:propagation}

Due to the probabilistic nature of the block dissemination in \protocolName, even if all processes are correct, there is a chance that some processes do not receive the block proposed in an epoch.
Recall that we addressed this issue by running a simple gossip protocol, called the propagation sub-protocol, in parallel with the main protocol, to further propagate the blocks.
In this section, we present two theorems related to this sub-protocol.
Assuming that $\chi \ge 1$ processes initially have a message, the first theorem gives a lower bound on the probability that all processes receive the message after $k$ rounds, and the second theorem gives the exact probability.
We then present a concrete example to highlight the difference between the probabilities provided by these theorems.

The following theorem provides a simple closed-form lower bound on the probability that the propagation sub-protocol disseminates a message to all processes within a fixed number of rounds.

\begin{thm}\label{thm:all:receive}
Assume that
(1) $\chi \ge 1$ processes have a given message~$m$, and
(2) in each round, each process~$i$ that has or has received~$m$ sends $m$ to each process~$j\in \Pi$ with probability~$\prProp$.
Then, for any number of rounds~$k\ge 1$, 
\begin{align*}
    \Pr(\text{all processes receive~$m$ by round $k$})
   &\ge 1 - (n-\chi)\cdot\exp(- k \chi \prProp)
   \\&\ge 1 - \frac{n-\chi}{k \chi \prProp}.  
\end{align*}
\end{thm}
\begin{proof}
Let $H_r$ be the set of processes that have message~$m$ at the start of round $r\ge 1$.
The probability that process~$i \in H_r$ sends $m$ to a process $j \in \Pi$ in round~$r$ is $\prProp$.
Since members of $H_r$ behave independently, given $j$ does not have $m$ at the start of $r$ and $|H_r| = s$, the probability that $j$ does not receive~$m$ in round~$r$ is 
$(1-\prProp)^{s} \le (1-\prProp)^{\chi}$.
We have:
\begin{align*}
    & \Pr(\text{$j$ does not receive $m$ by the end of round $k$}) 
    \\&= \Pi_{r=1}^{k}\Pr(\text{$j$ does not receive $m$ in round $r$} \mid \text{$j \notin H_r$ and $|H_r|=s$})  
    \\&\le (1-\prProp)^{k\chi}.
\end{align*}
Using the union bound, we have:
\begin{align*}
    &\Pr(\text{there is a process that does not receive $m$ by the end of round $k$})
    \\&\quad\le (n-\chi)(1-\prProp)^{k\chi}
    \\&\quad\le (n-\chi)\exp(- k \chi \prProp).
\end{align*}
Therefore,
\begin{align*}
    \Pr(\text{all processes receive~$m$ by round $k$})
   &\ge 1 - (n-\chi)\exp(- k \chi \prProp)
   \\&\ge 1 - \frac{n-\chi}{k \chi \prProp}.  
\end{align*}
\end{proof}

While Theorem~\ref{thm:all:receive} provides a simple closed-form lower bound on the dissemination probability, the next theorem gives an exact expression for the same probability by modeling the number of processes that have received the message as a Markov chain.

\begin{thm}\label{thm:all:receive:exact}
    Assume that
    (1) $\chi \ge 1$ processes have a given message~$m$, and
    (2) in each round, each process~$i$ that has or has received~$m$ sends $m$ to each process~$j\in \Pi$ with probability~$\prProp$.
    Then, for any number of rounds~$k\ge 1$, 
    \begin{align*}
        \Pr(\text{all processes receive~$m$ by round $k$})
       = I_\chi\,T^{\,k}\,I_n,  
    \end{align*}
    where 
    \begin{align*}
         T_{s,t} =
         \begin{cases}
            \Pr\!\big( \mathrm{Bin}\bm{(}n-s,1 - (1-\prProp)^{s}\bm{)} = t-s \big) & t \ge s
            \\
            0 & t<s,
         \end{cases}
    \end{align*}
    $s,t\in\{1,\dots,n\}$, $I_\chi$ is a row vector with length~$n$ that has $1$ at index $\chi$ and $0$ otherwise, and $I_n$
    is a column vector with length~$n$ that has $1$ at index $n$ and $0$ otherwise.
\end{thm}
\begin{proof}
    Let $H_r \subseteq \Pi$ be the set of processes that have message~$m$ at the start of round $r\ge 1$.
    Further, let $X_r=|H_r|$, with $X_1 = \chi$.
    The probability that process~$i \in H_r$ sends $m$ to a process $j \in \Pi$ in round~$r$ is $\prProp$.
    Since members of $H_r$ behave independently, given $j$ does not have $m$ at the start of $r$ and $X_r = s$, the probability that $j$ receives~$m$ in round $r$ equals 
    $1 - (1-\prProp)^{s}$.
    Accordingly, conditioned on $X_r = s$, the number of new processes that receive $m$ in round~$r$ is $\mathrm{Bin}(n-s,1 - (1-\prProp)^{s})$.
    Note that $(X_{r})_{r\ge 1}$ is a Markov-chain, with the initial condition $X_1=\chi$, and transitions
    \begin{align*}
        &\Pr(X_{r+1}=t\mid X_r=s)
        \\&\quad= \Pr\!\big( \mathrm{Bin}\bm{(}n-s,1 - (1-\prProp)^{s}\bm{)} = t-s \big), &\hfill t \in \{s,\dots,n\}.
    \end{align*}
    Hence, we can define the transition matrix $T$ of $(X_r)_{r\ge 1}$ as follows:
    \begin{align*}
         T_{s,t} =
         \begin{cases}
            \Pr\!\big( \mathrm{Bin}\bm{(}n-s,1 - (1-\prProp)^{s}\bm{)} = t-s \big) & t \ge s
            \\
            0 & t<s,
         \end{cases}
    \end{align*}
    where $s,t\in\{1,\dots,n\}$.
    Note that state $n$ is absorbing, i.e., $T_{n,n}=1$. 
    Let $I_\chi$ be a row vector with length~$n$ that has $1$ at index $\chi$ and $0$ otherwise, and $I_n$
    be a column vector with length~$n$ that has $1$ at index $n$ and $0$ otherwise. 
    Since $X_1=\chi$, after rounds $1,\dots,k$, the distribution of the number of processes that have the message is
    $I_\chi T^{\,k}$, hence
    \begin{align*}
    \Pr(\text{all processes receive $m$ by the end of round $k$})
    =\Pr(X_{k+1}=n)
    =I_\chi T^{\,k} I_n.
    \end{align*}
\end{proof}

\begin{figure}[!t]
    \centering
    \includegraphics[scale=0.53]{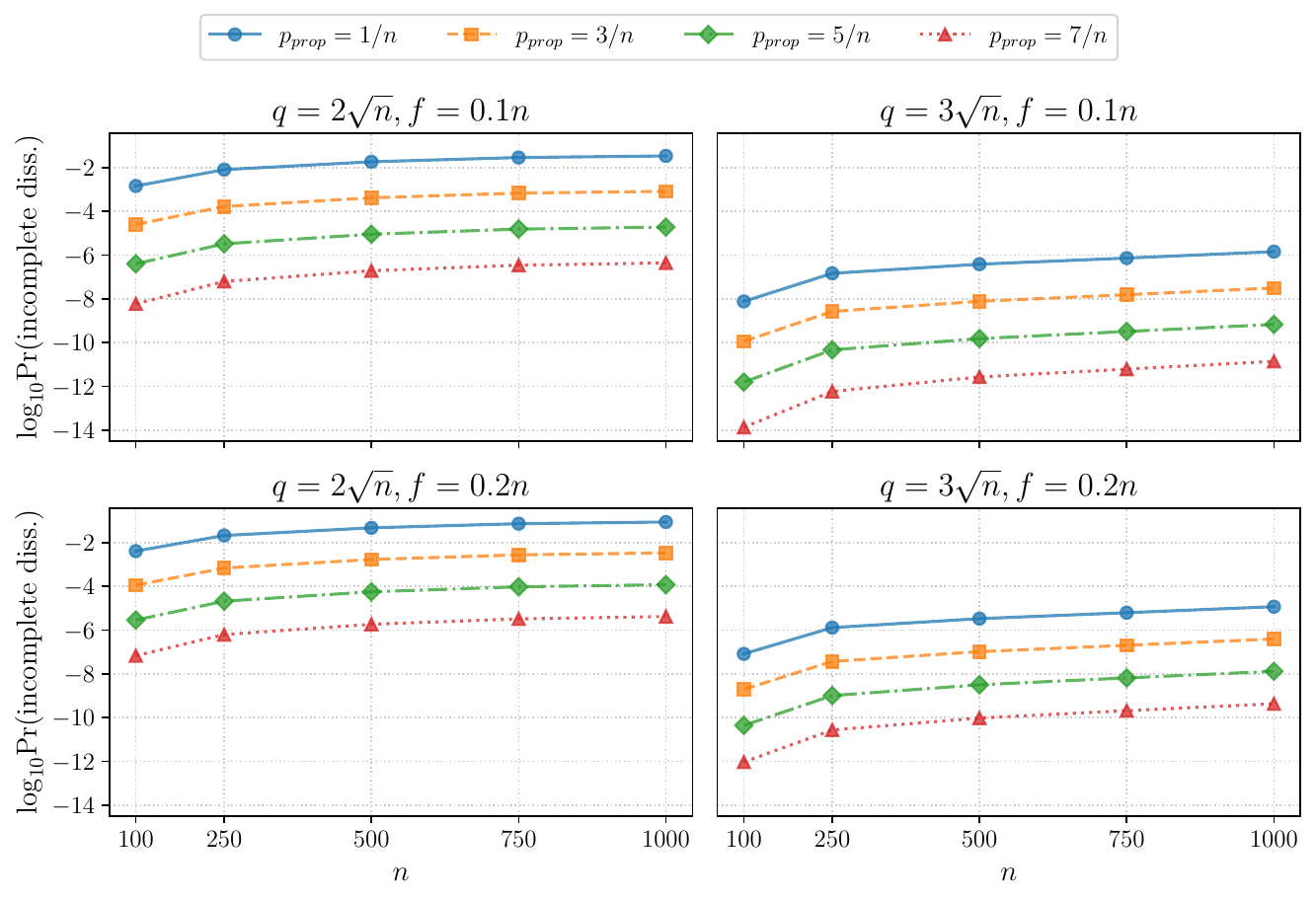}
    \caption{The probability of incomplete dissemination, i.e., in a given epoch $e$ with a correct leader, the probability that at least one correct process does not receive the block proposed by the leader until the end of the propose phase of epoch $e+1$.
    }
    \label{fig:dissemination}
\end{figure}

\myparagraph{Discussion.}
We now present a concrete example to highlight the difference between the guarantees provided by Theorems~\ref{thm:all:receive} and \ref{thm:all:receive:exact}.
Assume $n=500$, $\prProp = 10/500$, and $k=4$. 
Theorem~\ref{thm:all:receive:exact} computes the probability that all processes receive the message for any value of $\chi\ge 1$; however, for this example, Theorem~\ref{thm:all:receive} requires $\chi \ge 76$.
To see why, recall that by Theorem~\ref{thm:all:receive}, the probability that all processes receive a message by round~$k$ is at least $1 - (n-\chi)\cdot\exp(- k \chi \prProp)$; for any $\chi < 76$, this expression is negative.
Consequently, for small $\chi$, the guarantee provided by Theorem~\ref{thm:all:receive} cannot be used.
We now compare the probabilities for $\chi \ge 76$.
For $\chi = 76$, Theorem~\ref{thm:all:receive} yields $0.0298$ while Theorem~\ref{thm:all:receive:exact} gives $0.99999999995$.
Due to this reason, we use Theorem~\ref{thm:all:receive} for presenting closed-form bounds; however, we use Theorem~\ref{thm:all:receive:exact} in our evaluations.

\myparagraph{Concrete dissemination probabilities.}
Fig.~\ref{fig:dissemination} illustrates the probability of incomplete dissemination, namely the probability that, in an epoch with a correct leader, at least one correct process does not receive the proposed block by the end of the proposal phase of the next epoch. 
The results show that increasing $\prProp$ substantially decreases this probability across all configurations. 
Moreover, larger quorum sizes improve dissemination, since they correspond to larger samples and therefore increase the number of processes that can help propagate the block. Overall, the figure confirms that even a small amount of additional propagation is sufficient to make incomplete dissemination unlikely under the considered parameter choices.

%% file: safety.tex
\section{Safety Analysis}
\label{sec:app:safety}

This appendix presents the safety analysis of \protocolName. 
We first analyze the simpler case in which a client observes a chain of certified blocks proposed in consecutive epochs, and show that the probability of a safety violation decreases exponentially with the number of such blocks. 
We then extend the analysis to the general case where certified blocks may be proposed in non-consecutive epochs. 
In this case, the safety probability depends on the pattern of certified blocks and gaps observed by the client, so we specify the client-side finalization procedure used to compute this probability.

\subsection{Safety Analysis for Blocks with Consecutive Epochs}

The safety analysis relies on quorum intersection and uses two main ingredients.
First, assuming that $p > 1/2$, we choose the protocol parameters so that, whenever a block~$B$ proposed in an epoch~$e$ is certified in epoch~$e+1$, the number of processes that become candidates for $B$ is at least $(n+f)/2$ with probability at least~$p$.
Second, the protocol uses a locking rule.
When a correct process becomes a candidate to vote for a block~$B$ proposed in epoch~$e$, it \emph{locks} on the parent block of~$B$, and after this, it will not become a candidate for any block whose height is less than or equal to that of the locked block.

Now consider two chains $\chain = \genesis, \dots, B, B_1, B_2$ and $\chain' = \genesis, \dots, B, B'_1$, such that $B_1.\txs \neq B'_1.\txs$ and $B_2$ is proposed in an epoch earlier than the epoch in which $B'_1$ is proposed.
Any correct process that becomes a candidate for~$B_2$ locks on~$B_1$, and by the locking rule, it cannot later become a candidate for~$B'_1$. 
Thus, the sets of correct candidate processes for~$B_2$ and~$B'_1$ are disjoint.
It follows that $B_2$ and $B'_1$ cannot both have at least $(n+f)/2$ candidates. 
Hence, at least one of the two blocks has fewer than $(n+f)/2$ correct candidates. 
By the parameter choice above, the probability that a certified block has fewer than $(n+f)/2$ candidates is at most~$1-p$. Hence, if both $B_2$ and $B'_1$ are certified, at least one of these certifications must occur through an event of probability at most~$1-p$.
Extending this argument, if there exists a chain of $\kappa$ blocks $B_1,\dots,B_{\kappa}$ proposed in consecutive epochs, then the small probability of certifying two conflicting blocks at the same height decreases as $\kappa$ grows, and for sufficiently large $\kappa$ this probability becomes negligible.

\begin{lem}\label{lem_psync_safety_1}
In any epoch~$e$, if at most $\frac{n+f}{2}$ processes become candidates to vote for a block~$B$ proposed in epoch~$e$, then a quorum certificate is formed for $B$ with the probability of
$
    \exp\!\left( 
          - \frac{\delta^2(n+f)\prSample}{2(2+\delta)}
          \right)
$,
where $\delta \ge \frac{2q}{(n+f)\prSample} - 1$.
\end{lem}
\begin{proof}
Consider a block proposed in epoch~$e$.
Let $C_e \subseteq \Pi$ be the set of processes that become candidates to vote for that block.
Each process in $C_e$ votes with probability $\prSample$.
Let $X$ be the number of such votes. 
Given $|C_{e}| = c$, $X \sim \mathrm{Bin}\!\left(c, \prSample\right)$.
Recall that at least $q$ votes are needed to certify a block.
Accordingly, we must show if $|C_e| \le \frac{n+f}{2}$, then $\Pr\!\left( X \ge q \right) = \mathrm{neg}$.
Define $Y \sim \mathrm{Bin}\!\left(\frac{n+f}{2}, \prSample\right)$.
If $|C_e| \le \tfrac{n+f}{2}$, we have: 
\begin{align*}
\Pr\!\left( X \ge q \right) 
\le \Pr\!\left( Y \ge q \right). 
\end{align*}
By applying the Chernoff bound~\ref{ineq:chernof}, for any $\delta \ge 0$, we have: \kartik{perhaps we need the Chernoff bound with $1-\delta$?}
\hasan{I'm trying to find an upper bound for $\Pr(X\ge q)$; if we use the Chernoff bound with $1-\delta$, then we have: $\Pr(X \ge q) \le \Pr(Y \ge q) = \Pr(Y \ge (1-\delta)\EX{Y}) = 1 -  \Pr(Y \le (1-\delta)\EX{Y})$; now, note that the Chernoff bound gives us an upper bound for $\Pr(Y \le (1-\delta)\EX{Y})$, but we need its lower bound. Accordingly, we cannot use $1-\delta$, unless we formalize the problem differently.}\kartik{I see. what you are saying makes sense. The reason why I mentioned this is: I was thinking (n+f)/2 suffices to create a cert but perhaps an off by one suffices in our scenario.}
\begin{align*}
\Pr\!\left( Y \ge (1+\delta)\EX{Y} \right) 
&= \Pr\!\left( Y \ge \frac{(1+\delta)(n+f)\prSample}{2} \right) 
\\&\le \exp\!\left( 
          - \frac{\delta^2(n+f)\prSample}{2(2+\delta)}
          \right).          
\end{align*}
Now to bound $\Pr\!\left( Y \ge q \right)$, we choose $\delta \ge \frac{2q}{(n+f)\prSample} - 1$.
By assumption, $q \ge (n+f)\prSample/2$; hence, $\delta \ge 0$.
Therefore,
\begin{align*}
    \Pr\!\left( X \ge q \right) 
    \le \exp\!\left( 
          - \frac{\delta^2(n+f)\prSample}{2(2+\delta)}
          \right).
\end{align*}
\end{proof}

\begin{lem}\label{lem:new:byz}
    Let $S$ be a random sample obtained by including each process independently with probability~$\prSample$.
    Then, for any $\delta \ge 0$, the number of Byzantine processes in~$S$ is greater than or equal to $(1+\delta)f\,\prSample$
    with probability at most $1 - \exp\!\left(-\frac{\delta^2\,f\,\prSample}{2+\delta}\right)$.
\end{lem}    
\begin{proof}
    Each process is included in the random sample independently with probability~$\prSample$.
    Let $X$ denote the number of Byzantine processes in the random sample.
    Hence, $X \sim \mathrm{Bin}(f,\prSample)$, and $\EX{X} = f\,\prSample$.
    For any $\delta \ge 0$, we have:
    \begin{align*}
        \Pr(X \ge (1+\delta)\EX{X}) \le \exp\!\left(-\frac{\delta^2\,\EX{X}}{2+\delta}\right).
    \end{align*}
    Accordingly, we have:
    \begin{align*}
        \Pr(X \ge (1+\delta)f\,\prSample) \le \exp\!\left(-\frac{\delta^2\,f\,\prSample}{2+\delta}\right).
    \end{align*}
\end{proof}
\begin{thm}
    Let $S$ be a random sample obtained by including each process independently with probability~$\prSample$.
    Consider any set $Q \subseteq S$ of size $q$.
    Then, 
    \begin{align*}
        \Pr(\text{$Q$ has at least $c$ correct members}) 
        \geq
        \left(1 - \exp\!\left(-\frac{\delta^2\,f\,\prSample}{2+\delta}\right) \right),
    \end{align*}
    where $\delta = \frac{q-c}{f\prSample} - 1$.
\end{thm}
\begin{proof}
    Consider any set $Q \subseteq S$ of $q$ distinct processes.
    We need to compute a lower bound for the probability that $Q$ has at least $c$ correct members.
    Let $X$ denote the number of Byzantine processes in sample~$S$. 
    In the worst case, the adversary forms $Q$ by including all Byzantine processes from~$S$.
    Note that if $X \le q - c$, then $Q$ contains at least $c$ correct members.
    Hence, the desired probability equals $\Pr(X \le q - c)$.
    Using Lemma~\ref{lem:new:byz}, we have:
    \begin{align*}
        \Pr(X \le q - c)
        &\ge 
        \left(1 - \exp\!\left(-\frac{\delta^2\,f\,\prSample}{2+\delta}\right) \right),
    \end{align*}
    where $\delta = \frac{q-c}{f\prSample} - 1$.
\end{proof}


\begin{lem}\label{lem_psync_safety_2}
    Suppose $C_1$ and $C_2$ are two sets of processes that become candidates to vote for blocks~$B_1$ and $B_2$, respectively.
    If at least one of these sets has up to $\frac{n+f}{2}$ processes (i.e., $|C_1|\leq \frac{n+f}{2}$ or $|C_2|\leq \frac{n+f}{2}$), then two quorum certificates will be formed for $B_1$ and $B_2$  with a probability of at most 
    $2\exp\!\left( 
      - \frac{\delta^2(n+f)\prSample}{2(2+\delta)}
      \right)$, where
    $\delta \ge \frac{2q}{(n+f)\prSample} - 1$.
\end{lem}
\begin{proof}
    Consider two blocks~$B_1$ and $B_2$.
    Let $C_1$ and $C_2$ be two sets of processes that become candidates to vote for $B_1$ and $B_2$, respectively.
    Further, let $X_1$ and $X_2$ be the number of votes cast for blocks~$B_1$ and $B_2$, respectively.
    Recall that a quorum certificate is formed for $B_1$ (resp. $B_2$) if $X_1 \ge q$ (resp. $X_2 \ge q$).
    We must compute 
    $$\Pr\!\left( \{X_1 \ge q\} \cap \{X_2 \ge q\} \,\Bigg|\, \left\{|C_1| \le \frac{n+f}{2} \right\} \cup \left\{|C_2| \le \frac{n+f}{2} \right\}\right).$$
    Let $A = \{ X_1 \ge q \}$, 
    $B = \{ X_2 \ge q \}$, 
    $C = \left\{|C_1| \le \frac{n+f}{2}\right\}$, and
    $D = \left\{|C_2| \le \frac{n+f}{2}\right\}$.
    We have:
    \begin{align*}
    \Pr\!\left( A \cap B \mid C \cup D \right)
    & =
    \frac{\Pr\!\big( A \cap B \cap \left(C \cup D\right)\big)}{\Pr\!\left( C \cup D \right)}
    \\& =
    \frac{\Pr\!\big((A \cap B \cap C) \cup 
    (
    A \cap B \cap
    D)\big)}{\Pr\!\left( C \cup D \right)}
    \\& \le
    \frac{\Pr\!\big( (A \cap C ) \cup 
    (
    B \cap
    D)\big)}{\Pr\!\left( C \cup D \right)}
    \\& \le
    \frac{\Pr(A \cap C) + 
    \Pr(
    B \cap
    D)}{\Pr\!\left( C \cup D \right)}
    \\& =
    \frac{\Pr(A \mid C) \Pr(C)+ 
    \Pr(
    B \mid
    D)\Pr(D)}{\Pr\!\left( C \cup D \right)}
    \\& \le
    \frac{\Pr(A \mid C) \Pr(C)+ 
    \Pr(
    B \mid
    D)\Pr(D)}{\max\big\{\Pr(C),\Pr(D)\big\}}
    \\& \le
    \Pr(A \mid C) + \Pr(B \mid D)
    \\& \le 
    2\exp\!\left( 
          - \frac{\delta^2(n+f)\prSample}{2(2+\delta)}
          \right)
    & \hfill \text{(by Lemma \ref{lem_psync_safety_1})}
    \end{align*}
    where $\delta \ge \frac{2q}{(n+f)\prSample} - 1$.
\end{proof}

\begin{thm}\label{thm_psync_safety}
    Suppose $\kappa \ge 2$, and $\prSample = O(1/\sqrt{n})$.
    A safety violation occurs for the \protocolName protocol under partial synchrony with a probability of at most
    $2^{\kappa-1} \cdot \exp\!\left( 
      - \frac{\delta^2(n+f)\prSample}{2(2+\delta)}
      \right)^{(\kappa-1)}
    $,
    where $\delta \ge \frac{2q}{(n+f)\prSample} - 1$.
\end{thm}
\begin{proof}
Suppose a correct process~$i$ certifies $k+\kappa$ blocks $B^1,\dots,B^k,B_{1}, \dots, B_{\kappa}$ that satisfy the following conditions: 
\begin{enumerate}
    \item $B^1,\dots,B^k,B_{1}, \dots, B_{\kappa}$ form a chain, i.e., 
    \begin{align*}
    \begin{cases}
       B^1.\height = B^2.\height-1 = \dots = B_{\kappa}.\height - k - \kappa + 1, \text{ and}
    \\ \hash{B^1} = B^2.\parentHash, \dots, \hash{B_{\kappa-1}} = B_{\kappa}.\parentHash.
    \end{cases}
    \end{align*}
    \item $B_1,\dots,B_{\kappa}$ are proposed in consecutive epochs, i.e., 
    $$B_1.\var{epoch}=B_2.\var{epoch}-1=\dots=B_{\kappa}.\var{epoch}-\kappa+1.$$
    \item If $k=0$, $B_{1}$ extends the block committed by $i$ with the greatest height; 
    otherwise, $B^1$ extends the block committed by $i$ with the greatest height.
    \item From the viewpoint of $i$, there is no certified block $B\neq B_1$ with the same height as $B_1$.
\end{enumerate}
Consequently, process~$i$ returns $B_1$ and its ancestors by executing the function $\fn{get\_ledger}(\kappa)$.
A \emph{safety violation} occurs when there is at least a correct process~$j \neq i$ that observes an alternative chain of certified blocks 
$B_{1}', B_{2}', \dots, B_{\kappa}'$, such that
$B_{1}', B_{2}', \dots, B_{\kappa}'$ are proposed in consecutive epochs,
$B_1.\txs \neq B_1'.\txs$, 
$B_1.\height = B_{1}'.\height$,\kartik{it only needs to equivocate? need not have the same height?}\hasan{I didn't understand this question}\kartik{It should just be the case that $B_1$ and $B_{1}'$ are on equivocating chains. They need not have the same height, correct?} and 
$j$ commits $B_1'$.

Let $C_1,\dots,C_\kappa$ denote the sets of processes that become candidates to vote for $B_1,\dots,B_\kappa$, respectively.
Further, let $C_1',\dots,C_\kappa'$ denote the sets of processes that become candidates to vote for $B_1',\dots,B_\kappa'$, respectively.

\medskip
\noindent
\textbf{$\kappa = 2$.}
There are two possible cases:
\begin{itemize}
    \item There exists a block in the first chain and a block in the second chain that are proposed in the same epoch.
    For example, blocks $B_2$ and $B_1'$ are proposed in the same epoch.
    In this case, $C_2$ and $C_1'$ should not have any common correct process.
    This follows from the fact that, in each epoch, a correct process becomes a candidate to vote at most once.
    Accordingly, at least one of these sets has up to $\frac{n+f}{2}$ processes (i.e., $|C_2|\le \frac{n+f}{2}$ or $|C_1'|\le \frac{n+f}{2}$).
    Lemma~\ref{lem_psync_safety_2} provides an upper bound on the probability of this case.
    \item Every block in the first chain is proposed in an epoch different from every block in the second chain.
    Without loss of generality, we assume block~$B_2$ is proposed before $B_1'$.
    We have two sub-cases:
    \begin{itemize}
        \item $|C_2| > \frac{n+f}{2}$ and $|C_1'| > \frac{n+f}{2}$.
        In this case, $C_2$ and $C_1'$ have at least a common correct process.
        Since $B_2.\height > B_1'.\height$, the correct processes that are common between $C_2$ and $C_1'$ have locked on $B_1$; hence, they do not vote for $B_1'$ as $B_1'.\height \ngeq B_1.\height$.
        Consequently, at most $\frac{n+f}{2}$ processes become candidates for $B_1'$. 
        Lemma~\ref{lem_psync_safety_1} provides the probability of certifying $B_1'$ when at most $\frac{n+f}{2}$ processes become candidates.
        \item $|C_2| \le \frac{n+f}{2}$ or $|C_1'| \le \frac{n+f}{2}$.
        Lemma~\ref{lem_psync_safety_2} provides an upper bound on the probability of this case.
    \end{itemize}
\end{itemize}
Accordingly, if $\kappa = 2$, the probability of a safety violation is at most the maximum of the probabilities provided by Lemmas~\ref{lem_psync_safety_1} and \ref{lem_psync_safety_2}, which is $$2\exp\!\left( 
          - \frac{\delta^2(n+f)\prSample}{2(2+\delta)}
          \right),
    $$
    where $\delta \ge \frac{2q}{(n+f)\prSample} - 1$.

\medskip
\noindent
\textbf{$\kappa = 3$.}
There are two possible cases:
\begin{itemize}
    \item Among the first two blocks of both chains, there exists a pair of blocks, one from each chain, proposed in the same epoch.
    For example, blocks $B_2$ and $B_1'$ are proposed in the same epoch; hence, neither $(C_2,C_1')$ nor $(C_3,C_2')$ should have any common correct process.
    Using Lemma~\ref{lem_psync_safety_1}, the probability of this case is at most 
    $
    \exp\!\left( 
          - \frac{\delta^2(n+f)\prSample}{2(2+\delta)}
          \right)^2
    $,
    where $\delta \ge \frac{2q}{(n+f)\prSample} - 1$.
    \item Among the first two blocks of both chains, there is no pair of blocks, one from each chain, that were proposed in the same epoch.
    Without loss of generality, we assume block~$B_2$ is proposed before $B_1'$.
    We have two sub-cases:
    \begin{itemize}
        \item $|C_2| > \frac{n+f}{2}$, 
        $|C_3| > \frac{n+f}{2}$,
        $|C_1'| > \frac{n+f}{2}$,  
        and $|C_2'| > \frac{n+f}{2}$.
        In this case, sets~$C_2$ and $C_1'$ have at least a common correct process.
        Since $B_2.\height > B_1'.\height$, the correct processes that are common between $C_2$ and $C_1'$ have locked on $B_1$; consequently, they do not vote for $B_1'$ as $B_1'.\height \ngeq B_1.\height$.
        Consequently, at most $\frac{n+f}{2}$ processes vote for $B_1'$. 
        Similarly, sets $C_3$ and $C_2'$ have at least a common correct process, and at most $\frac{n+f}{2}$ processes vote for $B_2'$. 
        Accordingly, the probability of this case occurring is at most
        $\exp\!\left( 
              - \frac{\delta^2(n+f)\prSample}{2(2+\delta)}
              \right)^2$,
        where $\delta \ge \frac{2q}{(n+f)\prSample} - 1$.
        \item $|C_2| \le \frac{n+f}{2}$, 
        $|C_3| \le \frac{n+f}{2}$,
        $|C_1'| \le \frac{n+f}{2}$,  
        or $|C_2'| \le \frac{n+f}{2}$.
        Using Lemma~\ref{lem_psync_safety_2}, the probability of this case is at most
        $4\exp\!\left( 
              - \frac{\delta^2(n+f)\prSample}{2(2+\delta)}
              \right)^2$,
        where $\delta \ge \frac{2q}{(n+f)\prSample} - 1$.
    \end{itemize}
\end{itemize}
Accordingly, if $\kappa = 3$, the probability of a safety violation is at most $$4\exp\!\left( 
          - \frac{\delta^2(n+f)\prSample}{2(2+\delta)}
          \right)^2,
    $$
    where $\delta \ge \frac{2q}{(n+f)\prSample} - 1$.

\medskip
\noindent
\textbf{Any $\kappa \ge 2$.} 
For any $\kappa \ge 2$, we can follow a similar argument; hence, the probability of a safety violation is at most:
$2^{\kappa-1} \cdot \exp\!\left( 
  - \frac{\delta^2(n+f)\prSample}{2(2+\delta)}
  \right)^{(\kappa-1)}
$,
where $\delta \ge \frac{2q}{(n+f)\prSample} - 1$.
\end{proof}

\begin{cor}[Safety under partial synchrony]\label{cor_psync_safety}
    Suppose $q \ge (n+f)\prSample/2$, and $\kappa \ge 2$, $\prSample = O(1/\sqrt{n})$.
    Then, a safety violation for the \protocolName protocol under partial synchrony occurs with a probability of at most
    $$\exp\!\big(O(-(\kappa-1)\sqrt{n})\big).$$ 
\end{cor}
\begin{proof}
    This corollary directly follows from Theorem~\ref{thm_psync_safety}.
\end{proof}

\subsection{Safety Analysis for Blocks with Non-Consecutive Epochs}
\label{app:non-consecutive}

We now extend the safety analysis to the general case in which the certified blocks observed by a client are not necessarily proposed in consecutive epochs. 
This case is more involved than the consecutive-epoch case because gaps in the observed chain may give the adversarial scheduler opportunities to build a hidden conflicting chain. 
Hence, the safety probability is no longer determined only by the number of certified descendants of a block, but also by the pattern of epochs in which these descendants appear.

At a high level, the client-side safety computation proceeds as follows. Given a certified chain $\chain$ and a block $B \in \chain$, the client extracts the suffix of $\chain$ starting from the parent of $B$ and ending at the highest certified block in $\chain$. 
This suffix is represented as a sequence $C$ indexed by epochs, where missing entries correspond to epochs in which the observed chain contains no certified block. 
The client then considers all possible hidden chains $C'$ that start from the same parent block and may conflict with the observed chain. 
The goal is to identify the hidden chain that maximizes the probability of replacing the observed chain.

The analysis uses two types of conflicts. 
The first type arises from the locking rule. 
If two blocks are such that no correct process can become a candidate for both, then certifying both blocks requires an exceptional event whose probability is bounded by $p_{\mathsf{lock}}$. 
The second type arises from Byzantine leader equivocation: if two conflicting blocks are proposed in the same epoch, certifying both requires an equivocation event whose probability is bounded by $p_{\mathsf{eq}}$. 
Thus, the client searches over all candidate hidden chains $C'$ and identifies the one that maximizes the probability of replacing the observed chain. 
For this worst-case hidden chain, the client counts the number $c_{\mathsf{lock}}$ of lock-induced conflicts and the number $c_{\mathsf{eq}}$ of same-epoch equivocation conflicts, and uses
$p_{\mathsf{lock}}^{c_{\mathsf{lock}}} \cdot p_{\mathsf{eq}}^{c_{\mathsf{eq}}}$
as an upper bound on the probability that this hidden chain replaces the observed chain.

Algorithm~\ref{alg:client} implements this computation. The function $\fn{safety}$ builds the observed epoch pattern $C$ and uses $\fn{search}$ to enumerate possible hidden chains $C'$. Each hidden chain may grow during gaps in $C$, or either skip or conflict with epochs where $C$ contains a certified block. Algorithm~\ref{alg:client:recursive:prob:aux} then evaluates each completed hidden chain by counting its lock-induced and same-epoch equivocation conflicts, and the client returns the maximum value as a worst-case upper bound on the probability that a conflicting chain replaces the observed one.

\subsubsection*{Preliminaries}
We first establish the probabilistic bounds used to evaluate the safety of a client-observed chain. These bounds capture two ways in which a hidden conflicting chain can be certified: either through blocks that conflict because of the locking rule, or through distinct blocks proposed in the same epoch by a Byzantine leader. The resulting probabilities, denoted by $p_{\mathsf{lock}}$ and $p_{\mathsf{eq}}$, are later used by the client-side safety procedure to bound the probability that a hidden chain can replace the observed one.


\begin{lem}\label{lem:02}
In any epoch~$e$, if at most $\frac{n+f}{2}$ processes become candidates to vote for a block~$B$ proposed in epoch~$e$, then a quorum certificate is formed for $B$ with the probability of
\begin{align*}
    \exp\!\left( 
          - \frac{\alpha\,\delta^2(n+f)\prSample}{2(2+\delta)}
          \right)
    \cdot 
    \left(
        1 - \exp\!\left( 
          - \frac{\delta'^2 f\prSample}{2+\delta'}
          \right)
    \right),
\end{align*}
where $\delta \ge \frac{2q}{(n+f)\prSample} - 1$ and $\delta' \ge \frac{q}{f\prSample} - 1$.
\end{lem}
\begin{proof}
Consider a block~$B$ proposed in epoch~$e$.
Let $C \subseteq \Pi$ be the set of processes that become candidates to vote for that block.
Recall that the indirect certification mechanism uses $\alpha$ random samples $S_1,\dots,S_\alpha$ drawn from~$C$, where each process in~$C$ is included in each sample independently with probability~$\prSample$.
A voter (i.e., a process that is a member of any of the random samples $S_1,\dots,S_\alpha$) sends a \msgVote message to every member of the first random sample of epoch~$e+1$, $\sampleLOne{e+1}$.
If a correct process in~$S_{e+1}$ receives at least~$q$ matching votes from each of the $\alpha$ samples $S_1,\dots,S_\alpha$, it sends a signature share to all members of~$S_{e+1}$; if a process in~$S_{e+1}$ receives at least~$q$ signature shares from members of~$S_{e+1}$, it aggregates them into a quorum certificate for~$B$.
Thus, a quorum certificate for~$B$ is formed only if 
\begin{enumerate*}[label=(E\arabic*)]
    \item for each of the $\alpha$ samples $S_1,\dots,S_\alpha$ a probabilistic quorum can be formed, and
    \item any subset of~$S$ of size~$q$ contains at least one correct process (since such a correct process guarantees that the probabilistic quorums from $S_1,\dots,S_\alpha$).
\end{enumerate*}
Accordingly, we compute the probabilities of events~(E1) and~(E2) and, since these events are independent (because each sample $S_1,\dots,S_\alpha$ and $\sampleLOne{e+1}$ is selected independently), bound the probability that a quorum certificate is formed for~$B$ by their product.

\noindent\textit{Bounding the probability of event~(E1).}
For each $j \in \{1,\dots,\alpha\}$, let $X_j$ denote the number of votes contributed by sample~$S_j$.
Given $|C| = c$, we have $X_j \sim \mathrm{Bin}(c, \prSample)$.
A probabilistic quorum can be formed in sample~$S_j$ iff $X_j \ge q$, so event~(E1) is exactly $\bigcap_{j=1}^{\alpha}\{X_j \ge q\}$.
For each $j$, define $Y_j \sim \mathrm{Bin}\!\left(\frac{n+f}{2}, \prSample\right)$.
If $|C| \le \tfrac{n+f}{2}$, we have, for every $j$:
\begin{align*}
\Pr\!\left( X_j \ge q \right) 
\le \Pr\!\left( Y_j \ge q \right). 
\end{align*}
By applying the Chernoff bound~\ref{ineq:chernof}, for any $\delta \ge 0$:
\begin{align*}
\Pr\!\left( Y_j \ge (1+\delta)\EX{Y_j} \right) 
&= \Pr\!\left( Y_j \ge \frac{(1+\delta)(n+f)\prSample}{2} \right) 
\\&\le \exp\!\left( 
          - \frac{\delta^2(n+f)\prSample}{2(2+\delta)}
          \right).          
\end{align*}
To bound $\Pr(Y_j \ge q)$, we choose $\delta \ge \frac{2q}{(n+f)\prSample} - 1$.
By assumption, $q \ge (n+f)\prSample/2$; hence, $\delta \ge 0$.
Therefore, for each $j$,
\begin{align*}
    \Pr\!\left( X_j \ge q \right) 
    \le \exp\!\left( 
          - \frac{\delta^2(n+f)\prSample}{2(2+\delta)}
          \right).
\end{align*}
Since the events $\{X_j \ge q\}_{j=1}^{\alpha}$ are mutually independent, we obtain
\begin{align*}
\Pr(\text{event~(E1)})
= \prod_{j=1}^{\alpha}\Pr\!\left( X_j \ge q \right)
\le \exp\!\left( 
          - \frac{\alpha\,\delta^2(n+f)\prSample}{2(2+\delta)}
          \right).
\end{align*}

\noindent\textit{Bounding the probability of event~(E2).}
Each process is included in~$\sampleLOne{e+1}$ independently with probability~$\prSample$.
Let $Z$ denote the number of Byzantine processes in~$\sampleLOne{e+1}$.
Then $Z \sim \mathrm{Bin}(f, \prSample)$ with $\EX{Z} = f\prSample$.
Event~(E2) fails iff some subset of~$\sampleLOne{e+1}$ of size~$q$ consists entirely of Byzantine processes, i.e., iff $Z \ge q$.
By Lemma~\ref{lem:new:byz}, for any $\delta' \ge 0$:
\begin{align*}
\Pr(Z \ge (1+\delta')f\prSample) 
\le \exp\!\left( - \frac{\delta'^2 f\prSample}{2+\delta'} \right).
\end{align*}
Choosing $\delta' \ge \frac{q}{f\prSample} - 1$ (non-negative under standard parameter settings since $q > f\prSample$) yields
\begin{align*}
\Pr(\neg\text{event~(E2)}) = \Pr(Z \ge q) 
\le \exp\!\left( - \frac{\delta'^2 f\prSample}{2+\delta'} \right).
\end{align*}

\noindent\textit{Combining.}
By the independence argument above,
\begin{align*}
\Pr(\text{QC formed for $B$}) 
&\le \Pr(\text{event~(E1)})\cdot \Pr(\text{event~(E2)}) 
\\&\le \exp\!\left( 
          - \frac{\alpha\,\delta^2(n+f)\prSample}{2(2+\delta)}
          \right)
\cdot \left(1 - \exp\!\left( - \frac{\delta'^2 f\prSample}{2+\delta'} \right)\right),
\end{align*}
where $\delta \ge \frac{2q}{(n+f)\prSample} - 1$ and $\delta' \ge \frac{q}{f\prSample} - 1$.
\end{proof}

\begin{deff}\label{def:confilict}
Two blocks~$B$ and~$B'$ are \emph{conflicting} if no correct process can become a candidate to vote for both.
\end{deff}

\begin{thm}\label{thm:01}
Suppose $B$ and~$B'$ are two conflicting blocks.
Then, both $B$ and $B'$ become certified with probability at most $2 p + p^{\,2}$, where
\begin{align*}
    p = 
    \exp\!\left( - \frac{\alpha\,\delta^2(n+f)\prSample}{2(2+\delta)} \right)
    \cdot 
    \left( 1 - \exp\!\left( - \frac{\delta'^2 f\prSample}{2+\delta'} \right) \right),
\end{align*}
with $\delta \ge \frac{2q}{(n+f)\prSample} - 1$ and $\delta' \ge \frac{q}{f\prSample} - 1$.
\end{thm}
\begin{proof}
Let $C$ and $C'$ denote the sets of processes that become candidates to vote for~$B$ and~$B'$, respectively.
We first observe that it is impossible to have $|C| > (n+f)/2$ and $|C'| > (n+f)/2$ simultaneously.
Indeed, if both held, then by inclusion-exclusion at least one correct process would belong to $C \cap C'$, contradicting Definition~\ref{def:confilict}.
Hence, at least one of $|C|, |C'|$ is at most $(n+f)/2$, and the event ``both $B$ and $B'$ are certified'' decomposes into three disjoint cases:
\begin{enumerate*}[label=(\roman*)]
    \item $|C| > (n+f)/2$ and $|C'| \le (n+f)/2$;
    \item $|C| \le (n+f)/2$ and $|C'| > (n+f)/2$;
    \item $|C| \le (n+f)/2$ and $|C'| \le (n+f)/2$.
\end{enumerate*}

For simplicity, we assume that $\Pr(\text{a block is certified} \mid \text{\#candidates} > (n+f)/2) = 1$.
By Lemma~\ref{lem:02}, $\Pr(B \text{ certified} \mid |C| \le (n+f)/2) \le p$ and analogously for~$B'$.
Furthermore, since the samples used to certify~$B$ and those used to certify~$B'$ are selected independently, the two certification events are conditionally independent given $|C|$ and $|C'|$.
Therefore, we have:
\begin{itemize}
    \item Case~(i): $\Pr(B \text{ is certified} \wedge B' \text{ is certified} \mid \text{case (i)}) \le 1 \cdot p = p$.
    \item Case~(ii): symmetric, $\le p$.
    \item Case~(iii): $\Pr(B \text{ is certified} \wedge B' \text{ is certified} \mid \text{case (iii)}) \le p \cdot p = p^{\,2}$.
\end{itemize}
Summing over the three disjoint cases yields
\begin{align*}
    \Pr(B \text{ and } B' \text{ both certified}) 
    \le p + p + p^{\,2} 
    = 2 p + p^{\,2},
\end{align*}
where, by Lemma~\ref{lem:02},
\begin{align*}
    p = 
    \exp\!\left( - \frac{\alpha\,\delta^2(n+f)\prSample}{2(2+\delta)} \right)
    \cdot 
    \left( 1 - \exp\!\left( - \frac{\delta'^2 f\prSample}{2+\delta'} \right) \right),
\end{align*}
with $\delta \ge \frac{2q}{(n+f)\prSample} - 1$ and $\delta' \ge \frac{q}{f\prSample} - 1$, which completes the proof.
\end{proof}

\begin{lemma}
\label{lemma:maximum-uses-all-correct-candidates}
For $m \in \mathbb{N}$, let $\Phi(m)=\Pr(\Bin(m,p_{\mathsf{vote}})\ge q)$.
Suppose $c,c'\in\{0,1,2,\ldots,n-f\}$ and $c+c'\le n-f$.
Then the quantity
$
\Phi(c+f)\Phi(c'+f)
$
is maximized only at a pair $(c,c')$ satisfying
$
c+c'=n-f.
$
\end{lemma}
\begin{proof}
We first show that $\Phi(m)$ is non-decreasing in $m$.
Since $\Phi(m)$ is the probability that at least $q$ successes occur in $m$ independent Bernoulli trials, adding one more trial cannot decrease this probability.
Indeed, if at least $q$ successes occur among the first $m$ trials, then at least $q$ successes also occur among the first $m+1$ trials, regardless of the outcome of the additional trial.
Thus, the event counted by $\Phi(m)$ is contained in the event counted by $\Phi(m+1)$.
Hence, $\Phi(m+1)\ge \Phi(m)$, which implies $\Phi$ is non-decreasing.

Now suppose $c+c'<n-f$.
Then at least one of $c$ or $c'$ can be increased by one while preserving the constraint $c+c'\le n-f$.
Assume, without loss of generality, that we increase $c$ to $c+1$.
Since $\Phi$ is non-decreasing, we have
\begin{align*}
    \Phi(c+1+f)\ge \Phi(c+f).
\end{align*}
Therefore,
\begin{align*}
\Phi(c+1+f)\Phi(c'+f)
\ge
\Phi(c+f)\Phi(c'+f).
\end{align*}
So any feasible pair $(c,c')$ with $c+c'<n-f$ can be replaced by another feasible pair with larger sum and with product no smaller.
By repeating this argument, we eventually reach a feasible pair satisfying $c+c'=n-f$ whose product is at least as large as the original one.
Therefore, the maximum is attained by some pair $(c,c')$ satisfying $c+c'=n-f$.
\end{proof}
\begin{lemma}
\label{lemma:balanced-candidate-split}
For $m \in \mathbb{N}$, let $\Phi(m)=\Pr(\Bin(m,p_{\mathsf{vote}})\ge q)$.
Suppose $c,c'\in\{0,1,2,\ldots,n-f\}$ and $c+c'\le n-f$.
Then the quantity
$
\Phi(c+f)\Phi(c'+f)
$
is maximized when $c$ and $c'$ are as close as possible.
Specifically, if $n-f$ is odd, then
\begin{align*}
\max_{c+c'\le n-f}\Phi(c+f)\Phi(c'+f)
=
\Phi\!\left(f+\left\lfloor\frac{n-f}{2}\right\rfloor\right)
\Phi\!\left(f+\left\lceil\frac{n-f}{2}\right\rceil\right),
\end{align*}
and if $n-f$ is even, then
\begin{align*}
\max_{c+c' \le n-f}\Phi(c+f)\Phi(c'+f)
=
\Phi\!\left(\frac{n+f}{2}\right)^2.
\end{align*}
\end{lemma}
\begin{proof}
By Lemma~\ref{lemma:maximum-uses-all-correct-candidates}, the maximum is attained by a pair $(c,c')$ satisfying
$
c+c'=n-f.
$
Hence, it suffices to show that, under the constraint $c+c'=n-f$, the product
$
\Phi(c+f)\Phi(c'+f)
$
is maximized when $c$ and $c'$ are as close as possible.
We use the fact that binomial upper-tail probabilities form a log-concave sequence in $m$, i.e.,
$\Phi(m+1)^2
\ge
\Phi(m)\Phi(m+2)$
for every $m$~\cite{keilson1971some}.
Equivalently,
\begin{align*}
\frac{\Phi(m+1)}{\Phi(m)}
\ge
\frac{\Phi(m+2)}{\Phi(m+1)}.
\end{align*}

Now suppose $c+1<c'$.
Then $c+f+1<c'+f$.
Applying the inequality above gives
\begin{align*}
\frac{\Phi(c+f+1)}{\Phi(c+f)}
\ge
\frac{\Phi(c'+f)}{\Phi(c'+f-1)}.
\end{align*}
Rearranging yields
\begin{align*}
\Phi(c+f+1)\Phi(c'+f-1)
\ge
\Phi(c+f)\Phi(c'+f).
\end{align*}
Therefore, moving one correct candidate from the larger side to the smaller side cannot decrease the product.
Repeating this balancing step shows that the product is maximized when $c$ and $c'$ differ by at most one.
Therefore, if $n-f$ is odd, the maximum is attained at
$
c=\left\lfloor\frac{n-f}{2}\right\rfloor
$
and
$
c'=\left\lceil\frac{n-f}{2}\right\rceil,
$
or vice versa.
If $n-f$ is even, the maximum is attained at
$
c=c'=(n-f)/2.
$
\end{proof}

\begin{corollary}
\label{corollary:vote-quorum-two-blocks}
Given two blocks proposed in the same epoch, the probability that both blocks obtain the required $\alpha$ probabilistic quorums is at most
\begin{align*}
\Phi\!\left(f+\left\lfloor\frac{n-f}{2}\right\rfloor\right)^{\alpha}
\Phi\!\left(f+\left\lceil\frac{n-f}{2}\right\rceil\right)^{\alpha}.
\end{align*}
If $n-f$ is even, this becomes
\begin{align*}
\Phi\!\left(\frac{n+f}{2}\right)^{2\alpha}.
\end{align*}
\end{corollary}
\begin{proof}
Suppose blocks $B$ and $B'$ are proposed in the same epoch.
Accordingly, a correct process can become a candidate for at most one of them, and the correct candidate sets of $B$ and $B'$ are disjoint, which implies $c+c'\le n-f$.
In the worst case for safety, all Byzantine processes may appear in both candidate sets.
Hence, the total number of candidates for $B$ is at most $c+f$, and the total number of candidates for $B'$ is at most $c'+f$.

For a block with $m$ candidates, the probability of obtaining one probabilistic quorum is $\Phi(m)$.
Thus, the probability that $B$ obtains $\alpha$ probabilistic quorums is at most $\Phi(c+f)^\alpha$, and the probability that $B'$ obtains $\alpha$ probabilistic quorums is at most $\Phi(c'+f)^\alpha$.
Consequently, the probability that both blocks obtain $\alpha$ probabilistic quorums is at most
\begin{align*}
\Phi(c+f)^\alpha \Phi(c'+f)^\alpha
=
\left(\Phi(c+f)\Phi(c'+f)\right)^\alpha.
\end{align*}
By Lemma~\ref{lemma:balanced-candidate-split}, this product is maximized when the correct candidates are split as evenly as possible, and the claimed bounds follow immediately.
\end{proof}

\begin{lemma}
\label{lemma:chernoff-vote}
Let
$m_v=f+\left\lceil\frac{n-f}{2}\right\rceil$,
$\mu_v=m_vp_{\mathsf{vote}}$.
If $q\ge \mu_v$ and
$\delta_v=\frac{q}{\mu_v}-1$,
then
\begin{align*}
\Phi(m_v)
\le
\exp\!\left(-\frac{\delta_v^2\mu_v}{2+\delta_v}\right).
\end{align*}
In particular, if $n-f$ is even, so that $m_v=(n+f)/2$, then
\begin{align*}
\Phi\!\left(\frac{n+f}{2}\right)^{2\alpha}
\le
\exp\!\left(-\frac{\alpha\delta_v^2(n+f)p_{\mathsf{vote}}}{2+\delta_v}\right),
\end{align*}
where
$\delta_v=\frac{2q}{(n+f)p_{\mathsf{vote}}}-1$.
\end{lemma}
\begin{proof}
For one random sample, the number of voters is distributed as $\Bin(m_v,p_{\mathsf{vote}})$ with mean $\mu_v=m_vp_{\mathsf{vote}}$.
By the Chernoff bound,
\begin{align*}
\Pr(\Bin(m_v,p_{\mathsf{vote}})\ge (1+
\delta_v)\mu_v)
\le
\exp\!\left(-\frac{\delta_v^2\mu_v}{2+\delta_v}\right).
\end{align*}
Since $q=(1+\delta_v)\mu_v$, the first inequality follows.
If $n-f$ is even, then $m_v=(n+f)/2$.
Raising the one-sample bound to the power $2\alpha$ gives the second inequality.
\end{proof}

\begin{lemma}
\label{lemma:balanced-signature-share-split}
Let $S_{e+1}$ be the first sample of epoch $e+1$, and let $a=|S_{e+1}\cap \PiC|$ and $b=|S_{e+1}\cap \PiF|$.
If $x$ correct processes in $S_{e+1}$ send a signature share for $B$, and $a-x$ correct processes send a signature share for $B'$, then the probability that both blocks obtain a quorum of signature shares is maximized when $x$ and $a-x$ are as balanced as possible, i.e., when $x=\left\lfloor\frac{a}{2}\right\rfloor$
or $x=\left\lceil\frac{a}{2}\right\rceil$.
\end{lemma}
\begin{proof}
The Byzantine processes may sign both blocks, so they contribute $b$ possible signature shares to each block.
The correct processes cannot sign both blocks, so if $x$ correct processes sign $B$, then $a-x$ correct processes can sign $B'$.
Thus, the two available signature-share counts are $x+b$ and $a-x+b$.
For both blocks to obtain $q$ signature shares, the smaller of these two quantities must be as large as possible.
Note that the value
$
\min\{x+b,a-x+b\}
$
is maximized when $x$ and $a-x$ are as close as possible.
Therefore, the adversary's best split is balanced, and the smaller side has size $b+\lfloor a/2\rfloor$.
\end{proof}

\begin{lemma}
\label{lemma:chernoff-signatures}
Let $S_{e+1}$ be the first sample of epoch $e+1$.
Suppose two blocks are proposed in epoch $e$.
The probability that each block receives at least $q$ signature shares generated by members of $S_{e+1}$ is bounded by
\begin{align*}
\exp\!\left(
-\frac{\delta_s^2(n+f)p_{\mathsf{sample}}}{2(2+\delta_s)}
\right),
\end{align*}
where $\delta_s = \frac{2q}{(n+f)p_{\mathsf{sample}}}-1$.
\end{lemma}
\begin{proof}
Let $a=|S_{e+1}\cap \PiC|$ and $b=|S_{e+1}\cap \PiF|$.
By Lemma~\ref{lemma:balanced-signature-share-split}, the probability that both blocks obtain quorum certificates is maximized when the correct processes in $S_{e+1}$ are split as evenly as possible between the two blocks.
Thus, both blocks can receive at least $q$ signature shares only if
$
b+\left\lfloor a/2\right\rfloor\ge q.
$
Since
$
\lfloor a/2\rfloor\le a/2,
$
we have
\begin{align*}
\left\{
b+\left\lfloor\frac{a}{2}\right\rfloor\ge q
\right\}
\subseteq
\left\{
b+\frac{a}{2}\ge q
\right\}.
\end{align*}
Therefore, it is enough to bound
$
\Pr(b+a/2\ge q).
$
The random variable $b+a/2$ is a sum of independent random variables in $[0,1]$.
Its expectation is
\begin{align*}
\mathbb{E}\!\left[b+\frac{a}{2}\right]
=
fp_{\mathsf{sample}}
+
\frac{n-f}{2}p_{\mathsf{sample}}
=
\frac{n+f}{2}p_{\mathsf{sample}}.
\end{align*}
Let
$
\mu_s
=
\frac{n+f}{2}p_{\mathsf{sample}}.
$
Then
$
q=(1+\delta_s)\mu_s
$
for
$
\delta_s
=
\frac{q}{\mu_s}-1
=
\frac{2q}{(n+f)p_{\mathsf{sample}}}-1.
$
By the Chernoff bound, assuming $q\ge \mu_s$, we obtain
\begin{align*}
\Pr\!\left(
b+\frac{a}{2}\ge q
\right)
\le
\exp\!\left(
-\frac{\delta_s^2\mu_s}{2+\delta_s}
\right).
\end{align*}
Substituting
$
\mu_s=\frac{n+f}{2}p_{\mathsf{sample}}
$
gives
\begin{align*}
\Pr\!\left(
b+\frac{a}{2}\ge q
\right)
\le
\exp\!\left(
-\frac{\delta_s^2(n+f)p_{\mathsf{sample}}}{2(2+\delta_s)}
\right).
\end{align*}
Since the event
$
b+\lfloor a/2\rfloor\ge q
$
is contained in the event
$
b+a/2\ge q,
$
the same bound applies to the probability that both blocks receive at least $q$ signature shares.
\end{proof}

\begin{theorem}
\label{theorem:two-same-epoch-blocks-certified}
Consider two distinct blocks $B$ and $B'$ proposed in the same epoch~$e$.
If $n-f$ is even, then the probability that both $B$ and $B'$ become certified is bounded by
\begin{align*}
p_{\mathsf{eq}}
&\le
\Phi\!\left(\frac{n+f}{2}\right)^{2\alpha}
\cdot
\Pr\!\left(B_F+\left\lfloor\frac{A}{2}\right\rfloor\ge q\right).
\end{align*}
Moreover, using the Chernoff bounds of Lemma~\ref{lemma:chernoff-vote} and Lemma~\ref{lemma:chernoff-signatures}, if
\begin{align*}
\delta_v=\frac{2q}{(n+f)p_{\mathsf{vote}}}-1\ge 0
\qquad\text{and}\qquad
\delta_s=\frac{2q}{(n+f)p_{\mathsf{sample}}}-1\ge 0,
\end{align*}
then
\begin{align*}
p_{\mathsf{eq}}
\le
\exp\!\left(
-\frac{\alpha\delta_v^2(n+f)p_{\mathsf{vote}}}{2+\delta_v}
-\frac{\delta_s^2(n+f)p_{\mathsf{sample}}}{2(2+\delta_s)}
\right).
\end{align*}
\end{theorem}
\begin{proof}
By Corollary~\ref{corollary:vote-quorum-two-blocks}, the probability that both blocks pass the required probabilistic voting step is at most $\Phi((n+f)/2)^{2\alpha}$ when $n-f$ is even.
By Lemma~\ref{lemma:balanced-signature-share-split}, even after both blocks pass the voting step, they can both obtain quorum certificates only if each block receives at least $q$ signature shares.
This requires
$
b+\left\lfloor\frac{a}{2}\right\rfloor\ge q.
$
The voting samples and the next-epoch sample $S_{e+1}$ are independently selected, so the probabilities multiply.
This gives
\begin{align*}
p_{\mathsf{eq}}
\le
\Phi\!\left(\frac{n+f}{2}\right)^{2\alpha}
\cdot
\Pr\!\left(b+\left\lfloor\frac{a}{2}\right\rfloor\ge q\right).
\end{align*}
The Chernoff expression follows by substituting the bounds from Lemma~\ref{lemma:chernoff-vote} and Lemma~\ref{lemma:chernoff-signatures}.
\end{proof}

\subsubsection*{When Conflicts Occur}\label{subsec:conflicts}
Recall from Definition~\ref{def:confilict} that two blocks~$B$ and~$B'$ are \emph{conflicting} if no correct process can become a candidate to vote for both.
This subsection identifies the structural reasons such pairs of blocks arise in \protocolName, which is the basis for the safety analysis in the rest of the section: to bound the probability of a safety violation, we consider two chains of $\kappa$ certified blocks whose first blocks have the same height, \emph{count the number of conflicting block pairs across these two chains}, and then compose the per-pair bound from Theorem~\ref{thm:01} accordingly.
The larger the number of conflicting pairs, the stronger the resulting safety bound.
Two mechanisms in \protocolName\ result in conflicting blocks.

\noindent\textit{(1) Locking on the parent block.}
When a correct process becomes a candidate to vote for a block~$B$ proposed in epoch~$e$, it \emph{locks} on the parent block of~$B$.
After locking on a block~$L$, the correct process will not become a candidate for any block whose height is less than that of~$L$, nor for any block at the same height as~$L$ that was proposed in an epoch smaller than~$L$.
Hence, a block whose lock (i.e., whose parent) has a height smaller than that of~$L$, or equal to that of~$L$ but proposed in an epoch smaller than that of~$L$, is a conflicting block for~$B$.

\noindent\textit{(2) Byzantine leader equivocation.}
A Byzantine leader of epoch~$e$ may propose two distinct blocks~$B$ and~$B'$ in the same epoch.
However, if a correct process votes for a block proposed in epoch~$e$, it does not later vote on any other block proposed in the same epoch.
Hence, these two blocks are also conflicting in the sense of Definition~\ref{def:confilict}.

\begin{algorithm}[!t]
\caption{Client-side of QScale.}
\label{alg:client}
\begin{algorithmic}[1]
\STATEx{\hspace{-1.5em}\textbf{Input constants:} $p_{\mathsf{lock}}, p_{\mathsf{eq}}$}
\vspace{0.2em}

\STATEx{\hspace{-1.5em}\function $\fn{submit}(\tx,\varepsilon_\tx)$}
    \STATE{send $\tx$ to a process or a set of processes}
    \LOOP
        \STATE{$\textsf{ch} \gets$ the longest certified chain from a process or from a set of processes}
        \STATE{let $B \in \textsf{ch}$ be the block containing $\tx$, if such a block exists}
        \IF{$B \neq \perp$}
            \IF{$\fn{safety}(\textsf{ch},B) \le \varepsilon_\tx$}
                \STATE{\textbf{return} \emph{true}}
            \ENDIF
        \ENDIF
    \ENDLOOP
\vspace{0.2em}

\STATEx{\hspace{-1.5em}\textbf{function} $\fn{safety}(\textsf{ch},B)$}
\STATE{let $C$ be the part of the chain $\chain$ starting from a block $B'$ s.t. $\hash{B'} = B.\parentHash$ and ending at the block of maximum height in $\chain$}
\STATE{$C' \gets [\,C[0]\,]$}
\STATE{$h \gets C[0].\height+1$}
\STATE{$m \gets \fn{nblocks}(C)$}
\STATE{\textbf{return} $\fn{search}(C,1,C',h,m)$}
\vspace{0.2em}

\STATEx{\hspace{-1.5em}\textbf{function} $\fn{search}(C,e,C',h,m)$}
\IF{$\fn{nblocks}(C') = m$}\label{line:search:start}
    \STATE{\textbf{return} $\fn{evaluate}(C,C',h,m)$}
\ENDIF
\IF{$e = |C|$}
    \STATE{\textbf{return} $\fn{evaluate}(C,C',h,m)$}
\ENDIF

\IF{$C[e]=\perp$}
    \STATE{$C'_1 \gets C'$ with $(h,C[0].\epoch+e)$ appended}
    \STATE{\textbf{return} $\fn{search}(C,e+1,C'_1,h+1,m)$}
\ELSE
    \STATE{$C'_0 \gets C'$ with $\perp$ appended}
    \STATE{$x_0 \gets \fn{search}(C,e+1,C'_0,h,m)$}
    \STATE{$C'_1 \gets C'$ with $(h,C[0].\epoch+e)$ appended}
    \STATE{$x_1 \gets \fn{search}(C,e+1,C'_1,h+1,m)$}
    \STATE{\textbf{return} $\max\{x_0,x_1\}$}
\ENDIF\label{line:search:end}

\newcounter{savedALGlineClient}
\makeatletter
\setcounter{savedALGlineClient}{\value{ALG@line}}
\makeatother

\end{algorithmic}
\end{algorithm}

\begin{algorithm}[!t]
\caption{Auxiliary functions for client-side protocol.}
\label{alg:client:recursive:prob:aux}
\begin{algorithmic}[1]
\makeatletter
\setcounter{ALG@line}{\value{savedALGlineClient}}
\makeatother

\STATEx{\hspace{-1.5em}\textbf{struct} block: $\height$, $\epoch$ \hfill ($\perp$ denotes a missing block)}
\vspace{0.2em}

\STATEx{\hspace{-1.5em}\textbf{relation} $B \succeq B'$}
\STATE{\textbf{return} $B'=\perp \,\lor\, B.\height>B'.\height \,\lor\, \bm{(}B.\height=B'.\height \land B.\epoch\ge B'.\epoch\bm{)}$}
\vspace{0.2em}

\STATEx{\hspace{-1.5em}\textbf{function} $\fn{nblocks}(C)$}
\STATE{\textbf{return} the number of entries $C[j]$ such that $C[j]\neq\perp$}
\vspace{0.2em}

\STATEx{\hspace{-1.5em}\textbf{function} $\fn{evaluate}(C,C',h,m)$}
\STATE{$(\widehat C,\widehat C') \gets \fn{complete}(C,C',h,m)$}\label{line:evaluate:start}
\STATE{$c_{\mathsf{lock}} \gets \fn{lock\_conflicts}(\widehat C,\widehat C')$}
\STATE{$c_{\mathsf{eq}} \gets \fn{equivocation\_conflicts}(\widehat C,\widehat C')$}
\STATE{\textbf{return} $p_{\mathsf{lock}}^{c_{\mathsf{lock}}}\cdot p_{\mathsf{eq}}^{c_{\mathsf{eq}}}$}\label{line:evaluate:end}
\vspace{0.2em}

\STATEx{\hspace{-1.5em}\textbf{function} $\fn{complete}(C,C',h,m)$}
\STATE{$\widehat C \gets C$;\, $\widehat C' \gets C'$}
\WHILE{$\fn{nblocks}(\widehat C') < m$}
    \STATE{append $(h,\widehat C[0].\epoch+|\widehat C'|)$ to $\widehat C'$}
    \STATE{append $\perp$ to $\widehat C$}
    \STATE{$h \gets h+1$}
\ENDWHILE
\WHILE{$|\widehat C'| < |\widehat C|$}
    \STATE{append $\perp$ to $\widehat C'$}
\ENDWHILE
\STATE{\textbf{return} $(\widehat C,\widehat C')$}
\vspace{0.2em}

\STATEx{\hspace{-1.5em}\textbf{function} $\fn{lock}(C, k)$}
\STATE{\textbf{return} $C[j]$ for the greatest $j \le k$ with $C[j] \neq \perp$, or $\perp$ if no such $j$ exists}
\vspace{0.2em}

\STATEx{\hspace{-1.5em}\textbf{function} $\fn{lock\_conflicts}(C,C')$}
\STATE{$c \gets 0$}
\FORALL{$e$ such that $C[e]\neq \perp$ and $e>0$}
    \STATE{$B \gets C[e]$;\, $L \gets \fn{lock}(C,e{-}1)$}
    \FORALL{$e'$ such that $C'[e']\neq \perp$ and $e'>0$}
        \STATE{$B' \gets C'[e']$;\, $L' \gets \fn{lock}(C',e'{-}1)$}
        \IF{$B.\epoch = B'.\epoch$}
            \STATE{continue \hfill same-epoch conflicts are counted separately}
        \ENDIF
        \IF{$B.\epoch < B'.\epoch$ and not $(L' \succeq L)$}
            \STATE{$c \gets c+1$}
        \ELSIF{$B'.\epoch < B.\epoch$ and not $(L \succeq L')$}
            \STATE{$c \gets c+1$}
        \ENDIF
    \ENDFOR
\ENDFOR
\STATE{\textbf{return} $c$}
\vspace{0.2em}

\STATEx{\hspace{-1.5em}\textbf{function} $\fn{equivocation\_conflicts}(C,C')$}
\STATE{$c \gets 0$}
\FOR{$e=1,\ldots,\max\{|C|,|C'|\}-1$}
    \IF{$C[e]\neq\perp$ and $C'[e]\neq\perp$}
        \STATE{$c \gets c+1$}
    \ENDIF
\ENDFOR
\STATE{\textbf{return} $c$}
\end{algorithmic}
\end{algorithm}

\subsubsection*{Client-side finalization protocol}
Algorithm~\ref{alg:client} specifies the client-side finalization protocol.
A client submits a transaction~$\tx$ and repeatedly obtains the longest certified chain it can observe.
Once the client finds a block~$B$ containing~$\tx$, it calls $\fn{safety}(\textsf{ch},B)$ to compute an upper bound on the probability that the observed chain can be replaced by a hidden conflicting chain.
The function $\fn{safety}$ extracts the suffix~$C$ of the observed chain starting from the parent of~$B$, initializes a hidden chain~$C'$ from the same parent, and calls $\fn{search}$ to explore possible hidden chains.
If an epoch is missing from~$C$, the hidden chain is allowed to grow in that epoch; otherwise, the hidden chain may either skip the epoch or place a conflicting block in it.
Algorithm~\ref{alg:client:recursive:prob:aux} provides the auxiliary functions used to complete each hidden chain and evaluate it by counting lock-induced conflicts and same-epoch equivocation conflicts.
The client takes the maximum value over all hidden chains and finalizes~$\tx$ once this value is at most~$\varepsilon_\tx$.

\myparagraph{Searching for the worst-case hidden chain (lines~\ref{line:search:start}-\ref{line:search:end}).}
The recursive function $\fn{search}$ explores hidden chains consistent with the observed epoch pattern and computes the maximum value of the objective over them. 
The recursion terminates when the hidden chain reaches the same size as the observed chain or when all epochs have been processed, at which point the current hidden chain is evaluated. 
If an epoch contains no observed certified block, the hidden chain may add a block in that epoch; otherwise, it may either skip the epoch or introduce a conflicting block. 
By exploring these valid continuations and taking the maximum value over all branches, the procedure identifies the hidden chain that yields the largest upper bound on the probability of replacing the observed chain.

\myparagraph{Evaluating a hidden chain (lines~\ref{line:evaluate:start}-\ref{line:evaluate:end}).}
The function $\fn{evaluate}$ computes an upper bound on the probability that the completed hidden chain replaces the observed chain. 
It first calls $\fn{complete}$ to ensure that the observed and hidden chains can be compared over the same epoch range (i.e., the hidden chain is extended until it has the same number of blocks as the observed chain, and missing epochs are filled with $\perp$ entries).
It then counts the number $c_{\mathsf{lock}}$ of lock-induced conflicts and the number $c_{\mathsf{eq}}$ of same-epoch equivocation conflicts. 
The former is obtained by considering pairs of blocks from the observed and hidden chains proposed in different epochs and checking whether the latter block is incompatible with the lock induced by the earlier one. 
The latter is obtained by counting the epochs in which both chains contain a block, since such a pair corresponds to two conflicting blocks proposed in the same epoch.
The returned value $p_{\mathsf{lock}}^{c_{\mathsf{lock}}}\cdot p_{\mathsf{eq}}^{c_{\mathsf{eq}}}$ is an upper bound on the probability that this hidden chain replaces the observed chain.

\subsubsection*{Examples}
We conclude this section with examples illustrating how the client-side finalization procedure evaluates different observed patterns. In all examples, we use the following parameter setting: $n=500$, $f=75$, $q=67$, $p_{\mathsf{sample}}=1.5q/n$, and $\alpha=3$. 
The probability $p_{\mathsf{lock}}$ of certifying two lock-conflicting blocks is computed using Theorem~\ref{thm:01}. 
The probability $p_{\mathsf{eq}}$ of certifying two distinct blocks proposed in the same epoch is computed using Theorem~\ref{theorem:two-same-epoch-blocks-certified}. 
In each example, $C$ denotes the certified chain observed by the client, while $C'$ denotes a worst-case hidden chain considered by the adversary. The value returned by the finalization procedure depends on the number of conflicts between $C$ and $C'$: same-epoch conflicts contribute a factor $p_{\mathsf{eq}}$, while lock-induced conflicts contribute a factor $p_{\mathsf{lock}}$.

In the first example, shown in Fig.~\ref{fig:gap1}, there is a gap between every two successive certified blocks in the observed chain. 
The hidden chain can use these gaps to grow without creating either same-epoch equivocation conflicts or lock-induced conflicts. 
Therefore, the client-side procedure returns the trivial upper bound $1$, meaning that this observed pattern does not provide a useful safety guarantee.

\begin{figure}[!t]
    \centering
    \includegraphics[scale=0.44]{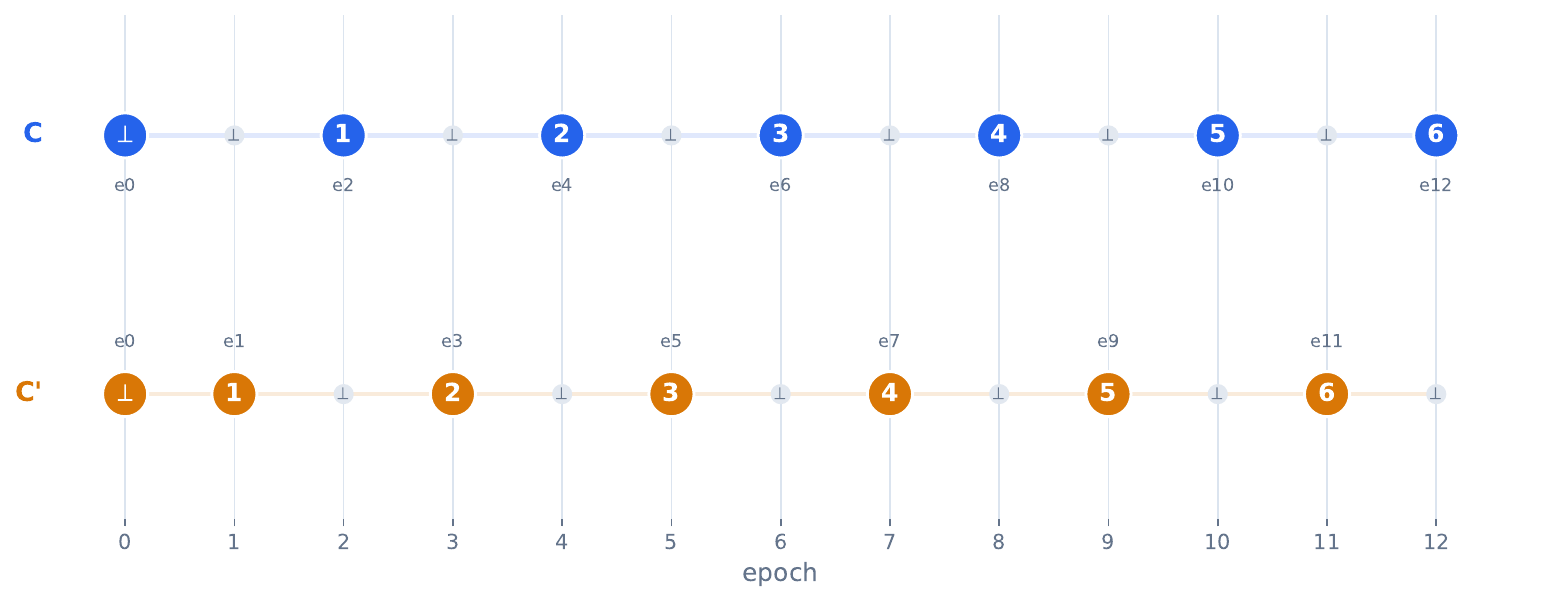}
    \caption{Example of the client-side safety computation, with the observed chain~$C$ drawn in blue and the worst-case hidden chain~$C'$ drawn in orange.
    Since every two successive certified blocks in the observed chain are separated by a gap, the hidden chain can grow without creating any counted conflict, so the safety-violation probability is bounded by~$1$.}
    \label{fig:gap1}
\end{figure}

The second example, shown in Fig.~\ref{fig:gap2}, illustrates a pattern in which the observed chain has several certified blocks proposed in consecutive epochs. To replace this chain, the hidden chain must first contain a block proposed in the same epoch as one of the observed blocks, creating a same-epoch equivocation conflict. 
After this point, the hidden chain grows in later epochs, but the consecutive certified blocks in the observed chain force three lock-induced conflicts. 
Consequently, the client-side safety procedure bounds the probability of replacement by $p_{\mathsf{eq}}\cdot p_{\mathsf{lock}}^3$, which gives $2^{-101}$ under the parameters used in the example.

\begin{figure}[!t]
    \centering
    \includegraphics[scale=0.44]{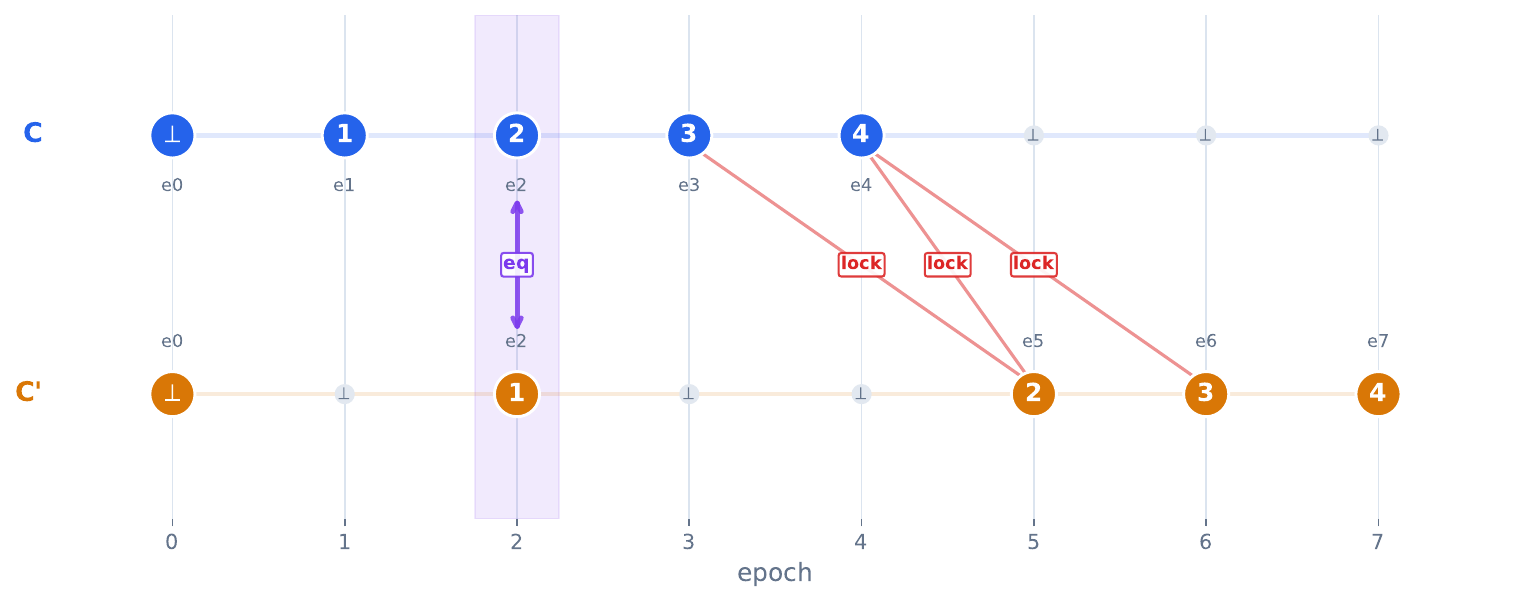}
    \caption{Example of the client-side safety computation, with the observed chain~$C$ drawn in blue and the worst-case hidden chain~$C'$ drawn in orange.
    The hidden chain conflicts with the observed chain through one same-epoch equivocation conflict and three lock-induced conflicts, yielding a safety-violation probability bounded by~$p_{\mathsf{eq}}\cdot p_{\mathsf{lock}}^3 = 2^{-101}$.}
    \label{fig:gap2}
\end{figure}

The third example, shown in Fig.~\ref{fig:gap3}, illustrates a pattern in which the observed chain contains two sequences of certified blocks proposed in consecutive epochs, separated by gaps. 
The hidden chain uses the gaps to grow, and no block of the hidden chain is proposed in the same epoch as a block of the observed chain. 
Hence, there is no same-epoch equivocation conflict. 
However, the consecutive certified blocks in the observed chain force six lock-induced conflicts with the hidden chain. 
Therefore, the client-side safety procedure bounds the probability of replacement by $p_{\mathsf{lock}}^6$, which gives $2^{-114}$ under the same parameters. 

\begin{figure}[!t]
    \centering
    \includegraphics[scale=0.44]{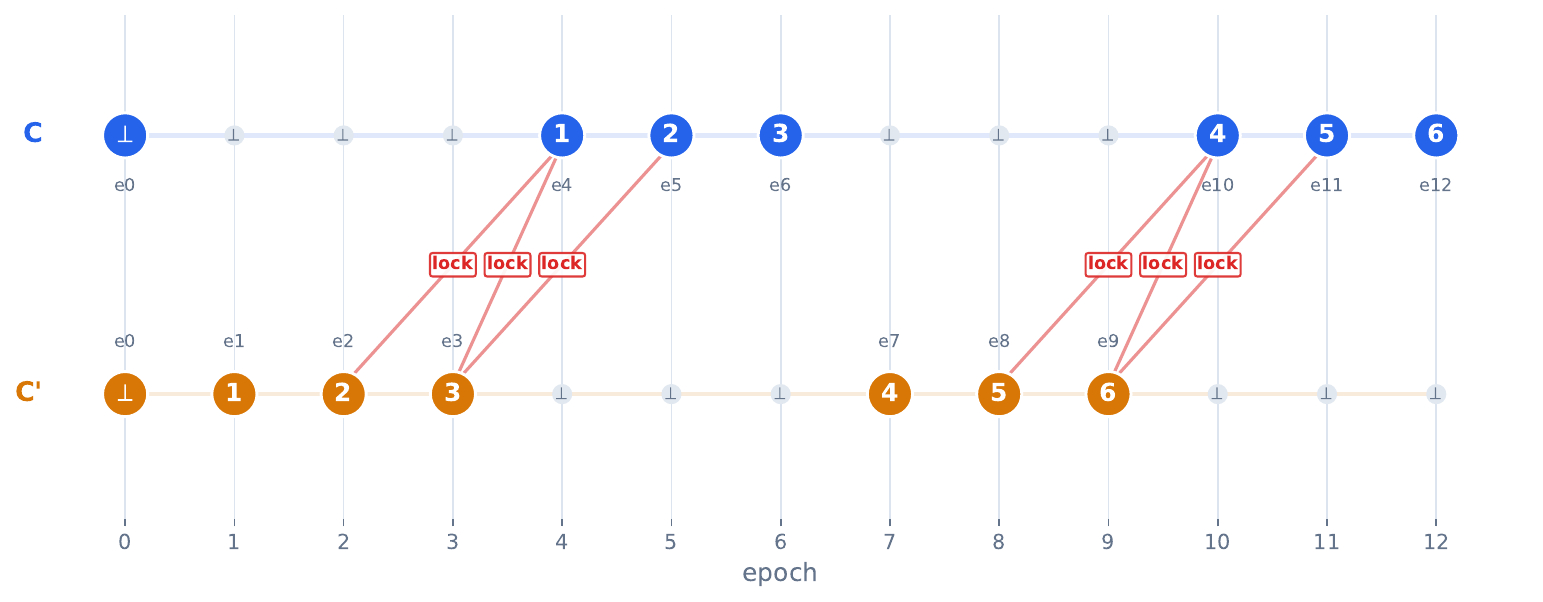}
    \caption{Example of the client-side safety computation, with the observed chain~$C$ drawn in blue and the worst-case hidden chain~$C'$ drawn in orange. 
    The hidden chain creates no same-epoch equivocation conflict, but the two sequences of consecutive certified blocks in the observed chain force six lock-induced conflicts. Hence, the safety-violation probability is bounded by~$p_{\mathsf{lock}}^6 = 2^{-114}$.}
    \label{fig:gap3}
\end{figure}

%% file: liveness.tex
\section{Liveness Analysis}
\label{sec:app:liveness}

This section presents the liveness analysis of \protocolName.
We first establish probabilistic bounds showing that, when a correct leader proposes a valid block, sufficiently many correct processes receive the proposal, become candidates, and participate in the indirect certification mechanism with positive probability. We then combine these bounds with the propagation results from Appendix~\ref{sec:propagation} to show that a correct leader can become synchronized with the highest-certified block known to correct processes. 
Finally, we use these ingredients to prove that, once a transaction is provided to all correct processes, there is a nonzero probability of progress in every sufficiently long window of epochs; repeating this argument over infinitely many windows implies that the transaction is eventually included in a certified chain finalized by a client with probability~$1$.

\begin{lem}\label{lem:percentage:0}
In any random sample where each process is included independently with probability $\prSample$, there are at least $(1-\varphi)(1 - \epsilon) n \prSample$ correct processes with probability at least
$1- \exp\bm{(}-\varphi^2(1 - \epsilon)n \prSample/2\bm{)}$, where $\varphi\in(0,1)$.
\end{lem}
\begin{proof}
Let $C$ denote the set of correct processes, so $|C| = (1 - \epsilon)n$.
Besides, let $C_1 \subseteq C$ denote the subset of correct processes included in a random sample.
Since each correct process is included in the random sample independently with probability $\prSample$, $|C_1| \sim \mathrm{Bin}(|C|,\prSample)$.
Therefore, the expected number of correct processes included in the random sample is
\begin{align*}
\mathbb{E}[|C_1|] = (1 - \epsilon)n \prSample.
\end{align*}
Using the Chernoff bound~\ref{ineq:chernof2}, for any $\varphi\in(0,1)$, we have:
\begin{align*}
\Pr\left(|C_1| < (1-\varphi)\mathbb{E}[|C_1|]\right)
\le \exp\left(-\frac{\varphi^2\mathbb{E}[|C_1|]}{2} \right)
=
\exp\left(-\frac{\varphi^2(1 - \epsilon)n \prSample }{2} \right).
\end{align*}
\end{proof}

\begin{lem}\label{thm:percentage}
Fix $\varphi\in(0,1)$, and let
$
\mu = (1 - \epsilon)n \cdot
\left(1 - (1-\prSample)^{(1-\varphi)(1-\epsilon)n\prSample}\right).
$
If a leader is correct, then, for every
$
a < 1-\exp\bm{(}-(1-\varphi)(1 - \epsilon)n\prSample^2\bm{)},
$
at least an $a$-fraction of correct processes become candidates to vote for the leader's proposal with probability at least
$
1
-
\exp\left(-\frac{\delta^2\mu}{2}\right)
-
\exp\left(-\frac{\varphi^2(1 - \epsilon)n\prSample}{2}\right),
$
where
$
\delta = 1 - a(1-\epsilon)n/\mu.
$
\end{lem}
\begin{proof}
In any epoch~$e$, a process becomes a candidate to vote for the leader's proposal if it receives that proposal from some process in the first-layer sample selected by the leader.
Let $C$ denote the set of correct processes, so $|C| = (1 - \epsilon)n$.
Besides, let $C_1 \subseteq C$ denote the subset of correct processes selected by the leader in the first-layer sample.
By Lemma~\ref{lem:percentage:0}, $|C_1| < (1-\varphi)(1 - \epsilon)n \prSample$ with probability at most
$\exp\!\left(-\frac{\varphi^2(1 - \epsilon)n\prSample}{2}\right)$,
where $\varphi\in(0,1)$.

Each correct process in $C_1$ selects a second-layer sample by including each process independently with probability $\prSample$.
Consider a correct process~$i \in C_1$.
The probability that a process~$j$ is not included in the second-layer sample of $i$ is equal to $1 - \prSample$, and since these samples are independent across all members of $C_1$, the probability that $j$ is not included in any second-layer sample is equal to $(1 - \prSample)^{|C_1|}$.
Let $I_j$ be the indicator random variable representing that $j$ is included in any second-layer sample and define $X = \sum_{j\in C}I_j$, i.e., $X$ is the number of correct processes that are included in at least one second-layer sample selected by a correct process in the first-layer sample.
We must compute a lower bound for $\Pr(X \ge a\cdot (1-\epsilon)n)$.
Using the linearity of expectation, we have:
\begin{align*}
    \mathbb{E}[X]
    = (1 - \epsilon)n \cdot \left( 1 - (1 - \prSample)^{|C_1|} \right).
\end{align*}
Let
\[
    \mu = (1 - \epsilon)n \cdot
    \left( 1 - (1 - \prSample)^{(1-\varphi)(1 - \epsilon)n \prSample} \right),
\]
and assume that $a < \mu / ((1 - \epsilon)n)$.
Further, let $E$ be the event that $|C_1| \ge (1-\varphi)(1 - \epsilon)n \prSample$.
By applying the Chernoff bound~\ref{ineq:chernof2}, we have:
\begin{align}\label{eq:lem:a}
    \Pr(X \le a\cdot (1-\epsilon)n \mid E)
    \le \exp\left(-\frac{\delta^2\mu}{2}\right),
\end{align}
where $\delta = 1 - a\cdot (1-\epsilon)n/\mu$.
Note that because $a < \mu / ((1 - \epsilon)n)$, it follows that $\delta \in (0,1)$.
Also, since $(1-\prSample)^m \le \exp(-m\prSample)$ for every $m>0$, it is sufficient to require
\[
    a < 1-\exp\left(-(1-\varphi)(1 - \epsilon)n\prSample^2\right).
\]
Consequently, we have:
\begin{align*}
    &\Pr(X \le a\cdot (1-\epsilon)n)
    \\& =
    \Pr(X \le a\cdot (1-\epsilon)n \mid E) \Pr(E)
    +
    \Pr(X \le a\cdot (1-\epsilon)n \mid \bar{E}) \Pr(\bar{E})
    \\& \le
    \Pr(X \le a\cdot (1-\epsilon)n \mid E)
    +
    \Pr(\bar{E})
    \\& \le
    \exp\left(-\frac{\delta^2\mu}{2}\right)
    +
    \exp\left(-\frac{\varphi^2(1 - \epsilon)n\prSample}{2}\right).
\end{align*}
Thus,
\begin{align*}
    \Pr(X \ge a\cdot (1-\epsilon)n)
    \ge
    1
    -
    \exp\left(-\frac{\delta^2\mu}{2}\right)
    -
    \exp\left(-\frac{\varphi^2(1 - \epsilon)n\prSample}{2}\right).
\end{align*}
\end{proof}

\begin{cor}\label{cor:percentage}
Suppose $\prSample \ge 2/\sqrt{n}$.
If a leader is correct, then at least an $0.6321$-fraction of correct processes become candidates to vote for the leader's proposal with a probability of at least
$
1
-
\exp\left(-\frac{\delta^2\mu}{2}\right)
-
\exp\left(-\frac{(1 - \epsilon)n\prSample}{8}\right),
$
where
$
\mu =
(1 - \epsilon)n \cdot
\left( 1 - (1 - \prSample)^{(1-\epsilon)n\prSample/2} \right)
$
and $\delta = 1 - 0.6321(1-\epsilon)n/\mu$.
\end{cor}
\begin{proof}
This corollary follows from Lemma~\ref{thm:percentage} by setting $\varphi=1/2$ and $a=0.6321$.
Since $\epsilon < 1/2$ and $\prSample \ge 2/\sqrt{n}$, we have
$
1-\exp\!\left(-(1-\varphi)(1-\epsilon)n\prSample^2\right)
>
1-\exp(-1).
$
Hence, choosing $a = 0.6321 < 1-\exp(-1)$ satisfies the required condition on $a$.
\end{proof}

\begin{lem}\label{lem:qc}
Suppose the leader of epoch~$e-1$ is correct, it proposes a valid value, and every correct process that receives the proposal evaluates it as valid.
Then, with probability at least
\begin{align*}
P_{\mathsf{icert}}
=
P_{\mathsf{leader}}
\cdot
\left(1-\exp\left(-\frac{\theta_v^2\mu_v}{2}\right)\right)^{\alpha+1}
\cdot
\left(1-\exp\left(-\frac{\theta_s^2\mu_s}{2}\right)\right),
\end{align*}
the leader of epoch~$e$ receives at least $q$ votes for the proposed block and an indirect quorum certificate is created for that block, where
$
P_{\mathsf{leader}}
=
1
-
\exp\left(-\frac{\delta^2\mu}{2}\right)
-
\exp\left(-\frac{(1-\epsilon)n\prSample}{8}\right)
$,
$\mu_v = 0.6321(1-\epsilon)n\prSample$, 
$\theta_v = 1 - q/\mu_v$,
$\mu_s = (1-\epsilon)n\prSample$, and
$\theta_s = 1 - q/\mu_s$.
\end{lem}
\begin{proof}
In a given epoch~$e-1$, since Byzantine processes might remain silent, at least $n-f = (1-\epsilon)n$ processes out of the total $n$ processes participate in the protocol.
By assumption, the leader of epoch~$e-1$ is correct.
Therefore, by Corollary~\ref{cor:percentage}, an $0.6321$-fraction of correct processes become candidates to vote in epoch~$e-1$ with probability at least $P_{\mathsf{leader}}$.

Let $E$ denote the event that at least $0.6321(1-\epsilon)n$ correct processes become candidates.
Conditioned on $E$, consider first the sample of candidates that send vote messages directly to the leader of epoch~$e$.
Each correct candidate sends such a vote with probability $\prSample$.
Let $X_0$ denote the number of correct votes received by the leader of epoch~$e$.
We have:
\[
    \mathbb{E}[X_0 \mid E] \ge \mu_v = 0.6321(1-\epsilon)n\prSample.
\]
By applying the Chernoff bound~\ref{ineq:chernof2}, we have:
\[
    \Pr(X_0 < q \mid E)
    \le
    \exp\left(-\frac{\theta_v^2\mu_v}{2}\right),
\]
where $\theta_v = 1-q/\mu_v$.
Next, consider the $\alpha$ random samples used for indirect certification.
For each $\ell \in \{1,\dots,\alpha\}$, let $X_\ell$ denote the number of correct candidates included in the $\ell$-th sample.
Conditioned on $E$, each $X_\ell$ has expectation at least $\mu_v$.
Again, by the Chernoff bound, for every $\ell \in \{1,\dots,\alpha\}$,
\[
    \Pr(X_\ell < q \mid E)
    \le
    \exp\left(-\frac{\theta_v^2\mu_v}{2}\right).
\]
Since these random samples are selected independently, the probability that the leader receives at least $q$ votes and that each sample $\ell$ contains at least $q$ correct voters is at least
\[
    \left(1-\exp\left(-\frac{\theta_v^2\mu_v}{2}\right)\right)^{\alpha+1}.
\]

It remains to show that enough correct processes in $S_e$ send signature shares.
Let $Y$ denote the number of correct processes in $S_e$.
Since each correct process is included in $S_e$ independently with probability $\prSample$, we have
\[
    \mathbb{E}[Y] = \mu_s = (1-\epsilon)n\prSample.
\]
By the Chernoff bound,
\[
    \Pr(Y < q)
    \le
    \exp\left(-\frac{\theta_s^2\mu_s}{2}\right),
\]
where $\theta_s = 1-q/\mu_s$.
Therefore, with probability at least
\[
    1-\exp\left(-\frac{\theta_s^2\mu_s}{2}\right),
\]
at least $q$ correct processes belong to $S_e$.

If each of the $\alpha$ samples contains at least $q$ correct voters, then every correct process in $S_e$ receives at least $q$ matching votes from each of the $\alpha$ samples.
Hence, every correct process in $S_e$ sends a signature share.
If at least $q$ correct processes belong to $S_e$, then at least $q$ signature shares are sent, and a process in $S_e$ can aggregate them into a quorum certificate.
Combining these independent events gives the desired lower bound.
\end{proof}

\kartik{The above lemma says: if a leader is honest, a QC will be formed. This requires at least two things. First, an honest leader has sufficient information to actually make a correct proposal. Second, we need to argue that sufficiently many parties will indeed vote based on the history of messages. For instance, on line 17 in Alg 2, the can\_vote function is triggered only if some conditions hold. We haven't argued that they hold. More generally, the liveness proof should argue that we will commit any submitted transaction when there is a sequence of good and bad leaders.}

\kartik{The safety proof should argue that if some honest party has committed B at slot x, no honest party will commit B' neq B at slot x. This requires arguing quorum intersections with high probability, potentially across slots.}

The following progress argument applies once correct processes already know the highest certified block.
We say that a certified block $B$ is \emph{highest among correct processes} at the beginning of an epoch if $B$ has the greatest height among the certified blocks stored by correct processes at that time.
A leader is \emph{synchronized} with such a block if it has received both $B$ and a valid quorum certificate for $B$.

\begin{thm}\label{thm:liveness:synchronized}
Suppose there are $\kappa+1$ consecutive epochs $e_1,\dots,e_{\kappa+1}$, each with a correct leader.
Suppose further that, at the beginning of epoch~$e_1$, all correct processes have received a certified block~$B$ that is highest among the certified blocks known to correct processes.
Then, at least one block is committed in epoch~$e_{\kappa+1}$ with probability at least
$
\left(P_{\mathsf{prop}}\cdot P_{\mathsf{icert}}\right)^\kappa,
$
where
$
P_{\mathsf{prop}}
=
1-\frac{n-1}{3\prProp}
$
and $P_{\mathsf{icert}}$ is the probability given in Lemma~\ref{lem:qc}.
\end{thm}
\begin{proof}
Suppose there are $\kappa+1$ consecutive epochs $e_1,\dots,e_{\kappa+1}$, each with a correct leader.
By assumption, at the beginning of epoch~$e_1$, all correct processes have received a certified block~$B$ that is highest among the certified blocks known to correct processes.
Hence, the leader of epoch~$e_1$ proposes a block~$B_1$ that extends~$B$.
Since all correct processes have received $B$ and its quorum certificate, every correct process that receives $B_1$ evaluates the proposed block as valid.
In particular, the parent of $B_1$ is certified, all predecessors of $B_1$ are certified, and $B_1$ is consistent with the locks of correct processes.

By Lemma~\ref{lem:qc}, the leader of epoch~$e_2$ receives at least $q$ votes for $B_1$, and an indirect quorum certificate is created for $B_1$, with probability at least $P_{\mathsf{icert}}$.
Thus, block~$B_1$ is certified with probability at least $P_{\mathsf{icert}}$.

In epoch~$e_2$, the leader proposes block~$B_2$.
Since the leader of epoch~$e_2$ receives at least $q$ votes for $B_1$, it proposes $B_2$ with parent $B_1$.
Moreover, since an indirect quorum certificate is created for $B_1$ and the propagation sub-protocol disseminates proposal blocks and quorum certificates, the probability that correct processes receive $B_1$ and its certificate in time to vote for $B_2$ is at least
$
P_{\mathsf{prop}}.
$
Thus, every correct process that receives $B_2$ evaluates it as valid.
Again, by Lemma~\ref{lem:qc}, the leader of epoch~$e_3$ receives at least $q$ votes for $B_2$, and an indirect quorum certificate is created for $B_2$, with probability at least $P_{\mathsf{icert}}$.

For each subsequent block, the same argument applies.
Namely, after block~$B_i$ is certified, for every $i\in{1,\dots,\kappa-1}$, the propagation sub-protocol disseminates $B_i$ and its quorum certificate with probability at least $P_{\mathsf{prop}}$, allowing the leader of the next epoch to propose $B_{i+1}$ extending $B_i$.
Then, by Lemma~\ref{lem:qc}, block~$B_{i+1}$ is certified with probability at least $P_{\mathsf{icert}}$.
Accordingly, block~$B_\kappa$ is certified in epoch~$e_{\kappa+1}$ with probability at least
$
P_{\mathsf{icert}}^\kappa \cdot P_{\mathsf{prop}}^{\kappa-1}.
$
Consequently, block~$B_1$ is committed in epoch~$e_{\kappa+1}$ with at least the same probability.
Since $P_{\mathsf{prop}}\le 1$, this probability is at least
$
\left(P_{\mathsf{prop}}\cdot P_{\mathsf{icert}}\right)^\kappa.
$
\end{proof}

The previous theorem assumes that all correct processes already know the highest certified block at the beginning of the first epoch.
We now bound the number of epochs needed for a correct leader to receive such a block and its certificate through dissemination.

\begin{lem}\label{lem:leader:sync}
Fix an epoch~$e$ and an integer $r\ge 1$.
Suppose that, at the beginning of epoch~$e$, at least one correct process stores a certified block~$B$ that is highest among correct processes.
Suppose also that the leader remains correct during epochs $e,\dots,e+r-1$.
Then, by the beginning of epoch~$e+r$, all correct processes have received $B$ and its certificate with probability at least
$
P_{\mathsf{sync}}(r)
=
1-\left(1-P_{\mathsf{prop}}\right)^r.
$
\end{lem}
\begin{proof}
Consider one epoch among $e,\dots,e+r-1$.
By Theorem~\ref{thm:all:receive}, if a correct process disseminates $B$ and its quorum certificate during that epoch, then all correct processes receive them within the epoch with probability at least
$
P_{\mathsf{prop}}
=
1-\frac{n-1}{3\prProp}.
$
Assuming that the random choices used by the propagation sub-protocol are independent across epochs, the probability that dissemination fails in all $r$ epochs is at most
$
(1-P_{\mathsf{prop}})^r.
$
Therefore, the probability that all correct processes receives $B$ and its certificate by the beginning of epoch~$e+r$ is at least
$
1-(1-P_{\mathsf{prop}})^r.
$
\end{proof}

\begin{thm}\label{thm:liveness:with-sync}
Fix an integer $r\ge 1$.
Suppose that, at the beginning of epoch~$e_1$, at least one correct process stores a certified block~$B$ that is highest among correct processes.
Suppose further that the leader remains correct during the consecutive epochs
$
e_1,\dots,e_{r+\kappa+1}
$.
Then, at least one block is committed by epoch~$e_{r+\kappa+1}$ with probability at least
$
P_{\mathsf{sync}}(r)
\cdot
\left(P_{\mathsf{prop}}\cdot P_{\mathsf{icert}}\right)^\kappa.
$
\end{thm}
\begin{proof}
By Lemma~\ref{lem:leader:sync}, all correct processes receives the highest certified block~$B$ and its quorum certificate by the beginning of epoch~$e_{r+1}$ with probability at least $P_{\mathsf{sync}}(r)$.
Conditioned on this event, the leader is synchronized with the highest certified block known to correct processes.
Now consider the next $\kappa+1$ epochs, namely
$
e_{r+1},\dots,e_{r+\kappa+1}.
$
By Theorem~\ref{thm:liveness:synchronized}, at least one block is committed by epoch~$e_{r+\kappa+1}$ with probability at least
$
P_{\mathsf{icert}}^\kappa \cdot P_{\mathsf{prop}}^{\kappa-1}.
$
Multiplying this probability by the probability of the synchronization event gives
$
P_{\mathsf{sync}}(r)
\cdot
P_{\mathsf{icert}}^\kappa
\cdot
P_{\mathsf{prop}}^{\kappa-1}.
$
Since $P_{\mathsf{prop}}\le 1$, this is at least
$
P_{\mathsf{sync}}(r)
\cdot
\left(P_{\mathsf{prop}}\cdot P_{\mathsf{icert}}\right)^\kappa.
$
\end{proof}

\begin{thm}\label{thm:liveness:probability-one}
Every transaction submitted to all correct processes is eventually included in a committed chain with probability~$1$.
\end{thm}
\begin{proof}
Let $\tx$ be a transaction submitted to a correct process after GST.
Fix an integer $r\ge 1$, and consider disjoint windows of $r+\kappa+1$ consecutive epochs.
By Theorem~\ref{thm:liveness:with-sync}, in each such window, at least one block is committed with probability at least
$
P_{\mathsf{sync}}(r)
\cdot
\left(P_{\mathsf{prop}}\cdot P_{\mathsf{icert}}\right)^\kappa.
$
Moreover, since $\tx$ has been received by all correct processes, the first correct leader in a window proposes a block that either contains $\tx$ or extends a chain that already contains $\tx$.
Hence, when the block finalized in that window is committed, $\tx$ belongs to the committed chain.

Since the adversary is static and there is a nonzero probability that a leader is correct, there exists a constant $p_{\mathsf{live}}>0$ such that each window commits a chain containing $\tx$ with probability at least $p_{\mathsf{live}}$.
For instance, if leaders are selected independently and uniformly at random, one can take
$
p_{\mathsf{live}}
=
(1-\epsilon)^{\kappa+1}
\cdot
P_{\mathsf{sync}}(r)
\cdot
\left(P_{\mathsf{prop}}\cdot P_{\mathsf{icert}}\right)^\kappa.
$

Now consider the first $m$ disjoint windows after $\tx$ has been received by all correct processes.
By independence across epochs, the probability that none of these windows commits a chain containing $\tx$ is at most
$
(1-p_{\mathsf{live}})^m.
$
Since $p_{\mathsf{live}}>0$, we have
$
\lim_{m\to\infty}(1-p_{\mathsf{live}})^m = 0.
$
Thus, the probability that $\tx$ is never included in a committed chain is zero.
Consequently, $\tx$ is eventually included in a committed chain with probability~$1$.
\end{proof}

%% file: pipes.tex
\newcommand{\Tprop}{T_{\mathsf{prop}}}
\newcommand{\Tsoft}{T_{\mathsf{soft}}}
\newcommand{\Tcert}{T_{\mathsf{cert}}}
\newcommand{\Tep}{T_{\mathsf{ep}}}
\newcommand{\second}{\textsf{sec}\xspace}

\section{Latency Analysis}\label{app:pipes}
The Pipes model~\cite{lewispye2026pipes} is a framework for analyzing the latency and throughput of distributed protocols by explicitly modeling process bandwidth, transaction arrival rates, and network delays, rather than characterizing latency solely in terms of communication rounds.
In this section, we use this model to compare the latency of \protocolName with that of Algorand~\cite{blum2023analyzing,algorand}.

\subsection{Pipes Model}
In the Pipes model, each process maintains a single FIFO upload buffer and a single FIFO download buffer to model bandwidth-constrained communication. Communication is performed using fixed-size \emph{information parcels}, and the model tracks the flow of these parcels through the upload and download buffers of processes. This enables the analysis of both protocol latency as a function of system parameters and the corresponding \emph{latency bottleneck}, defined as the transaction arrival rate beyond which latency becomes unbounded.
The model assumes a system of $n$ processes, $\{p_1, \ldots, p_n\}$, connected through pairwise reliable, in-order channels (i.e., TCP-like channels). 
When a process sends a parcel to another process, the resulting copy placed in the upload or download buffers is called an \emph{addressed parcel}.

Each process has a \emph{bandwidth} of $S$ parcels per time slot, meaning that it can upload and download up to $S$ parcels per slot. 
We assume that upload and download rates are equal and constant, and that all processes have the same bandwidth $S$.
In addition, each process maintains a single FIFO upload buffer and a single FIFO download buffer. At the end of time slot $t$, up to $S$ parcels are removed from the upload buffer of $p_i$ and arrive at the download buffers of their intended recipients at time slot $t + \Delta$.

The Pipes model assumes a synchronous network with a message delay of $\Delta$ time slots between every pair of processes, where time is divided into discrete slots.
Every process is divided into a \emph{client process} and a \emph{consensus process}, connected by an infinite-bandwidth, zero-delay channel. 
The client processes collectively receive transactions from external clients at a total rate of $D$ parcels per slot.
In the \emph{single-sender} setting, a designated leader's client process receives all $D$ parcels per slot, while in the \emph{multi-sender} setting, each process's client receives $D/n$ parcels per slot, and all transactions are assumed unique.
A \emph{block} consists of $B$ transaction parcels together with block metadata of size $M$. 
Vote messages have size $\lambda$ parcels.
Since the available bandwidth limits the amount of block data that can be disseminated per time slot, the following feasibility conditions must hold: $D \le \frac{S}{n}$ in the single-sender setting, and $D \le S$ in the multi-sender setting.
The difference arises because, in the single-sender setting, a single process is responsible for disseminating every transaction to all $n$ processes, whereas in the multi-sender setting, this dissemination responsibility is shared among the $n$ processes.

Standard cryptographic primitives (e.g., signatures, PKI, threshold signatures, erasure coding, and a hash function) are assumed against a computationally bounded adversary, and these schemes are treated as perfect for simplicity.
Finally, the analysis is restricted to the \emph{good case}: all processes are correct, message delivery is reliable, and computation costs (signature generation/verification, erasure coding) are ignored, so the bounds derived in the model act as idealized lower bounds on real-world latency.

\subsection{Algorand's Latency}
Before computing the latency of Algorand using the Pipes model, we first present an over of Algorand in the good case, where the network is synchronous and the processes behave correctly.

\subsubsection*{Overview of Algorand in Good Case}
Algorand proceeds in epochs, and in the good case, it simplifies significantly. 
Let $t$ denote the time at which an epoch~$e$ starts.
It relies on three committees -- \textit{propose committee}, \textit{soft-vote committee}, and \textit{cert-vote committee}.
The protocol has the following phases:
\begin{itemize}
    \item \textbf{Propose.} 
    At time $t$, each process uses a VRF to determine whether it is a member of the propose committee.
    If so, it sends a proposal block to all processes.
    \item \textbf{Soft-vote.} 
    At time $t + 2\Delta$, each process uses a VRF to determine whether it belongs to the soft-vote committee.
    A committee member begins voting only after receiving at least one proposal block. 
    Upon receiving one or more proposals, it deterministically selects one of them and broadcasts a soft vote for the selected block to all processes.
    \item \textbf{Cert-vote.} 
    Once the soft-vote step begins, each process uses a VRF to determine whether it is a member of the cert-vote committee.
    If so, upon hearing from a quorum of soft-voters for some block $B$, it sends a cert-vote for $B$ to all processes.
    \item \textbf{Finalize.}
    When a process hears from a quorum of cert-voters for $B$, it finalizes block $B$.
\end{itemize}

\subsubsection*{Phase Duration}
Let $c_p$, $c_s$, and $c_c$ denote the sizes of the propose, soft-vote, and cert-vote committees, respectively.
In the following, we compute the duration of each phase of the protocol.
\begin{itemize}
    \item \textbf{Propose.}
    Each proposer places $(B+M)n$ addressed parcels into its upload buffer, yielding a sender-side cost of $\frac{(B+M)n}{S}$ per proposer.
    On the receiver side, each process must download one proposal from each of the $c_p$ proposers, for a total of $c_p(B+M)$ parcels; hence, the receiver-side cost is $\frac{c_p(B+M)}{S}$.
    Thus, the duration of the propose phase is
    \begin{align*}
        \Tprop
        &=
        \max\left\{
            \frac{(B+M)n}{S},
            \frac{c_p(B+M)}{S}
        \right\}
        +
        \Delta.
    \end{align*}
    \item \textbf{Soft-vote.}
    Each soft-vote committee member places $\lambda n$ addressed parcels into its upload buffer, yielding a sender-side cost of $\frac{\lambda n}{S}$ per soft-voter.
    On the receiver side, each process must download one soft-vote from each of the $c_s$ soft-voters, for a total of $c_s \lambda$ parcels.
    Hence, the receiver-side cost is $\frac{c_s \lambda}{S}$.
    Therefore, the soft-vote phase duration becomes
    \begin{align*}
        \Tsoft
        &=
        \max\left\{
            \frac{\lambda n}{S},
            \frac{c_s \lambda}{S}
        \right\}
        +
        \Delta.
    \end{align*}
    \item \textbf{Cert-vote.}
    Each cert-vote committee member places $\lambda n$ addressed parcels into its upload buffer, yielding a sender-side cost of $\frac{\lambda n}{S}$ per cert-voter.
    On the receiver side, each process must download one cert-vote from each of the $c_c$ cert-voters, for a total of $c_c \lambda$ parcels.
    Hence, the receiver-side cost is $\frac{c_c \lambda}{S}$.
    Therefore, the cert-vote round duration becomes
    \begin{align*}
        \Tcert
        &=
        \max\left\{
            \frac{\lambda n}{S},
            \frac{c_c \lambda}{S}
        \right\}
        +
        \Delta.
    \end{align*}
\end{itemize}

\subsubsection*{Epoch Duration}
Since each committee is sampled without replacement from the set of all $n$ processes, we have $c_p, c_s, c_c \le n$.
Accordingly, the sender-side upload cost dominates the receiver-side download cost in every phase, so the phase durations simplify to $\Tprop = \frac{(B+M)n}{S} + \Delta$ and $\Tsoft = \Tcert = \frac{\lambda n}{S} + \Delta$.
The total time required to propose, soft-vote, and cert-vote a block of size $B$ is therefore
\begin{align}\label{eq:algorand-cycle}
    T(B)
    &=
    \Tprop
    +
    \Tsoft
    +
    \Tcert
    =
    \frac{(B+M+2\lambda)n}{S}
    +
    3\Delta.
\end{align}

\subsubsection*{Bounded-Latency Constraint}
For latency to remain bounded, the number of transactions arriving during one epoch must not exceed the number of transactions the block contains.
Since transactions arrive at rate~$D$, we require $B \ge D \cdot T(B)$.
Using~\eqref{eq:algorand-cycle}, we have:
\begin{align*}
    B
    &\ge
    D\left(
        \frac{(B+M+2\lambda)n}{S}
        +
        3\Delta
    \right).
\end{align*}
By expanding and rearranging, we obtain:
\begin{align*}
    B\left(1-\frac{Dn}{S}\right)
    &\ge
    \frac{D\bigl((M+2\lambda)n + 3\Delta S\bigr)}{S},
\end{align*}
which gives
\begin{align}\label{eq:algorand-constraint}
    B
    &\ge
    \frac{
        D\bigl((M+2\lambda)n + 3\Delta S\bigr)
    }{
        S-Dn
    }.
\end{align}
This constraint is feasible if and only if $S > Dn$.
Thus, the good-case throughput bottleneck is $D = \frac{S}{n}$.
It is important to note that although there are $c_p$ proposers, only \emph{one} block is finalized per epoch, so the throughput is $B / T(B)$ regardless of $c_p$.
The additional proposals do not slow the protocol either: since $c_p \le n$, the extra download work imposed on receivers is dominated by the proposers' upload time and does not affect the phase duration.

\subsubsection*{Latency as a Function of $B$}
Consider a transaction that arrives just after the proposer has fixed the current block.
In the worst case, the transaction cannot be included in the current block and must wait until the next block is formed.
Since transactions arrive at rate~$D$, filling the next block with $B$ transactions takes $\frac{B}{D}$ time.
Once that next block is formed, it must pass through the propose, soft-vote, and cert-vote phases, contributing $T(B)$ additional time. \kartik{minor: the next block is formed while we are working on the previous one; so it should be max of (B/D, T(B))? and it is ideal when both of them are equal. So, the "2" term will go away below? the same holds for qscale too?}
Therefore,
\begin{align}\label{eq:algorand-latency-B}
    L(B)
    &=
    \frac{B}{D}
    +
    T(B)
    =
    \frac{B}{D}
    +
    \frac{(B+M+2\lambda)n}{S}
    +
    3\Delta.
\end{align}

\subsubsection*{Minimizing Latency}
Since latency is an increasing function of the block size $B$, it is minimized by choosing the smallest block size satisfying~\eqref{eq:algorand-constraint}:
\begin{align}\label{eq:algorand-B-star}
    B^{*}
    &=
    \frac{
        D\bigl((M+2\lambda)n + 3\Delta S\bigr)
    }{
        S-Dn
    }.
\end{align}
At $B = B^{*}$, the bounded-latency constraint holds with equality, so $\frac{B^{*}}{D} = T(B^{*})$ and hence
\begin{align*}
    L(B^{*})
    &=
    \frac{B^{*}}{D}
    +
    T(B^{*})
    =
    2T(B^{*}).
\end{align*}
By substituting $T(B^{*}) = \frac{(M+2\lambda)n + 3\Delta S}{S-Dn}$, we obtain
\begin{align}\label{eq:algorand-latency}
    L(B^{*})
    &=
    \frac{
        2\bigl((M+2\lambda)n + 3\Delta S\bigr)
    }{
        S-Dn
    } \nonumber
    \\&=
    \left(
        \frac{2(M+2\lambda)n}{S}
        +
        6\Delta
    \right)
    \left(
        1
        +
        \frac{1}{(S/(Dn))-1}
    \right).
\end{align}

\subsubsection*{Discussion}
We consider moderate-scale systems, i.e., systems with $n$ ranging from roughly one hundred to a few thousand processes.
In this regime, for realistic values of $S$, the term $\frac{2(M+2\lambda)n}{S}$ in~\eqref{eq:algorand-latency} can be comparable to or even exceed $6\Delta$, so both terms contribute meaningfully to the additive term of the latency.
If $S/(Dn)$ is large, then the right term in parentheses is approximately $1$.
However, as $S/(Dn)$ approaches $1$, the right term (and latency) tends to infinity, meaning that a latency bottleneck occurs at $D = S/n$.
This bottleneck arises from the propose phase: each proposer must send the full block to all $n$ processes, so a single proposer's upload cost $(B+M)n/S$ determines the phase duration.
The soft-vote and cert-vote phases add only fixed vote traffic and two extra network delays; they contribute to the additive term but do not affect the bottleneck, since their cost does not grow with the block size $B$.

The number of proposers $c_p$ does not appear in the final latency expression.
This is because $c_p \le n$ ensures that the sender-side upload dominates the receiver-side download in every phase.
The extra $c_p(B+M)/S$ download cost imposed on each receiver during the propose phase is therefore masked by the per-proposer upload cost $(B+M)n/S$, and the phase duration is the same as if there were a single proposer.
Consequently, the additional proposals introduced by a larger propose committee neither increase the rate at which blocks are finalized nor increase latency.

\subsection{\protocolName's Latency}
We now apply the Pipes model to \protocolName in the good case, where the network is synchronous, and the leader is correct. 
Recall that in each epoch~$e$ of \protocolName, the leader's proposal block is disseminated in two steps: first, the leader sends the proposal to a random sample $S_e$, and in the second step, each member of $S_e$ forwards the proposal to another random sample. 
Each random sample is obtained by including every process independently with probability $\prSample = c/\sqrt{n}$, where $c$ is a constant. 
Hence, each sample has expected size $c\sqrt n$. 
We write $s = c\sqrt n$ for this expected size, and our analysis uses this expected size.

After proposal dissemination, processes that receive the proposal become candidates for the proposal.
Then, $\alpha$ voting samples $S_1,\dots,S_\alpha$ are selected from the candidate processes, where $\alpha$ is a constant, and each sample is selected like the random samples selected during the proposal dissemination.
In addition, we have the first sample $S_{e+1}$ of the next epoch.
Every voter that belongs to $S_1$ sends its vote to the leader of epoch $e+1$; moreover, each voter belonging to at least one $S_i$, sends its vote to all members of $S_{e+1}$.
If the leader of epoch $e+1$ receives a probabilistic quorum of $q$ votes from members of $S_1$, it proposes a new block; besides, a process $i \in S_{e+1}$ that receives a probabilistic quorum of $q$ votes from each of the $\alpha$ samples sends a signature share to all members of $S_{e+1}$; a process $i \in S_{e+1}$ that collects $q$ signature shares from members of $S_{e+1}$ aggregates them into a quorum certificate.
Once this certificate exists and a member of $S_{e+1}$ receives the $\leader{e+1}$'s proposal, it can disseminate the proposal further.

\subsubsection*{Round Duration.}
We first compute the latency imposed by each round of the protocol.
\begin{itemize}
    \item \textbf{Round 1.}
    The leader sends the proposal of size $B+M$ to the $s$ members of $S_e$, placing $(B+M)s$ addressed parcels into its upload buffer.
    Hence,
    \begin{align*}
        T_1 = \frac{(B+M)s}{S} + \Delta.
    \end{align*}
    \item \textbf{Round 2.}
    Each member of $S_e$ forwards the proposal to a random sample of expected size $s$, uploading $(B+M)s$ addressed parcels.
    Therefore,
    \begin{align*}
        T_2 = \frac{(B+M)s}{S} + \Delta.
    \end{align*}
    Collectively, the $s$ members of $S_e$ add $s \cdot (B+M)s = c^2 n(B+M)$ addressed parcels to their upload buffers, but since these parcels are spread evenly across the $s$ senders, the per-sender upload is only $(B+M)s/S$.
    \item \textbf{Round 3.}
    Each voter sends its vote of size $\lambda$ to all $s$ members of $S_{e+1}$, and voters in $S_1$ additionally send their vote to the leader of epoch $e+1$.
    Hence, the worst-case per-voter upload is incurred by voters in $S_1$, who upload $\lambda(s+1)$ addressed parcels, giving a sender-side cost of $\lambda(s+1)/S$.
    On the download side, the leader of epoch $e+1$ needs $q = 2s/3$ votes from $S_1$, costing $2s\lambda/(3S)$.
    Each member of $S_{e+1}$ must receive a probabilistic quorum of $q$ votes from each of the $\alpha$ samples, downloading a total of $\alpha q = 2\alpha s/3$ parcels at a cost of $2\alpha s\lambda/(3S)$.
    For $\alpha \ge 2$, the bottleneck is the download at $S_{e+1}$, since $2\alpha s\lambda/(3S) \ge \lambda(s+1)/S$ and $2\alpha s\lambda/(3S) \ge 2s\lambda/(3S)$.
    Therefore,
    \begin{align*}
        T_3 = \frac{2\alpha s\lambda}{3S} + \Delta.
    \end{align*}
    \item \textbf{Round 4.}
    Each member of $S_{e+1}$ sends a signature share of size $\lambda$ to all $s$ members of $S_{e+1}$, uploading $\lambda s$ addressed parcels.
    On the download side, each member needs only $q = 2s/3$ shares to form the quorum certificate, costing $2s\lambda/(3S)$.
    Since $\lambda s \ge 2s\lambda/3$, the sender-side upload dominates the receiver-side download, so
    \begin{align*}
        T_4 = \frac{\lambda s}{S} + \Delta.
    \end{align*}
\end{itemize}

\subsubsection*{Pipelining of Round 4 with Round 1 of the Next Epoch}
Note that round 4 of epoch $e$ and round 1 of epoch $e+1$ use disjoint senders, and both start as soon as round 3 of epoch $e$ ends, so they run in parallel and share the same network delay~$\Delta$.
During this interval, each member of $S_{e+1}$ receives parcels from two sources: the leader of epoch $e+1$ sends the new proposal of size $B+M$, while the other members of $S_{e+1}$ send signature shares of size $\lambda$.
We must therefore ensure that adding the signature-share download on top of the proposal download does not make the receiver's download buffer the new bottleneck.
Specifically, the three per-process costs during this interval are:
\begin{itemize}
    \item Leader of $e+1$ upload: $(B+M)s/S$, since the leader places $(B+M)s$ addressed parcels into its upload buffer.
    \item Member of $S_{e+1}$ upload: $\lambda s/S$, since each member places $\lambda s$ addressed signature-share parcels into its upload buffer.
    \item Member of $S_{e+1}$ download: $((B+M) + q\lambda)/S = ((B+M) + 2s\lambda/3)/S$, where the $B+M$ parcels come from the leader and the $q = 2s/3$ parcels are the signature shares needed to form the quorum certificate.
\end{itemize}
For any reasonable block size, $B + M \ge \lambda$ and $(B+M)(s-1) \ge 2s\lambda/3$, so the leader's upload dominates: $(B+M)s \ge \lambda s$ and $(B+M)s \ge (B+M) + 2s\lambda/3$.
Thus, this interval requires $(B+M)s/S + \Delta$ time, which is exactly $T_1$ of the next epoch.
Round 4, therefore, adds no additional time to the epoch duration and is absorbed entirely into round 1 of the next epoch.

\subsubsection*{Epoch Duration}
The epoch duration is the interval between the certification of a block and its child block.
Since round 4 of epoch $e$ is absorbed into round 1 of epoch $e+1$, the epoch duration is
\begin{align}\label{eq:qscale-cycle}
    \Tep(B) \nonumber
    &= T_1 + T_2 + T_3 \\
    &= \frac{(B+M)s}{S} + \frac{(B+M)s}{S} + \frac{2\alpha s\lambda}{3S} + 3\Delta \nonumber \\
    &= \frac{6(B+M)s + 2\alpha s\lambda}{3S} + 3\Delta.
\end{align}

\subsubsection*{Bounded-Latency Constraint}
\protocolName certifies one block per epoch, so for latency to remain bounded, transactions must arrive no faster than blocks are certified, i.e., $B \ge D \cdot \Tep(B)$.
Using~\eqref{eq:qscale-cycle}, we have:
\begin{align*}
    B \ge \frac{D\bigl(6Bs + 6Ms + 2\alpha s\lambda + 9\Delta S\bigr)}{3S}.
\end{align*}
Moving the $B$ term to the left,
\begin{align*}
    B\!\left(1 - \frac{2Ds}{S}\right) \ge \frac{D\bigl(6Ms + 2\alpha s\lambda + 9\Delta S\bigr)}{3S}.
\end{align*}
Hence,
\begin{align}\label{eq:qscale-constraint}
    B \ge \frac{D\bigl(6Ms + 2\alpha s\lambda + 9\Delta S\bigr)}{3(S - 2Ds)}.
\end{align}
The constraint is feasible if and only if $S > 2Ds$.
Since $s = c\sqrt n$, the good-case latency bottleneck is $D = \frac{S}{2c\sqrt n}$.
This bottleneck is unchanged by indirect certification: the additional round-3 and round-4 communication costs are independent of $B$ and appear only in the numerator.

\subsubsection*{Finalization Latency as a Function of $B$}
Recall that in \protocolName, a client finalizes a block based on the pattern of certified blocks that appear after it.
In the special case where the certified descendants appear in consecutive epochs, a block $B_1$ is finalized when there exists a sequence of certified blocks
$B_1, B_2, \ldots, B_\kappa$
that form a chain, are proposed in consecutive epochs, and have consecutive heights.
Consider a transaction arriving just after the leader has fixed the current block.
In the worst case, the transaction is not included in the current block and must wait until block $B_1$ is fully formed, which takes $B/D$.
After that, $B_1$ proceeds through round 1, round 2, and round 3 of its own epoch and round 4 in parallel with round 1 of the next epoch, so $B_1$'s certificate is created at time $B/D + \Tep(B) + T_4$.
Each subsequent block $B_2, \dots, B_\kappa$ is certified one epoch later than the previous one.
Therefore, block $B_\kappa$ is certified, and $B_1$ is finalized, at time
\begin{align}\label{eq:qscale-latency-B}
    L(B) = \frac{B}{D} + \kappa\,\Tep(B) + T_4.
\end{align}

\subsubsection*{Minimizing Latency}
The latency $L(B)$ is increasing in $B$, so it is minimised by the smallest block size satisfying~\eqref{eq:qscale-constraint}:
\begin{align*}
    B^{*} = \frac{D\bigl(6Ms + 2\alpha s\lambda + 9\Delta S\bigr)}{3(S - 2Ds)}.
\end{align*}
At $B = B^{*}$, the bounded-latency constraint holds with equality, so $B^{*}/D = \Tep(B^{*})$, and therefore
\begin{align*}
    L(B^{*}) = \Tep(B^{*}) + \kappa\,\Tep(B^{*}) + T_4 = (\kappa + 1)\,\Tep(B^{*}) + T_4.
\end{align*}
Substituting $\Tep(B^{*}) = B^{*}/D$ and $T_4 = \lambda s/S + \Delta$,
\begin{align*}
    L(B^{*}) 
    &= \frac{(\kappa + 1)\bigl(6Ms + 2\alpha s\lambda + 9\Delta S\bigr)}{3(S - 2Ds)} + \frac{\lambda s}{S} + \Delta
    \\&=
    L(B^{*}) = (\kappa + 1)\!\left(\frac{2Ms + \tfrac{2\alpha}{3}s\lambda}{S} + 3\Delta\right)\!\left(1 + \frac{2}{(S/Ds) - 2}\right) + \frac{\lambda s}{S} + \Delta.
\end{align*}
Substituting $s = c\sqrt n$,
\begin{align*}
    L(B^{*}) = (\kappa + 1)\!\left(\frac{2Mc\sqrt n + \tfrac{2\alpha c\sqrt n}{3}\lambda}{S} + 3\Delta\right)\!\left(1 + \frac{2}{(S/(Dc\sqrt n)) - 2}\right) + \frac{\lambda c\sqrt n}{S} + \Delta.
\end{align*}

\subsection{Latency Comparison of Algorand and \protocolName}

\begin{figure}[!t]
    \centering
    \includegraphics[scale=0.61]{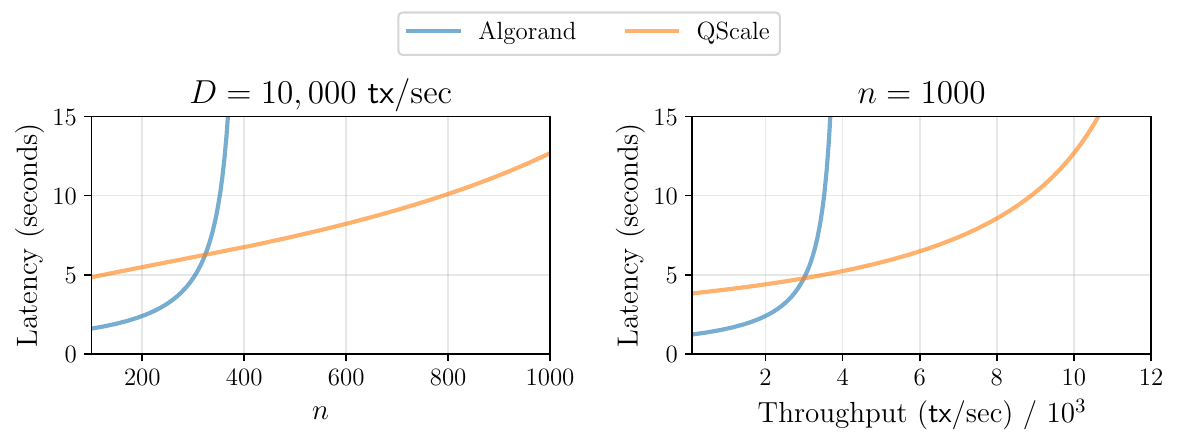}
    \caption{Latency of Algorand and \protocolName under the Pipes model. Left: latency as a function of $n$ at $D = 10{,}000$ \txs/\second. Right: latency as a function of throughput at $n = 1000$.}
    \label{fig:algorand-vs-qscale}
\end{figure}

We compare the latency of Algorand and \protocolName under the Pipes model in the good case, using the default parameters of~\cite{lewispye2026pipes}: bandwidth $S = 10$ Gbps, transaction size $\tau = 2500$ bits, block metadata size $M = 1000$ bits, vote size $\lambda = 500$ bits, and network delay $\Delta = 0.2\,\second$.
In this comparison, we set the system size to $n=1000$.
Since Algorand targets $\epsilon = 0.2$ with $\approx\!56$ bits of security, we consider a \protocolName' configuration with $q=3\sqrt{n}=94$, $\kappa=5$, and $\alpha = 3$.

Fig.~\ref{fig:algorand-vs-qscale} shows the results of this analysis for a fixed workload (left) and system size (right).
The results show \protocolName is more scalable than Algorand.
This happens because Algorand employs a committee-to-all dissemination strategy, which requires a proposer to disseminate a full block to all processes; as a result, latency increases rapidly as either the number of processes or the workload (transaction arrival rate) grows.
In contrast, \protocolName disseminates blocks only to samples of expected size $O(\sqrt n)$. 
Consequently, for high-throughput deployments with hundreds to a few thousand processes, \protocolName maintains bounded latency over a wider range of transaction rates and is expected to achieve a substantially higher maximum throughput than Algorand.